\documentclass[a4paper,USenglish,cleveref, autoref,pagebackref,
thm-restate,authorcolumns]{lipics-v2019}

  \nolinenumbers
  \hideLIPIcs

\usepackage[T1]{fontenc}
\usepackage{amsthm}
\usepackage{amsmath}
\usepackage[disable]{todonotes}
\usepackage{amssymb}
\usepackage{subcaption}
\usepackage{amsfonts}
\usepackage{bm}
\usepackage{nicefrac}
\usepackage{esint}
\usepackage{dsfont}
\usepackage{subcaption}
\usepackage[font=small,labelfont=bf]{caption}
\usepackage{thmtools}
\usepackage{thm-restate}
\usepackage[algo2e,ruled,vlined,linesnumbered,algosection,hanginginout]{algorithm2e}
\SetEndCharOfAlgoLine{}
\usepackage{mathtools}

\usepackage{booktabs}
\usepackage{tabularx}
\usepackage{multirow}
\usepackage{makecell}
\usepackage{dsfont}

\usepackage[inline,shortlabels]{enumitem}
\setlist{nosep}
\usepackage{tikz}
\usetikzlibrary{arrows}
\usetikzlibrary{shapes}
\usetikzlibrary{decorations.pathreplacing}
\usetikzlibrary{shapes.geometric}
\usetikzlibrary{patterns}
\usetikzlibrary{fadings}
\usepackage{bm}
\usepackage{etoolbox}
\newcommand{\tworows}[2]{\begin{tabular}{c}{#1}\\{#2}\end{tabular}}
\tikzstyle{para}=[rectangle,draw=black,minimum height=.2cm,fill=gray!10,rounded corners=1mm,font=\footnotesize]

\definecolor{ourblue}{RGB}{130,200,250}
\definecolor{ourdarkblue}{RGB}{30,60,140}
\definecolor{ourgreen}{RGB}{0,100,0}
\definecolor{ourred}{RGB}{180,20,30}
\definecolor{lipicsyellow}{RGB}{220,200,0}
\definecolor{ourorange}{RGB}{230,100,0}
\definecolor{ourgray}{RGB}{100,100,100}

\newtheorem{observation}[theorem]{Observation}
\newtheorem{rrule}{Reduction Rule}[section]

\theoremstyle{definition}%
\newtheorem{constr}{Construction}

\crefname{rrule}{Reduction Rule}{Reduction Rules}
\crefname{observation}{Observation}{Observations}
\crefname{lemma}{Lemma}{Lemmata}
\crefname{figure}{Figure}{Figures}
\Crefname{observation}{Obs}{Obs}

\makeatletter
\renewcommand\subparagraph{%
 \@startsection {subparagraph}{5}{\z@ }{3.25ex \@plus 1ex
 \@minus .2ex}{-1em}{\normalfont \normalsize \bfseries }}%
\makeatother

\newcommand{\proofpar}[1]{\smallskip\noindent\emph{#1}}

\newcommand{\wx}[1]{\text{W[#1]}\xspace}
\newcommand{\wone}{\wx{1}}

\newcommand{\NP}{\text{NP}\xspace}
\newcommand{\classNP}{\text{NP}\xspace}
\newcommand{\classcoNP}{\text{coNP}\xspace}

\newcommand{\polyadvice}{{\text{poly}}\xspace}
\newcommand{\FPT}{\text{FPT}\xspace}
\newcommand{\NoKernelAssume}{\classNP $\subseteq$ \classcoNP/\polyadvice}
\newcommand{\NNoKernelAssume}{\classNP $\not\subseteq$ \classcoNP/\polyadvice}

\newcommand{\truevalue}{\textrm{true}}
\newcommand{\falsevalue}{\textrm{false}}
\newcommand{\fes}{\fint}

\newcommand{\probDef}[3]{
       \vspace{-4pt}
\begin{center}   
    \fbox{~\begin{minipage}{.95\textwidth}

      \noindent
      \normalsize\textsc{#1}
      
      \vspace{2pt}
      \setlength{\tabcolsep}{3pt}
      \renewcommand{\arraystretch}{1.0}
      \begin{tabularx}{\textwidth}{@{}lX@{}}
	\normalsize\emph{Input:} 	& \normalsize#2 \\
	\normalsize\emph{Question:} 	& \normalsize#3
      \end{tabularx}
    \end{minipage}}
    \end{center}
       \vspace{-4pt}
}

\newcommand{\probExactThreeFourSAT}{\textsc{Exact $(3,4)$-SAT}\xspace}
\newcommand{\probThreeSAT}{\textsc{3-SAT}\xspace}
\newcommand{\probMCClique}{\textsc{Multicolored Clique}\xspace}

\newcommand{\lifetime}{\ensuremath{\ell}}

\newcommand{\TG}{\ensuremath{\mathcal{G}}\xspace}
\newcommand{\TGfull}{\ensuremath{\mathcal{G}=(V, E_1, \ldots, E_\lifetime)}\xspace}
\newcommand{\TGcompact}{\ensuremath{\mathcal{G}=(V, (E_i)_{i\in[\lifetime]})}\xspace}

\newcommand{\UG}{\ensuremath{G_{\downarrow}}\xspace}

\newcommand{\ug}[1]{\ensuremath{G_\downarrow (#1)}\xspace}

\newcommand{\layer}{layer\xspace}

\newcommand{\timestep}{time step\xspace}

\newcommand{\RArrow}{\smallskip$(\Rightarrow)$:\xspace}
\newcommand{\LArrow}{\smallskip$(\Leftarrow)$:\xspace}

\newcommand{\nonstrpath}[1]{temporal~$(#1)$-path}
\newcommand{\nonstrpaths}[1]{temporal~$(#1)$-paths}

\newcommand{\npath}[1]{$(#1)$-path}

\newcommand{\N}{\mathbb{N}}
\newcommand{\I}{\mathcal{I}}

\newcommand{\TPP}{\probSRestlessPath}
\newcommand{\IPP}{\textsc{Independent Path}}
\newcommand{\yes}{\emph{yes}}
\newcommand{\no}{\emph{no}}
\newcommand{\tpath}[1][s,z]{$\wait$-restless temporal $(#1)$-path}
\newcommand{\twalk}[1][s,z]{$\wait$-restless temporal $(#1)$-walk}
\newcommand{\dipath}[1][s,z]{$(#1)$-dipath}

\newcommand{\probRestlessPath}{\textsc{Restless Temporal Path}\xspace}
\newcommand{\probSRestlessPath}{\textsc{Short Restless Temporal Path}\xspace}

\newcommand{\wait}{\ensuremath{\Delta}}

\newcommand{\ownparagraph}[1]{\smallskip\noindent\textbf{#1}\xspace}

\newcommand{\probDefRestlessPath}{
\probDef{\probRestlessPath} 
{A temporal graph \TGcompact, two distinct vertices $s,z\in V$, and an integer $\wait \leq \lifetime$.}
{Is there a $\Delta$-restless temporal $(s,z)$-path in~$\TG$?}
}

\tikzstyle{square}=[regular polygon,regular polygon sides=4]
\tikzstyle{vert2}=[circle,inner sep=1.5,fill=white,draw,minimum size=.2cm]
\tikzstyle{vert3}=[inner sep=1.5,fill=white,draw=black,minimum size=.2cm]

\tikzstyle{vertexset}=[circle,draw,thick,inner sep=3pt,fill=white]
\tikzstyle{evertex}=[circle,draw,fill=black, inner sep= 2pt]
\tikzstyle{evertex2}=[square,draw,fill=black, inner sep= 2pt]

\tikzstyle{edge}=[thick]

\pgfmathsetmacro\sprayRadius{.2pt}
\pgfmathsetmacro\sprayPeriod{.5cm}

\pgfdeclarepatternformonly{spray}{\pgfpoint{-\sprayRadius}{-\sprayRadius}}{\pgfpoint{1cm + \sprayRadius}{1cm + \sprayRadius}}{\pgfpoint{\sprayPeriod}{\sprayPeriod}}{
    \foreach \x/\y in {2/53,6/52,11/48,23/49,20/47,32/46,41/47,47/51,56/52,46/44,4/43,16/42,33/41,41/37,49/35,55/31,37/35,44/30,28/37,24/36,17/37,7/38,0/31,8/29,18/31,28/30,37/28,30/27,46/24,51/21,24/23,12/24,4/21,18/19,12/16,31/21,38/18,26/16,46/16,56/12,52/10,45/8,51/4,37/12,35/7,24/9,14/9,2/12,8/6,15/4,27/0,34/1,40/1} {
        \pgfpathcircle{\pgfpoint{(\x + random()) / 57 * \sprayPeriod}{\sprayPeriod - (\y + random()) / 55 * \sprayPeriod}}{\sprayRadius}
    }
    \pgfusepath{fill}
}
\newcommand{\theslope}{0.7}
\pgfdeclarepatternformonly{lines}
{\pgfpoint{-.1mm/\theslope}{-.1mm}}{\pgfpoint{1.1mm/\theslope}{1.1mm}}
{\pgfpoint{1mm/\theslope}{1mm}}
{
    \pgfsetlinewidth{0.4pt}
    \pgfpathmoveto{\pgfpoint{-.1mm/\theslope}{-.1mm}}
    \pgfpathlineto{\pgfpoint{1.1mm/\theslope}{1.1mm}}
    \pgfusepath{stroke}
}

\title{The Computational Complexity of Finding Temporal Paths under Waiting Time Constraints}

\titlerunning{Finding Temporal Paths under Waiting Time Constraints}%
\newcommand{\tuaddress}{Technische Universit\"at Berlin, Algorithmics and Computational Complexity, Berlin, Germany}
\author{Arnaud Casteigts}{LaBRI, Universit\'e de Bordeaux, CNRS, Bordeaux INP, France}{arnaud.casteigts@labri.fr}{https://orcid.org/0000-0002-7819-7013}{Supported by the ANR, project ESTATE (ANR-16-CE25-0009-03).}
\author{Anne-Sophie Himmel}{\tuaddress}{anne-sophie.himmel@tu-berlin.de}{https://orcid.org/0000-0001-7905-7904}{Supported by the DFG, project FPTinP (NI 369/16).}
\author{Hendrik Molter}{\tuaddress}{h.molter@tu-berlin.de}{https://orcid.org/0000-0002-4590-798X}{Supported by the DFG, project MATE (NI 369/17).}

\author{Philipp~Zschoche}{\tuaddress}{zschoche@tu-berlin.de}{https://orcid.org/0000-0001-9846-0600}{}

\authorrunning{A. Casteigts, A.-S. Himmel, H. Molter, P. Zschoche}%

\Copyright{Arnaud Casteigts, Anne-Sophie Himmel, Hendrik Molter, Philipp Zschoche}%



\category{}%

\supplement{}%

\EventEditors{}
\EventNoEds{0}
\EventLongTitle{}
\EventShortTitle{}
\EventAcronym{}
\EventYear{}
\EventDate{}
\EventLocation{}
\EventLogo{}
\SeriesVolume{}
\ArticleNo{}

\begin{document}

\maketitle

\begin{abstract}
  Computing a (short) path between two vertices is one of the
  most fundamental primitives in graph algorithmics. 
  In recent years, the study
  of paths in temporal graphs, that is, graphs where the vertex set is fixed
  but the edge set changes over time, gained more and more attention. 
  A path is time-respecting, or \emph{temporal}, if it uses edges with 
  non-decreasing time stamps.

  We investigate a basic constraint for temporal paths, where
  the time spent at each vertex must not exceed a given duration~$\Delta$,
  referred to as $\Delta$-\emph{restless temporal paths}. 
  This constraint arises naturally in the modeling of real-world processes like 
  packet routing in communication networks and
  infection transmission routes 
  of diseases where recovery confers lasting resistance.
  
  While finding temporal paths without waiting time restrictions is known to be
  doable in polynomial time, we show that the ``restless variant'' of
  this problem becomes computationally hard even in very restrictive settings.
  For example, it is \wone-hard when parameterized by the \emph{distance to disjoint path} of the underlying graph, which implies \wone-hardness for many other parameters like feedback vertex number and pathwidth.
  A natural question is thus whether the problem becomes tractable in some
  natural settings. 
  We explore several natural parameterizations, presenting
  \FPT algorithms for three \emph{kinds} of parameters: (1) output-related
  parameters (here, the maximum length of the path), (2) classical parameters
  applied to the underlying graph (e.g., feedback edge number), and (3) a new
  parameter called \emph{timed feedback vertex number}, which captures
  finer-grained temporal features of the input temporal graph, and which may be
  of interest beyond this work.
\end{abstract}

\section{Introduction}

A highly successful strategy to control (or eliminate) outbreaks of 
infectious diseases is contact tracing \cite{eames2003contact}---whenever an individual is diagnosed positively, 
every person who is possibly infected by this individual is put into quarantine.
However, the viral spread can be too fast to be traced manually, 
e.g., if the disease is transmittable in a pre-symptomatic (or asymptomatic) stage, 
then it seems likely that an individual already caused infection chains when
diagnosed positively. 
Hence, large-scale digital systems are recommended
which use physical proximity networks based on location and contact
data~\cite{ferretti2020quantifying}---this allows fast and precise contact tracing while
avoiding the harmful effect of mass quarantines to society~\cite{ferretti2020quantifying}.
Physical proximity networks can be understood as temporal graphs\footnote{Also known as
time-varying graphs, evolving graphs, or link streams.}
\cite{Cas+12,Hol15,HS19,LVM18,Mic16}, that is, graphs where the vertex set (individuals) remains static but the edge set (physical contacts) may change over time.
In this paper, we extend the literature on reachability in 
temporal graphs %
\cite{Akr+19b, %
AF16, %
BF03, %
XFJ03, %
Him+19, %
KKK02, %
MOS19, %
Wu+16} %
by a computational complexity analysis of an important variation of one of the
most fundamental combinational problems arising in the above mentioned scenario: 
given a temporal graph and two individuals~$s$ and~$z$, is a chain of infection from $s$
to $z$ possible, that is, is there a \emph{temporal path} from $s$
to $z$? %
In particular, we use a reachability concept that captures the standard 3-state
\emph{SIR-model} (Susceptible-Infected-Recovered), a canonical spreading model 
for diseases where recovery confers lasting resistance~\cite{Bar16,kermack1927contribution,New18}.

In temporal graphs, the basic concepts of paths and reachability are 
defined in a time-respecting way~\cite{KKK02}: a (strict) temporal path, also called
``journey'', is a path that uses edges with non-decreasing (increasing)
time steps.
To represent infection chains in the SIR-model, 
we restrict the time of ``waiting'' or ``pausing'' at each intermediate vertex
to a prescribed duration. We call these paths \emph{restless} temporal paths. 
They model infection transmission routes of diseases that grant immunity upon recovery~\cite{Hol16}: 
An infected individual can transmit the disease until it is recovered (reflected by bounded waiting time) and it cannot be
infected a second time afterwards since then it is immune (reflected by considering path instead of walk: 
every vertex can only be visited at most once).
Another natural example of restless temporal paths is delay-tolerant networking among mobile
entities, where the routing of a packet is performed over time and space by
storing the packet for a limited time at intermediate nodes.

In the following we give an example to informally describe our problem setting\footnote{We refer to
\cref{sec:path:prelims} for a formal definition.}.
In \cref{fig:simple-example}
we are given the depicted temporal graph, vertices $s$ and $z$, and the time bound $\Delta=2$.
We are asked to decide whether there is a restless temporal path from $s$ to $z$, that is, a path which visits each vertex at most
once and pauses at most $\Delta$ units of time between consecutive hops.
Here, $(s,d,b,z)$ is a feasible solution, but $(s,b,z)$ is not because the waiting time at $b$ exceeds $\Delta$.
The walk $(s,b,c,d,b,z)$ is not a valid solution because it visits vertex~$b$ twice. 
Finally, $(s,a,c,d,b,z)$ is also a feasible solution.

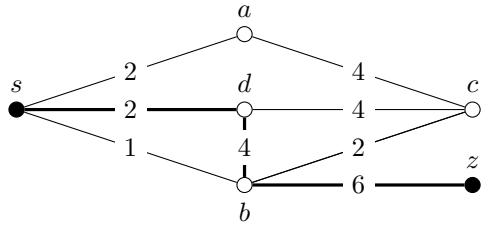
\begin{figure}[t]
  \centering
\begin{tikzpicture}[scale=1.0,]
		\node[vert2,fill=black,label=$s$] (s) at (0,0) {};
		\node[vert2,label=$a$] (a) at (3,1) {};
		\node[vert2,label={below:$b$}] (b) at (3,-1) {};
		\node[vert2,label=$c$] (c) at (6,0) {};
		\node[vert2,label=$d$] (d) at (3,0) {};
		\node[vert2,fill=black,label=$z$] (z) at (6,-1) {};

		\draw (s) -- node[fill=white] {$2$} (a);
		\draw (s) -- node[fill=white] {$1$} (b);
		\draw (a) -- node[fill=white] {$4$} (c);
		\draw (b) -- node[fill=white] {$2$} (c);
		\draw (b) -- node[fill=white] {$2$} (c);
		\draw (c) -- node[fill=white] {$4$} (d);
		\draw[very thick] (d) to node[fill=white] {$4$} (b);
		\draw[very thick] (b) to node[fill=white] {$6$} (z);
		\draw[very thick] (s) to node[fill=white] {$2$} (d);

\end{tikzpicture}
\caption{Example of a temporal graph whose edges are labeled with time stamps.
Bold edges depict a $2$-restless temporal $(s,z)$-path.
(In general, multiple time stamps per edge are possible.)
}
\label{fig:simple-example}
\vspace{-2ex}
\end{figure}

\ownparagraph{Related work.}
Several types of waiting time constraints have been considered in the temporal
graph literature. %
An empirical study by Pan and
Saram\"aki~\cite{PS11} based on phone calls datasets observed a threshold in
the correlation between the duration of pauses between calls and the ratio of
the network reached over a spreading process. 
Casteigts~et~al.~\cite{Cas+15a} showed a dramatic impact of waiting time constraints
to the expressivity of a temporal graph, 
when considering such a graph as an automaton and temporal paths as words.
In the context of temporal flows, Akrida~et~al.~\cite{Akr+19a} considered a
concept of ``vertex buffers'', which however pertains to the quantity of
information that a vertex can store, rather than a duration.
Enright~et~al.~\cite{Enr+19} considered deletion problems for reducing temporal
connectivity. %
More closely related to our work, Himmel~et~al.~\cite{Him+19}
studied a variant of restless temporal paths where several
visits to the same vertex are allowed, i.e., restless temporal
\emph{walks}. 
They showed, among other things, that such walks can be 
computed in polynomial time. %

Many path-related problems have been studied in the temporal setting and the
nature of temporal paths significantly increases the computational complexity of many of them (compared to their
static counterparts). In the temporal setting, reachability is not an
equivalence relation among vertices, which makes many problems more complicated.
For example, finding a maximum temporally connected component is
NP-hard~\cite{BF03}. %
We further have that in a temporal graph,
spanning trees may not exist.
In fact, even the existence of \emph{sparse} spanners (i.e., subgraphs with
$o(n^2)$-many edges ensuring temporal connectivity) is not
guaranteed~\cite{AF16}, unless the underlying graph is complete~\cite{CPS19},
and computing a minimum-cardinality spanner is APX-hard~\cite{Akr+17,MOS19}.
Yet another example is the problem of deciding
whether there are $k$ disjoint temporal paths between two given vertices. 
In a seminal
article, Kempe et al.~\cite{KKK02} showed that this problem,
whose classical analogue is (again) polynomial-time solvable, becomes NP-hard.
They further investigated the related problem of finding temporal separators,
which is also NP-hard~\cite{Flu+19a,KKK02,Zsc+19}. Deciding whether there exists a separator of a given size that cuts all \emph{restless} temporal paths is known to be $\Sigma_2^P$-complete~\cite{Molter20}, that is, the problem is located in the second level of the polynomial time hierarchy.
\ownparagraph{Our contributions.}
We introduce the problem 
\probRestlessPath. %
To get a finer understanding of the computational complexity of this problem, 
we turn our attention to its parametrized complexity. 
In stark contrast to both restless temporal walks and non-restless temporal paths, 
we show that this problem is \NP-hard even in very restricted settings---in particular, even when the lifetime is restricted to only three time steps---and \wone-hard when parameterized by the
(vertex deletion) \emph{distance to disjoint paths} of the underlying
graph, which implies \wone-hardness with respect to many other parameters like feedback vertex number and pathwidth (\cref{sec:path:hardness}).
This is tight in the sense that the problem can be solved in polynomial time when the underlying graph is a forest.
On the positive side, we explore parameters of three different natures. First,
we show that the problem is fixed-parameter tractable (\FPT) for the
\emph{length} (in number of hops) of the temporal path (\cref{sec:fptalg}).
We further show that the problem is \FPT when parameterized by the \emph{feedback edge number} of the underlying graph (\cref{sec:structparam}). 
Additionally, we show that the problem presumably does not admit a polynomial kernel
under the previously mentioned parameterizations where the problem is in \FPT.\todo{pz: gramma ok?}
Our results provide a fine-grained characterization of the tractability boundary of the
computation of restless temporal paths for parameters of the underlying graph,
as illustrated by the vicinity of the corresponding parameters in
\cref{fig:hierarchy}. %
Then, going beyond parameters related to the output and to the underlying graph,
we define a novel temporal version of the classic \emph{feedback vertex
number} called \emph{timed feedback vertex number}. Intuitively, it counts the
number of vertex \emph{appearances} that have to be removed from the temporal
graph such that its underlying graph becomes cycle-free.
We show that finding restless temporal paths is FPT when parameterized by 
this parameter (\cref{sec:tfvn}).
We believe that the latter is an interesting turn of events compared to our hardness results.

\ownparagraph{Strict versus non-strict temporal paths.} 
In this paper, we focus mainly on the case of non-strict temporal paths, i.e., the times along a path are required to be non-decreasing. We expect most of the algorithms and reductions to be extendable to a strict setting, albeit with some change in the results themselves. For instance, a similar NP-hardness reduction as for non-strict temporal paths may apply, but requires more than a constant lifetime to be adapted. In fact, the length of a strict temporal path is trivially bounded by the lifetime itself, thus an FPT algorithm for the \emph{length} parameter implies one for the lifetime parameter as well.

\tikzfading[name=fade bottom,
  top color=transparent!0,
  bottom color=transparent!100]
\tikzfading[name=fade top,
  top color=transparent!100,
  bottom color=transparent!0]
\tikzfading[name=fade left,
  right color=transparent!100,
  left color=transparent!0]
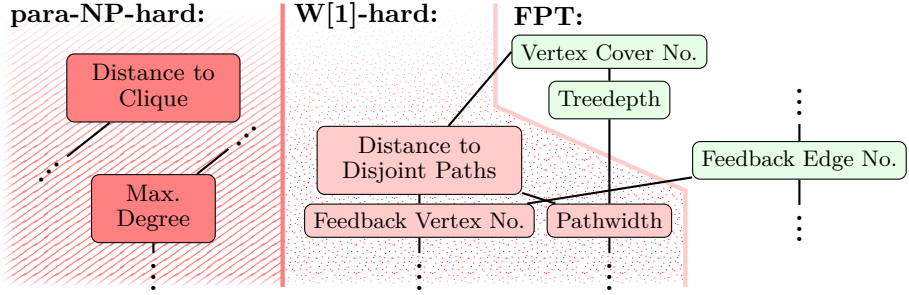
\begin{figure}[t]
		\centering
		\begin{tikzpicture}[yscale=0.8]
		\filldraw[pattern=spray,pattern color=red,draw=white] (0.7,0.6) -- (0.7,-4.1) -- 
				 (6,-4.1) -- (6,-2.5) -- 
				(3.5,-1.1) -- (3.5,0.6) -- cycle;
		\filldraw[pattern=lines,pattern color=red!50!white,draw=white] (-3,0.6) -- (-3,-4.1) -- (0.7,-4.1) --
				(0.7,0.6) --  cycle;

		\fill[white, path fading=fade bottom] (-3, 0.6) rectangle (3.6, -2);
		\fill[white, path fading=fade top] (-3, -3.9) rectangle (6.1, -4.1);
		\fill[white, path fading=fade left] (-3, 0.6) rectangle (-2.8, -4.1);

		\draw[ultra thick,red!50!white] (0.7,0.6) -- (0.7,-4.1);
		\draw[ultra thick,red!20!white] (3.5,0.6) -- (3.5,-1.1) -- (6,-2.5) -- (6,-4.1);

		\node[] at (-1.6,0.4) (wone)  {\textbf{para-NP-hard:}};
		\node[] at (1.8,0.4) (wone)  {\textbf{W[1]-hard:}};
		\node[] at (4.2,0.4) (wone)  {\textbf{FPT:}};

	\node[para,fill=red!50!white] at (-1,-2.8) (md)  {\tworows{Max.}{Degree}};
	\node[para,fill=red!50!white] at (-1,-0.8) (dc) {\tworows{Distance to}{Clique}};
	\node[para,fill=red!20!white] at (2.5,-2) (ddp)  {\tworows{Distance to}{Disjoint Paths}};
	\node[para,fill=red!20!white] at (2.5,-3) (fvs)  {Feedback Vertex No.};
	\node[para,fill=red!20!white] at (5,-3) (pw)  {Pathwidth};
	\draw[thick] (ddp) to (pw);
	\draw[,thick] (ddp) to (fvs);

  	\node[para,fill=green!10!white] at (5,-0.2) (vc) {Vertex Cover No.};
  	\node[para,fill=green!10!white] at (5,-1) (td)  {Treedepth};
  	\node[para,fill=green!10!white] at (7.5,-2) (fes)  {Feedback Edge No.};
	\draw[,thick] (vc.west) to (ddp);
	\draw[,thick] (vc) to (td);
	\draw[,thick] (fes) to (fvs);
	\draw[,thick] (td) to (pw);
	
	\node[rotate=-45] at (0.3,-1.5) (dots1)  {$\boldsymbol\vdots$};
	\draw[,thick] (dots1) to (md);
	\node[rotate=-45,inner sep=-2pt] at (-2.3,-2.1) (dots2)  {$\boldsymbol\vdots$};
	\draw[,thick] (dots2) to (dc);
	\node[inner sep=-2pt] at (-1,-3.8) (dots3)  {$\boldsymbol\vdots$};
	\draw[,thick] (dots3) to (md);
	\node[inner sep=-2pt] at (2.5,-3.8) (dots4)  {$\boldsymbol\vdots$};
	\draw[,thick] (dots4) to (fvs);
	\node[inner sep=-2pt] at (5,-3.8) (dots5)  {$\boldsymbol\vdots$};
	\draw[,thick] (dots5) to (pw);
	\node[inner sep=-2pt] at (7.5,-3) (dots5)  {$\boldsymbol\vdots$};
	\draw[,thick] (dots5) to (fes);
	\node at (7.5,-0.9) (dots6)  {$\boldsymbol\vdots$};
	\draw[,thick] (dots6) to (fes);
\end{tikzpicture}
\caption{Relevant part of the hierarchy among classic parameters of the underlying graph (cf.~Sorge~et~al.~\cite{parameters})
 for our results for \probRestlessPath.}
 \label{fig:hierarchy} 
 \vspace{-2ex}
\end{figure}

\section{Preliminaries}
\label{sec:path:prelims}
Here, we formally introduce the most important concepts related to temporal graphs and paths, 
and give the formal problem definition %
of \textsc{(Short) Restless Temporal $(s,z)$-Path}.

  An \emph{interval} is an ordered set
$[a,b] \coloneqq \{n \mid n \in \N \wedge a \leq n \leq b \},$ where~$a,b \in \N$. Further, let $[a]:=[1,a]$.

\ownparagraph{Static graphs.}
We use standard notation from (static) graph theory~\cite{Die16}. Unless stated otherwise, we assume graphs in this paper to be \emph{undirected} and \emph{simple}. To clearly distinguish them from temporal graphs, they are sometimes referred to as \emph{static} graphs.
Given a (static) graph $G=(V,E)$ with $E\subseteq \binom{V}{2}$, we denote by~$V(G):=V$ and $E(G):=E$ the sets of its
vertices and edges, respectively. %

We call two vertices $u,v\in V$ \emph{adjacent} if $\{u,v\}\in E$. Two edges $e_1,e_2\in E$ are \emph{adjacent} if $e_1\cap e_2\neq \emptyset$.
For a vertex $v \in V$, we denote by $\deg_G(v)$ the \emph{degree} of the vertex, that is, $\deg_G(v) = |\{w \in V \mid \{v,w\} \in E\}|$.
For some vertex subset $V'\subseteq V$, we denote by $G[V']$ the \emph{induced} subgraph of $G$ on the vertex set $V'$, that is, $G[V']=(V',E')$ where $E' = \{\{v,w\}\mid \{v,w\}\in E\wedge v\in V'\wedge w\in V'\}$.
For some vertex subset $V'\subseteq V$, we denote by $G-V'$ the subgraph of $G$ without the vertices in $V'$, that is, $G-V'=G[V\setminus V']$. 
For some edge subset $E'\subseteq E$, we denote by $G-E'$ the subgraph of $G$ without the edges $E'$, that is, $G-E'=(V,E \setminus E')$.

An \emph{$(s,z)$-path} of length $k$ is a sequence~$P=(\{s=v_0,v_1\},\{v_1,v_2\},\ldots,$ $\{v_{k-1},v_{k}=z\})$ of edges such that for all $i\in [k]$ we have that $\{v_{i-1}, v_i\} \in E$ and $v_i \not = v_j$ for all $i,j \in [k]$.
We denote $v_0$ and $v_k$ as the endpoints of~$P$. 
We further denote by~$E(P)$ the set of edges of path $P$, that is, 
$E(P)=\{\{v_0,v_1\}, \{v_1,v_2\},\ldots, \{v_{k-1},v_{k}\}\}$ and 
by $V(P)$ the set of vertices visited by the path, that is, $V(P)=\bigcup_{e\in E(P)} e$. 
If $v_0 = v_k$ and $P$ is of length at least three, then $P$ is a \emph{cycle}.

\ownparagraph{Temporal graphs.}
An (undirected, simple) \emph{temporal graph} is a tuple~$\TGfull$ (or \TGcompact for short), with $E_i\subseteq\binom{V}{2}$ for all $i\in[\lifetime]$.
We call $\lifetime(\TG) := \lifetime$ the \emph{lifetime} of $\TG$. As with
static graphs, we assume all temporal graphs in this paper to be undirected and
simple.
We call the graph $G_i(\TG) = (V, E_i(\TG))$ the \emph{\layer}~$i$ of $\TG$
where $E_i(\TG) := E_i$. If $E_i=\emptyset$, then $G_i$ is a \emph{trivial} layer.
We call layers $G_i$ and $G_{i+1}$ \emph{consecutive}.
We call $i$ a \emph{\timestep}. If an edge $e$ is present at time $i$, that is, $e\in E_i$, we say that $e$ has \emph{time stamp} $i$. We further denote $V(\TG):=V$. %
The \emph{underlying graph}~$\ug{\TG}$ of $\TG$ is defined as~$\ug{\TG} := (V, \bigcup_{i=1}^{\lifetime(\TG)} E_i(\TG))$.
To improve readability, we remove~$(\TG)$ from the introduced notations whenever it is clear from the context.
For every $v\in V$ and every \timestep $t\in [\lifetime]$, we denote the \emph{appearance
of vertex} $v$ \emph{at time}~$t$ by the pair $(v,t)$. For every $t\in [\lifetime]$ and every $e\in E_t$ we call the pair $(e,t)$ a \emph{time edge}.
For a time edge $(\{v,w\},t)$ we call the vertex appearances $(v,t)$ and $(w,t)$ its \emph{endpoints}.
We assume that the \emph{size} (for example when referring to input sizes in running time analyzes) of $\TG$ is $|\TG|:=|V|+\sum_{i=1}^\lifetime \min\{1,|E_i|\}$, that is, we do not assume that we have compact representations of temporal graphs. Finally, we write $n$ for $|V|$.

		\label{def:temppath}
		A 
		\emph{temporal ($s,z$)-walk} (or \emph{temporal walk})  
		of length~$k$ from vertex $s=v_0$ to vertex $z=v_k$ in a temporal graph~$\TGcompact$ is a sequence
$P = \left(\left(v_{i-1},v_i,t_i\right)\right)_{i=1}^k$
of triples that we call \emph{transitions} 
such that for all $i\in[k]$ we have that $\{v_{i-1},v_i\}\in E_{t_i}$ and for
all $i\in [k-1]$ we have that $t_i \leq t_{i+1}$.
Moreover, we call $P$ a 
\emph{temporal ($s,z$)-path} (or \emph{temporal path})  
of length $k$ 
if~$v_i\neq v_j$ for all~$i, j\in \{0,\ldots,k\}$ with $i\neq j$.
Given a temporal path $P=\left(\left(v_{i-1},v_i,t_i\right) \right)_{i=1}^k$, 
we denote the set of vertices of $P$ by $V(P)=\{v_0,v_1,\ldots,v_k\}$.
Moreover, we say that $P$ \emph{visits} the vertex $v_i$ at time $t$ if $t \in [t_i,t_{i+1}]$, where~$i \in [k-1]$.
A \emph{restless} temporal path is not allowed to wait an arbitrary amount of
time in a vertex, but has to leave any vertex it visits within the next
$\Delta$ time steps, for some given value of~$\Delta$. Analogously to the
non-restless case, a restless temporal walk may visit a vertex multiple times.
\begin{definition}%
		\label{def:rtemppath}
		A temporal path (walk) $P = \left(\left(v_{i-1},v_i,t_i\right)\right)_{i=1}^k$
		is \emph{$\Delta$-restless}
if  $t_i \leq t_{i+1} \leq t_i + \wait$, for all $i\in [k-1]$.
We say that $P$ \emph{respects} the waiting time~$\wait$.
\end{definition}
Having this definition at hand, we are ready to define the main decision problem of this work.

\probDefRestlessPath
Note the waiting time at the source vertex $s$ is ignored. 
This is without loss of generality, since one can add an auxiliary degree one source vertex which is only in the first layer adjacent to $s$.
We also consider a variant, where we want to find $\Delta$-restless paths of a certain maximum length.
In the \probSRestlessPath problem, we are additionally given a integer $k \in \N$ and the question 
is whether there is a $\Delta$-restless temporal path of length at most $k$ from
$s$ to $z$ in~$\TG$?
Note that \probRestlessPath is the special case of \probSRestlessPath for
$k=|V|-1$ and that both problems are in \NP. %

\ownparagraph{Parameterized complexity.}
We use standard notation and terminology from parameterized
complexity theory~\cite{Cyg+15}
and give here a brief overview of the most important concepts that are used in this paper.
A \emph{parameterized problem} is a language $L\subseteq \Sigma^* \times \N$, where $\Sigma$ is a finite alphabet. We call the second component
the \emph{parameter} of the problem.
A parameterized problem is \emph{fixed-pa\-ram\-e\-ter tractable} (in the complexity class \FPT{})
if there is an algorithm that solves each instance~$(I, r)$ in~$f(r) \cdot |I|^{O(1)}$ time,
for some computable function $f$. 
A decidable parameterized problem $L$ admits a \emph{polynomial kernel} if there is a polynomial-time algorithm that transforms each instance $(I,r)$ into an instance $(I', r')$ such that $(I,r)\in L$ if and only if $(I',r')\in L$ and $|(I', r')|\in r^{O(1)}$. 
If a parameterized problem is hard for the parameterized complexity class \wone, then it is (presumably) not in~\FPT{}.
The complexity classes \wone\ is closed under parameterized reductions, which may run in \FPT-time and additionally set the new parameter to a value that exclusively depends on the old parameter. 

\ownparagraph{Basic observations.}\label{sec:basicobs}
If there is a $\Delta$-restless \nonstrpath{s,z} 
$\left(\left(v_{i-1},v_i,t_i\right)\right)_{i=1}^k$ %
in a temporal graph $\TG$, then  
$\big( \{v_0,v_1\}, \dots, \{v_{k-1}, v_k\} \big)$ 
is an $(s,z)$-path in the underlying graph~$\UG$. 
The other direction does not necessarily hold.
However, we now show that for any $(s,z)$-path in~$\UG$ we can decide in linear time whether this path is the support of a $\Delta$-restless \nonstrpath{s,z} in~$\TG$.
As a consequence, we can decide \probRestlessPath~in linear time for any temporal graph where there exists a unique $(s,z)$-path in the underlying graph, in particular, if the underlying graph is a forest. %

\begin{lemma}
	\label{lem:restless-path-on-a-path}
	Let $\TGcompact$ be a temporal graph where the underlying graph $\UG$ is an~$(s,z)$-path with $s,z \in V$.
	Then there is an algorithm which computes in $O(|\TG|)$ time
	the set 
	$$\mathcal A =\{ t \mid 
	\text{there is a $\Delta$-restless \nonstrpath{s,z} with arrival time } t \}.$$
\end{lemma}
\begin{proof}
Let $V(\UG)=\{s=v_0, \ldots,v_n=z\}$ be the vertices and $E(\UG)=\{e_1=\{v_0,v_1\},\ldots,e_n=\{v_{n-1},v_n\}\}$ be the edges of the underlying path. 
We further define $L_i$ as the set of layers of $\TG$ in which the edge $e_i\in E(\UG)$ exists, that is, $L_{i}:=\{t\mid e_i\in E_t\}$. 

In the following, we construct a dynamic program on the path.
We compute for every vertex~$v_i$ %
the table entry~$T[v_i]$ which is defined as the set of all layers~$t$ such that there exists a $\Delta$-restless \nonstrpath{s,v_i} with arrival time~$t$. 
It holds that~$T[v_1]=L_{1}$. 
Then, for all $i \in [2,\lifetime]$, we compute the table entry $T[v_i]$ by checking for each layer~$t \in L_{i}$ 
whether there exists a $\Delta$-restless \nonstrpath{s,v_{i-1}} that arrives in a layer $t' \in T[v_{i-1}]$ 
such that we can extend the path to the vertex $v_i$ in layer~$t$ without exceeding the maximum waiting time $\Delta$, 
that is, $0\leq t-t' \leq \Delta$.
Formally, we have
$$ T[v_i] := \{t \in L_{i} \mid \text{there is a }t' \in T[v_{i-1}] \text{ with } 0 \leq t-t' \leq \Delta \}.$$
It is easy to verify that $T[v_i]$ contains all layers~$t$ such that there exists a $\Delta$-restless \nonstrpath{s,v_i} with arrival time~$t$.  
After computing the last entry $T[v_n]$, 
this entry contains the set~$\mathcal A$ of all layers~$t$ such that there exists a $\Delta$-restless \nonstrpath{s,z} with arrival time~$t$. 

In order to compute a table entries $T[v_i]$ in linear time, we will need sorted lists of layers for~$L_i$ and~$T[v_{i-1}]$ in ascending order.
The sorted lists~$L_i$ of layers can be computed in~$O(|\TG|)$: 
For every $t \in [\ell]$, we iterate over each $e_i \in E_t$ and add $t$ to~$L_i$.
Now assume that $L_i$ and $T[v_{i-1}]$ are lists of layers both in ascending order, then we can compute the table entry~$T[v_i]$ in~$O(|T[v_{i-1}]| + |L_{i}|)$ time as follows.
Let $T[v_i]$ be initially empty. 
Let $t$ be the first element in~$L_{i}$ and let~$t'$ be the first element in~$T[v_{i-1}]$: 
\begin{enumerate}
\item If $t'> t$, then replace $t$ with the next layer in $L_{i}$ and repeat.
\item If $t-t' \leq \Delta$, then add $t$ to $T[v_i]$, replace $t$ with the next layer in $L_{i}$ and repeat.
\item Else, replace $t'$ with the next layer in $T[v_{i-1}]$ and repeat.
\end{enumerate}
This is done until all elements in one of the lists are processed. 

The resulting list $T[v_i]$ is again sorted. 
Due to this and $T[v_1] (= L_1)$ being sorted, we can assume that $T[v_{i-1}]$ is given as a sorted list of layers when computing $T[v_i]$.
Hence, we can compute each table entry $T[v_i]$ in $O(|T[v_{i-1}]| + |L_{i}|)$ time.
It further holds that $|T[v_{i}]|\leq |L_{i}|$ and $\sum_{i=1}^{n} |L_{i}|=\sum_{i=1}^{\ell} |E_i|$. 
Hence, the dynamic program runs in~$O(|\TG|)$ time. 
\end{proof}

Furthermore, it is easy to observe that computational hardness of \probRestlessPath for some fixed value of $\Delta$ implies hardness for all larger finite values of $\Delta$.
This allows us to construct hardness reductions for small fixed values of $\Delta$ and still obtain general hardness results.

\begin{observation}\label{obs:path:delta}
		Given an instance $I=(\TG, s,z,k,\Delta)$ of \probSRestlessPath,
		we can construct in linear time an instance $I'=(\TG', s,z,k,\Delta+1)$ of \probSRestlessPath
		such that $I$ is a \yes-instance if and only if $I'$ is a \yes-instance.
\end{observation}
\begin{proof}
	The result immediately follows from the observation that a temporal graph~$\TG$ contains a $\Delta$-restless \nonstrpath{s,z}  
	if and only if 
	the temporal graph $\TG'$ contains a $(\Delta+1)$-restless \nonstrpath{s,z}, where
	$\TG'$ is obtained from $\TG$ by inserting one trivial (edgeless) layer
	after every $\Delta$ consecutive layers. %
\end{proof}
However, for some special values of $\Delta$ we can solve \probRestlessPath in polynomial time.
\begin{observation}
		\label{obs:delta:zero}\label{obs:delta:infinite}
  \probRestlessPath on instances $(\TG, s,z,\Delta)$ can be solved in polynomial time, 
  if $\Delta = 0$ or $\Delta \geq \lifetime$.
\end{observation}
\begin{proof}
  Considering $\Delta=0$ implies that the entirety of a path between $s$ and $z$ must be realized in a single layer. 
  Thus, the problem is equivalent to testing if at least one of the layers $G_i$
  contains a (static) path between $s$ and $z$.

  If $\Delta \geq \lifetime$,
  then $\Delta$-restless temporal paths correspond to unrestricted temporal paths, whose computation can be made using any of the (polynomial time) algorithms in~Bui-Xuan, Ferreira, and Jarry~\cite{XFJ03}. 
\end{proof}

\section{Hardness results for restless temporal paths}
\label{sec:path:hardness}

In this section we present a thorough analysis of the computational hardness of \probRestlessPath which also transfers to \probSRestlessPath.

\ownparagraph{NP-hardness for few layers.}\label{sec:nph}
We start by showing that \probRestlessPath is \NP-complete even if the lifetime
of the input temporal graph is constant. The reduction is similar in
spirit to the classic NP-hardness reduction for \textsc{2-Disjoint Paths} in
directed graphs by Fortune et al.~\cite{FortuneHW80}.

\begin{theorem}\label{thm:probRestlessPath:NPh}
\probRestlessPath is \classNP{}-complete for all finite $\wait\ge 1$ and $\lifetime\ge \Delta+2$ even if every edge has only one time stamp.  
\end{theorem}
\begin{proof}
We show this result by a reduction from the \classNP{}-complete \probExactThreeFourSAT problem~\cite{Tov84}. The problem \probExactThreeFourSAT asks whether a formula $\phi$ is satisfiable, assuming that it is given in conjunctive normal form, each clause having exactly three literals and each variable appearing in exactly four clauses.

Let $\phi$ be an instance of \probExactThreeFourSAT with $n$ variables and $m$ clauses. We construct a temporal graph $\TGcompact$ with $\lifetime=3$ consisting of a series of variable gadgets followed by dedicated vertices $s_n$ and $s'$ and then a series of clause gadgets. It is constructed in such a way that for $\Delta=1$, any $\Delta$-restless temporal $(s,z)$-path has to visit a vertex $s_n$ and each possible $\Delta$-restless temporal $(s,s_n)$-path represents exactly one variable assignment for the formula $\phi$. Further we show that for any $\Delta$-restless temporal $(s,s_n)$-path it holds that it can be extended to a $\Delta$-restless temporal $(s,z)$-path if and only if the $\Delta$-restless temporal $(s,s_n)$-path represents a satisfying assignment for the formula~$\phi$.

\proofpar{Variable Gadget.} 
We start by adding a vertex $s$ to the vertex set $V$ of $\TG$. For each variable $x_i$ with $i \in [n]$ of $\phi$, we add 9 fresh vertices to $V$: $x_i^{(1)}$, $x_i^{(2)}$, $x_i^{(3)}$, $x_i^{(4)}$, $\bar x_i^{(1)}$, $\bar x_i^{(2)}$, $\bar x_i^{(3)}$, $\bar x_i^{(4)}$, and $s_i$.
Each variable~$x_i$ is represented by a gadget consisting two disjoint path segments of four vertices each. One path segment is formed by $x_i^{(1)}$, $x_i^{(2)}$, $x_i^{(3)}$, and $x_i^{(4)}$ in that order and the second path segment is formed by $\bar x_i^{(1)}$, $\bar x_i^{(2)}$, $\bar x_i^{(3)}$, and~$\bar x_i^{(4)}$ in that order. 
The connecting edges all appear exclusively at time step one, that is, $\{x_i^{(1)}, x_i^{(2)}\}$, $\{x_i^{(2)}, x_i^{(3)}\}$, and $\{x_i^{(3)}, x_i^{(4)}\}$ are added to $E_1$. Analogously for the edges connecting $\bar x_i^{(1)}$, $\bar x_i^{(2)}$, $\bar x_i^{(3)}$, and $\bar x_i^{(4)}$.
 Intuitively, if a $\Delta$-restless temporal $(s,z)$-path passes the first segment, then this corresponds to setting the variable $x_i$ to \falsevalue. If it passes the second segment, then the variable is set to \truevalue.
For all $i \in [n-1]$ we add the edges $\{x_i^{(4)},s_i\}$, $\{\bar x_i^{(4)},s_i\}$, $\{s_i,\bar x_{i+1}^{(1)}\}$, and $\{s_i,\bar x_{i+1}^{(1)}\}$ to $E_1$ and, additionally, we add~$\{s,x_{1}^{(1)}\}$, $\{s,\bar x_{1}^{(1)}\}$, $\{x_{n}^{(4)},s_n\}$, and $\{\bar x_{n}^{(4)},s_n\}$ to $E_1$. 

We can observe that there are exactly $2^n$ different temporal $(s,s_n)$-paths at time step~one. Intuitively, each path represents exactly one variable assignment for the formula~$\phi$. 

\proofpar{Clause Gadget.}
We add a vertex $z$ to $V$. For each clause $c_j$ with $j \in [m]$ we add a fresh vertex $c_j$ to $V$. We further add a vertex $s'$ to $V$ and add the edge $\{s_n,s'\}$ to $E_2$. Let $x_i$ (or $\bar x_i$) be a literal that appears in clause~$c_j$ and let this be the $k$th appearance of variable~$x_i$ in $\phi$. Then, we add the edges~$\{c_j, x_i^{(k)}\}, \{x_i^{(k)},c_{j+1}\}$ (or $\{c_j, \bar x_i^{(k)}\}, \{\bar x_i^{(k)},c_{j+1}\}$) to~$E_3$ (where $c_{m+1} = z$). 
Finally, we add the edge~$\{s',c_1\}$ to $E_3$. 

Hence, there are exactly $3^m$ different temporal $(s',z)$-paths at time step~three. Each path must visit the clause vertices $c_1,\ldots,c_m$ in the given order by construction.

Finally, we set $\wait=1$. This finishes the construction, for a visualization see \cref{fig:pathred1}. It is easy to check that every edge in the constructed temporal graph has only one time step and that the temporal graph can be computed in polynomial time.

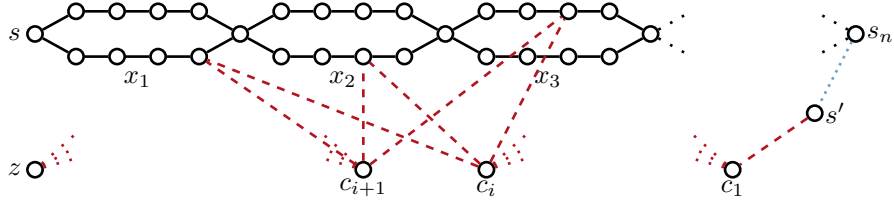
\begin{figure}[t]
\begin{center}
\begin{tikzpicture}[line width=1pt, scale=.3,xscale=0.9]
    \node (Sl) at (-3,1) {$s$};
    \node (X1l) at (3,-1) {$x_1$};
    \node (X2l) at (13,-1) {$x_2$};
    \node (X3l) at (23,-1) {$x_3$};
    \node[vert2] (S) at (-2,1) {};
    \node[vert2] (A1) at (0,0) {}; 
    \node[vert2] (A2) at (2,0) {};
    \node[vert2] (A3) at (4,0) {};
    \node[vert2] (A4) at (6,0) {};
    \node[vert2] (A5) at (0,2) {};
    \node[vert2] (A6) at (2,2) {};
    \node[vert2] (A7) at (4,2) {};
    \node[vert2] (A8) at (6,2) {};
    \node[vert2] (S2) at (8,1) {};
    \node[vert2] (B1) at (10,0) {}; 
    \node[vert2] (B2) at (12,0) {};
    \node[vert2] (B3) at (14,0) {};
    \node[vert2] (B4) at (16,0) {};
    \node[vert2] (B5) at (10,2) {};
    \node[vert2] (B6) at (12,2) {};
    \node[vert2] (B7) at (14,2) {};
    \node[vert2] (B8) at (16,2) {};
    \node[vert2] (S3) at (18,1) {};
    \node[vert2] (C1) at (20,0) {}; 
    \node[vert2] (C2) at (22,0) {};
    \node[vert2] (C3) at (24,0) {};
    \node[vert2] (C4) at (26,0) {};
    \node[vert2] (C5) at (20,2) {};
    \node[vert2] (C6) at (22,2) {};
    \node[vert2] (C7) at (24,2) {};
    \node[vert2] (C8) at (26,2) {};
    \node[vert2] (S4) at (28,1) {};
    \node[vert2] (SN) at (38,1) {};
    \node (SNl) at (39.2,1) {$s_n$};
    
    \node[vert2] (SP) at (36,-2.5) {};
    \node (SPl) at (37,-2.5) {$s'$};

    \node[vert2] (CC1) at (32,-5) {};   
    \node (CC1l) at (32,-5.8) {$c_1$};   
    
    \node[vert2] (CCi) at (20,-5) {};
    \node (CCil) at (20,-5.8) {$c_i$};
    \node[vert2] (CCii) at (14,-5) {};
    \node (CCiil) at (14,-5.8) {$c_{i+1}$};
    
    \node (Zl) at (-3,-5) {$z$};
    \node[vert2] (Z) at (-2,-5) {};
    
    \draw (S) -- (A1);
    \draw (A1) -- (A2);
    \draw (A2) -- (A3);
    \draw (A3) -- (A4);
    \draw (A4) -- (S2);
    \draw (S) -- (A5);
    \draw (A5) -- (A6);
    \draw (A6) -- (A7);
    \draw (A7) -- (A8);
    \draw (A8) -- (S2);
    
    \draw (S2) -- (B1);
    \draw (B1) -- (B2);
    \draw (B2) -- (B3);
    \draw (B3) -- (B4);
    \draw (B4) -- (S3);
    \draw (S2) -- (B5);
    \draw (B5) -- (B6);
    \draw (B6) -- (B7);
    \draw (B7) -- (B8);
    \draw (B8) -- (S3);
    
    \draw (S3) -- (C1);
    \draw (C1) -- (C2);
    \draw (C2) -- (C3);
    \draw (C3) -- (C4);
    \draw (C4) -- (S4);
    \draw (S3) -- (C5);
    \draw (C5) -- (C6);
    \draw (C6) -- (C7);
    \draw (C7) -- (C8);
    \draw (C8) -- (S4);
    
    \draw[loosely dotted] (S4) -- (30,0);
    \draw[loosely dotted] (S4) -- (30,2);
    \draw[loosely dotted] (SN) -- (36,0);
    \draw[loosely dotted] (SN) -- (36,2);
    
    \draw[color=ourblue!80!black, dotted] (SN) -- (SP);
    \draw[color=ourred, dashed] (SP) -- (CC1);

    \draw[color=ourred, dashed] (CCi) -- (C7);
    \draw[color=ourred, dashed] (CCii) -- (C7);
    \draw[color=ourred, dashed] (CCi) -- (B3);
    \draw[color=ourred, dashed] (CCii) -- (B3);
    \draw[color=ourred, dashed] (CCi) -- (A4);
    \draw[color=ourred, dashed] (CCii) -- (A4);
    
    \draw[loosely dotted,color=ourred] (CC1) -- (30,-4.5);
    \draw[loosely dotted,color=ourred] (CC1) -- (30,-4);
    \draw[loosely dotted,color=ourred] (CC1) -- (30,-3.4);

    \draw[loosely dotted,color=ourred] (CCi) -- (22,-4.5);
    \draw[loosely dotted,color=ourred] (CCi) -- (22,-4);
    \draw[loosely dotted,color=ourred] (CCi) -- (22,-3.4);
    
    \draw[loosely dotted,color=ourred] (CCii) -- (12,-4.5);
    \draw[loosely dotted,color=ourred] (CCii) -- (12,-4);
    \draw[loosely dotted,color=ourred] (CCii) -- (12,-3.4);
    
    \draw[loosely dotted,color=ourred] (Z) -- (0,-4.5);
    \draw[loosely dotted,color=ourred] (Z) -- (0,-4);
    \draw[loosely dotted,color=ourred] (Z) -- (0,-3.4);
\end{tikzpicture}
    \end{center}
    \caption{Illustration of the temporal graph constructed by the reduction in
    the proof of \cref{thm:probRestlessPath:NPh}. An excerpt is shown with
    variable gadgets for $x_1$, $x_2$, and $x_3$ and the clause gadget for
    $c_i=(x_1\vee x_2\vee \neg x_3)$, where $x_1$ appears for the fourth time,
    $x_2$ appears for the third time, and~$x_3$ also appears for the third time.
    Black edges appear at time step one, the blue (dotted) edge~$\{s_n,s'\}$
    appears at time step two, and the red (dashed) edges appear at time step
    three.}\label{fig:pathred1}
\end{figure}

\proofpar{Correctness.}
Now we can show that $\phi$ is satisfiable if and only if $\TG$ has a $\Delta$-restless temporal $(s,z)$-path.

\RArrow Let us assume there is a satisfying assignment for formula~$\phi$. 
Then we construct a $\Delta$-restless temporal path from vertex $s$ to $z$ as follows. Starting from $s$, for each variable $x_i$ of $\phi$ the $\Delta$-restless temporal path passes through the variables $x_i^{(1)}$, $x_i^{(2)}$, $x_i^{(3)}$, and $x_i^{(4)}$, if~$x_i$ is set to \falsevalue, and $\bar x_i^{(1)}$, $\bar x_i^{(2)}$, $\bar x_i^{(3)}$, and $\bar x_i^{(4)}$, if~$x_i$ is set to \truevalue, at time step one. 
The $\Delta$-restless temporal path arrives at time step one in the vertex $s_n$. In time step two it goes from $s_n$ to $s'$.

At time step three, the $\Delta$-restless temporal path can be extended to~$c_1$. 
In each clause~$c_j$ for $j \in [m]$ there is at least one literal $x_i$ (or $\bar x_i$) that is evaluated to \truevalue. 
Let~$c_j$ be the $k$th clause in which $x_i$ appears. 
We have that, depending on whether $x_i$ is set to \truevalue\ (or \falsevalue), the vertex $x^{(k)}_i$ (or $\bar x^{(k)}_i$) has not been visited so far. 
Hence, the $\Delta$-restless temporal path can be extended from $c_j$ to $c_{j+1}$ (or to~$z$ for $j = m$) at time step three via~$x^{(k)}_i$ (or $\bar x^{(k)}_i$). 
Thus, there exists a $\Delta$-restless temporal~$(s,z)$-path in $\TG$.     

\LArrow
Let us assume that there exists a $\Delta$-restless temporal $(s,z)$-path in the constructed temporal graph $\TG$. 
Note that any $\Delta$-restless temporal $(s,z)$-path must reach $s_n$ in time step~one because the variable gadget has only edges at time step~one and the waiting time $\wait=1$ prevents the path to enter the clause gadget (which only has edges at time step three) before using the edge $\{s_n, s'\}$ at time step two.

It is easy to see that for the first part of the $\Delta$-restless temporal graph from $s$ to $s_n$ it holds that for each $i\in [n]$, it visits either vertices $x_i^{(1)}$, $x_i^{(2)}$, $x_i^{(3)}$, and $x_i^{(4)}$, or vertices $\bar x_i^{(1)}$, $\bar x_i^{(2)}$, $\bar x_i^{(3)}$, and $\bar x_i^{(4)}$. In the former case we set $x_i$ to \falsevalue\ and in the latter case we set $x_i$ to \truevalue. We claim that this produces a satisfying assignment for $\phi$.

In time step~three, the part of the $\Delta$-restless temporal path from~$s'$ to~$z$ has to pass vertices~$c_1,c_2,\ldots,c_m$ to reach~$z$.  
The $\Delta$-restless temporal path passes exactly one variable vertex~$x_i^{(k)}$ (or $\bar x^{(k)}_i$) when going from $c_j$ to~$c_{j+1}$ (and finally from~$c_m$ to~$z$) that has not been visited so far and that corresponds to a variable that appears in the clause $c_j$ for the $k$th time. 
The fact that the variable vertex was not visited implies that we set the corresponding variable to a truth value that makes it satisfy clause~$c_j$. This holds for all $j \in [m]$.
Hence, each clause is satisfied by the constructed assignment and, consequently, $\phi$ is satisfiable. 
\end{proof}

The reduction used in the proof of \cref{thm:probRestlessPath:NPh} 
also yields a running time lower bound 
assuming the Exponential Time Hypothesis (ETH)~\cite{IP01,IPZ01}.

\begin{corollary}\label{cor:path:eth}
\probRestlessPath\ does not admit a $f(\lifetime)^{o(|\TG|)}$-time algorithm for any computable function~$f$ unless the ETH fails.
\end{corollary}
\begin{proof}
		First, note that any \probThreeSAT formula with $m$ clauses can be transformed into an equisatisfiable \probExactThreeFourSAT\ formula with $O(m)$ clauses~\cite{Tov84}. 
The reduction presented in the proof of \cref{thm:probRestlessPath:NPh} produces an instance of \probRestlessPath\ with a temporal graph of size $|\TG|= O(m)$ and $\lifetime=3$. 
Hence an algorithm for \probRestlessPath\ with running time $f(\lifetime)^{o(|\TG|)}$ for some computable function~$f$ would imply the existence of an $2^{o(m)}$-time algorithm for \probThreeSAT. 
This is a contradiction to the ETH~\cite{IP01,IPZ01}.
\end{proof}

Furthermore, the reduction behind \cref{thm:probRestlessPath:NPh} can be modified such that it also yields that \probRestlessPath is
\classNP-hard, even if the underlying graph has constant maximum degree or 
the underlying graph is a clique where one edge ($\{s,z\})$ is missing.
Note that in the latter case the underlying graph contains all edges except the one edge which would turn the instance into a trivial \yes-instance.
\begin{corollary}\label{cor:underlying-nph}
\probRestlessPath\ is \NP-hard, even if the underlying graph has all but one edge or maximum degree six.
\end{corollary}
\begin{proof}
		That \probRestlessPath is \classNP-hard, even if the underlying graph has maximum degree six follows directly from the 
		construction used in the proof of \cref{thm:probRestlessPath:NPh}.
		To show that that \probRestlessPath is \classNP-hard, even if the underlying
		graph has all edges except $\{s,z\}$, we reduce from \probRestlessPath.
		Let $I = (\TGcompact,s,z,\Delta)$ be an instance of \probRestlessPath with $\lifetime = 3$.
		We construct an instance $I' := (\TG'=(V,E'_1,E'_2,E'_3,E'_4,E'_5),s,z,\Delta)$ of \probRestlessPath, 
		where $E'_1 = {V\setminus \{s\} \choose2}$, $E'_2:=E_1$, $E'_3 := E_2$, $E'_4 := E_3$, and $E'_5 = {V\setminus \{z\} \choose2}$.
		Observe that none of the edges in $E_1 \cup E_5$ can be used in temporal $(s,z)$-path.
		Hence, $I$ is a \yes-instance if and only if $I'$ is a \yes-instance.
		Furthermore, $E_1\cup E_5$ contain all possible edges except $\{s,z\}$.	
\end{proof}

\ownparagraph{\wone-hardness for distance to disjoint paths.}\label{sec:wone}
In the following, we show that parameterizing \probRestlessPath with structural graph parameters of the underlying graph of the input temporal graph presumably does not yield fixed-parameter tractability for a large number of popular parameters.
In particular, we show that \probRestlessPath parameterized by the distance to disjoint paths of the underlying graph is \wone-hard.
The \emph{distance to disjoint paths} of a graph $G$ is the minimum number of vertices we have to remove from $G$ 
such that the reminder of $G$ is a set of disjoint paths.
Many well-known graph parameters can be upper-bounded in the distance to disjoint paths, e.g.,~pathwidth, treewidth, and feedback vertex number \cite{parameters}.
Hence, the following theorem also implies that \probRestlessPath is \wone-hard when parameterized by the pathwidth or the feedback vertex number of the underlying graph.

\begin{theorem}\label{thm:probRestlessPath:W1hFVS}
\probRestlessPath parameterized by the distance to disjoint path of the underlying graph is \wone{}-hard for all $\wait\ge 1$ even if every edge has only one time stamp.  
\end{theorem}
\begin{proof}
We present a parameterized reduction from \probMCClique where, given a $k$-partite graph $H=(U_1\uplus U_2\uplus\ldots\uplus U_k, F)$, we are asked to decide whether $H$ contains a clique of size $k$. \probMCClique is known to be \wone-hard when parameterized by the clique size $k$~\cite{Fel+09,Cyg+15}.

Let $(H=(U_1\uplus U_2\uplus\ldots\uplus U_k, F),k)$ be an instance of \probMCClique. For each~$i,j\in[k]$ with $i<j$ let $F_{i,j} = \{\{u,v\}\in F \mid u\in U_i \wedge v\in U_j\}$ be the set of edges between vertices in $U_i$ and $U_j$. We can assume that $k\ge 3$, otherwise we can solve the instance in polynomial time. Without loss of generality, we assume that for all $i,j,i',j'\in[k]$ with $i<j$ and $i'<j'$ we have that $|F_{i,j}|=|F_{i',j'}|=m$ for some $m\in\N$. Note that if this is not the case, we add new vertices and single edges to increase the cardinality of some set $F_{i,j}$ and this does not introduce new cliques since $k\ge 3$. We further assume without loss of generality that $|U_1|=|U_2|=\ldots=|U_k|=n$ for some $n\in\N$. If this is not the case, we can add additional isolated vertices to increase the cardinality of some set $U_i$. We construct a temporal graph $\TGcompact$ with two distinct vertices $s,z\in V$ such that there is a $\Delta$-restless \nonstrpath{s,z}\ in $\TG$ if and only if~$H$ contains a clique of size $k$. 
Furthermore, we show that the underlying graph~$\UG$ of~$\TG$ has a distance to disjoint paths of $O(k^2)$.

\proofpar{Vertex Selection Gadgets.} For each set $U_i$ with $i\in[k]$ of the vertex set of $H$ we create the following gadget. Let $U_i = \{u_1^{(i)}, u_2^{(i)}, \ldots, u_{n}^{(i)}\}$. We create a path of length $k\cdot n + n + 1$ on fresh vertices $w_1^{(i)},v_{1,1}^{(i)}, v_{1,2}^{(i)}, \ldots, v_{1,k}^{(i)},w_2^{(i)}, v_{2,1}^{(i)}, \ldots, v_{n,k}^{(i)},w_{n+1}^{(i)}$. Intuitively, this path contains a segment of length $k$ for each vertex in $U_i$ which are separated by the vertices $w_j^{(i)}$, and the construction will allow a $\Delta$-restless \nonstrpath{s,z} to skip exactly one of these segments, which is going to correspond to selecting this vertex for the clique.

Formally, for each vertex $u_j^{(i)}\in U_i$ we create $k$ vertices $v_{j,1}^{(i)}, v_{j,2}^{(i)}, \ldots, v_{j,k}^{(i)}$, which we call the \emph{segment corresponding to $u_j^{(i)}$}. We further create vertices $w_1^{(i)}, w_2^{(i)}, \ldots, w_{n+1}^{(i)}$. For all $j\in[n]$ and $x\in [k-1]$ we connect vertices~$v_{j,x}^{(i)}$ and~$v_{j,x+1}^{(i)}$ with an edge at time $(i-1)\cdot n + j$ and we connect $w_j^{(i)}$ with $v_{j,1}^{(i)}$ and $w_{j+1}^{(i)}$ with $v_{j,k}^{(i)}$ at time $(i-1)\cdot n + j$ each.

Lastly, we introduce a ``skip vertex'' $s^{(i)}$ that will allow a $\Delta$-restless \nonstrpath{s,z} to skip one path segment of length $k$ that corresponds to one of the vertices in~$U_i$. For each $j\in[n+1]$, we connect vertices~$s^{(i)}$ and~$w_{j}^{(i)}$ with an edge at time $(i-1)\cdot n+j$.

Now we connect the gadgets for all $U_i$'s in sequence, that is, a $\Delta$-restless \nonstrpath{s,z} passes through the gadgets one after another, selecting one vertex of each part $U_i$. Formally, for all $i\in [k-1]$, we connect vertices $w_{n+1}^{(i)}$ and $w_{1}^{(i+1)}$ with an edge at time $i\cdot n + 1$. 
It is easy to check that after the removal of the vertices $\{s^{(1)}, s^{(2)}, \ldots, s^{(k)}\}$,
the vertex selection gadget is a path.
The vertex selection gadget is visualized in \cref{fig:pathred2}.\todo{not ready for black and white printing.}

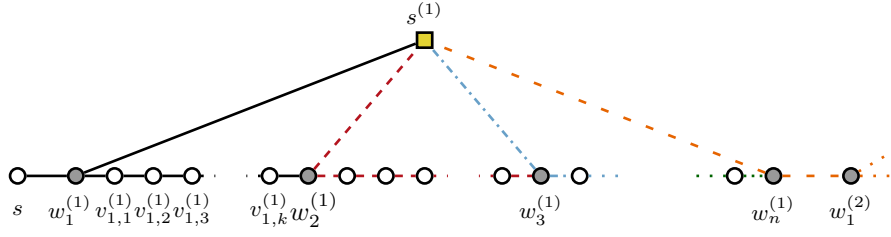
\begin{figure}[t]
\begin{center}
\begin{tikzpicture}[line width=1pt, scale=.3, yscale=1.5,xscale=0.85]

    \node (Sl) at (-3,-1) {\footnotesize $s$}; 
    \node[vert2] (S) at (-3,0) {}; 
    \node[vert3,fill=lipicsyellow!80] (S1) at (18,4) {}; 
    \node (S1l) at (18,4.8) {\footnotesize $s^{(1)}$}; 
    \node[vert2,fill=gray!80] (W1) at (0,0) {}; 
    \node (W1l) at (-.3,-1) {\footnotesize $w_1^{(1)}$}; 
    \node[vert2] (A1) at (2,0) {}; 
    \node (A1l) at (2,-1) {\footnotesize $v_{1,1}^{(1)}$}; 
    \node[vert2] (A2) at (4,0) {}; 
    \node (A2l) at (4,-1) {\footnotesize $v_{1,2}^{(1)}$}; 
    \node[vert2] (A3) at (6,0) {}; 
    \node (A3l) at (6,-1) {\footnotesize $v_{1,3}^{(1)}$}; 
    \node[vert2] (Ak) at (10,0) {}; 
    \node (Akl) at (10,-1) {\footnotesize $v_{1,k}^{(1)}$}; 
    \node[vert2,fill=gray!80] (W2) at (12,0) {}; 
    \node (W2l) at (12.3,-1) {$w_2^{(1)}$}; 
    \node[vert2] (B1) at (14,0) {}; 
    \node[vert2] (B2) at (16,0) {}; 
    \node[vert2] (B3) at (18,0) {}; 
    \node[vert2] (Bk) at (22,0) {}; 
    \node[vert2,fill=gray!80] (W3) at (24,0) {}; 
    \node (W3l) at (24,-1) {\footnotesize $w_3^{(1)}$}; 
    \node[vert2] (C1) at (26,0) {}; 
    \node[vert2] (Dk) at (34,0) {}; 
    \node[vert2,fill=gray!80] (Wn) at (36,0) {}; 
    \node (Wnl) at (36,-1) {\footnotesize $w_n^{(1)}$}; 
    \node[vert2,fill=gray!80] (Wnn) at (40,0) {}; 
    \node (Wnnl) at (40,-1) {\footnotesize $w_1^{(2)}$}; 
    
    \draw (S) -- (W1);
    \draw (W1) -- (A1);
    \draw (A1) -- (A2);
    \draw (A2) -- (A3);
    \draw[loosely dotted] (A3) -- (7.2,0);
    \draw[loosely dotted] (Ak) -- (8.8,0);
    \draw (Ak) -- (W2);
    
    \draw[color=ourred, dashed] (W2) -- (B1);
    \draw[color=ourred, dashed] (B1) -- (B2);
    \draw[color=ourred, dashed] (B2) -- (B3);
    \draw[loosely dotted,color=ourred] (B3) -- (19.2,0);
    \draw[loosely dotted,color=ourred] (Bk) -- (20.8,0);
    \draw[color=ourred, dashed] (Bk) -- (W3);
    
    \draw[color=ourblue!80!black, dashdotted] (W3) -- (C1);
    \draw[loosely dotted,color=ourblue!80!black] (C1) -- (28,0);
    
    \draw[loosely dotted,color=ourgreen] (Dk) -- (32,0);
    \draw[color=ourgreen, dotted] (Dk) -- (Wn);
    
    \draw[color=ourorange, loosely dashed] (Wn) -- (Wnn);
    \draw[loosely dotted,color=ourorange] (Wnn) -- (42,0);
    \draw[loosely dotted,color=ourorange] (Wnn) -- (41.8,0.6);
    
    \draw (W1) -- (S1);
    \draw[color=ourred, dashed] (W2) -- (S1);
    \draw[color=ourblue!80!black, dashdotted] (W3) -- (S1);
    \draw[color=ourorange, loosely dashed] (Wn) -- (S1);
\end{tikzpicture}
    \end{center}
    \caption{Visualization of the vertex selection gadget for $U_1$ from the
    reduction of \cref{thm:probRestlessPath:W1hFVS}. Black edges appear at time
    step one, red edges (densely dashed) at time step two, blue edges
    (dashdotted) at time step three, green edges (dotted) at time step $n-1$, and orange edges
    (loosely dashed) at time step $n$. For the segment corresponding to
    $u_1^{(1)}\in U_1$ all vertex names are presented, for the other segments the names are analogous but omitted. The auxiliary $w_1^{(1)},\ldots,w_n^{(1)},\ldots$ vertices are colored gray.
	The ``skip vertex'' $s^{(1)}$ is depicted as a yellow square. 
	Note that after the removal of $s^{(1)}$ the vertex selection gadget for $U_1$ is a path.}\label{fig:pathred2}

\end{figure}

\proofpar{Validation Gadgets.} A $\Delta$-restless \nonstrpath{s,z} has to pass through the validation gadgets after it passed through the vertex selection gadgets. 
Here, we are forced to choose a point in time where we visit two vertices of two different vertex selection gadgets.
This choice corresponds to the selection of an edge.
Intuitively, this should only be possible if the selected vertices form a clique. We construct the gadget in the following way.

For each $i,j\in[k]$ with $i<j$ let the edges in $F_{i,j}$ be ordered in an arbitrary way, that is, $F_{i,j} = \{e_1^{(i,j)}, e_2^{(i,j)},\ldots, e_{m}^{(i,j)}\}$. We create two paths of length $2m$ on fresh vertices $v_{1,1}^{(i,j)}, v_{1,2}^{(i,j)}, v_{2,1}^{(i,j)},$ $v_{2,2}^{(i,j)}, \ldots, v_{m,2}^{(i,j)}$ and $v_{1,3}^{(i,j)}, v_{1,4}^{(i,j)}, v_{2,3}^{(i,j)}, v_{2,4}^{(i,j)}, \ldots, v_{m,4}^{(i,j)}$, respectively. Intuitively, the \emph{first path} selects an edge from $U_i$ to $U_j$ and the transition to the \emph{second path} should only be possible if the two endpoints of the selected edge are selected in the corresponding vertex selection gadgets.

Formally, for each edge $e_h^{(i,j)}\in F_{i,j}$ we create four vertices $v_{h,1}^{(i,j)}, v_{h,2}^{(i,j)}, v_{h,3}^{(i,j)}, v_{h,4}^{(i,j)}$. Furthermore, we introduce three extra vertices $s_1^{(i,j)},s_2^{(i,j)},s_3^{(i,j)}$. For all $h\in[m]$ we connect vertices~$v_{h,1}^{(i,j)}$ and~$v_{h,2}^{(i,j)}$ with an edge at time $y_{i,j} + 2h-1$, we connect vertices~$v_{h,1}^{(i,j)}$ and~$s_1^{(i,j)}$ with an edge at time $y_{i,j} + 2h-1$, we connect vertices~$v_{h,3}^{(i,j)}$ and~$v_{h,4}^{(i,j)}$ with an edge at time $y_{i,j} + 2h-1$, we connect vertices~$v_{h,3}^{(i,j)}$ and~$s_3^{(i,j)}$ with an edge at time $y_{i,j} + 2h-1$, and if $h<m$, we connect vertices~$v_{h,2}^{(i,j)}$ and~$v_{h+1,1}^{(i,j)}$ with an edge at time $y_{i,j} + 2h$ and we connect vertices~$v_{h,4}^{(i,j)}$ and~$v_{h+1,3}^{(i,j)}$ with an edge at time $y_{i,j} + 2h$, where $y_{i,j}=k\cdot n+2m\cdot (i\cdot j+\frac{1}{2}\cdot i\cdot (i-1)-1)$ (the value of $y_{i,j}$ can be interpreted as a ``time offset'' for the validation gadget for $F_{i,j}$, the value is computed by adding all time steps needed in validation gadget for $F_{i',j'}$ with $i'<j'$, $i'\le i$, $j'\le j$, and $(i',j')\neq (i,j)$). Next, for each edge $e_h^{(i,j)}= \{u_a^{(i)},u_b^{(j)}\}\in F_{i,j}$ we connect vertices $s_1^{(i,j)}$ and $v_{a,j}^{(i)}$  (from the vertex selection gadget for $U_i$) with an edge at time $y_{i,j} + 2h-1$, we connect vertices~$s_2^{(i,j)}$ and~$v_{a,j}^{(i)}$ with an edge at time $y_{i,j} + 2h-1$, we connect vertices~$s_2^{(i,j)}$ and~$v_{b,i}^{(j)}$ (from the vertex selection gadget for $U_j$) with an edge at time $y_{i,j} + 2h-1$, and we connect vertices $s_3^{(i,j)}$ and $v_{b,i}^{(j)}$ with an edge at time $y_{i,j} + 2h-1$.

Intuitively, the time labels on the edges and the waiting time restrictions enforce that when arriving at $s_1^{(i,j)}$ there is only one way to continue to $s_2^{(i,j)}$ for which is it necessary to visit a vertex in the vertex selection gadget that corresponds to an endpoint of the selected edge. Similarly, from $s_2^{(i,j)}$ there is only one way to continue to $s_3^{(i,j)}$ for which it is necessary to visit a vertex in the vertex selection gadget that corresponds to the other endpoint of the selected edge. For a visualization of the validation gadget see \cref{fig:pathred3}, where the red dashed path corresponds to the selection of an edge.

Now we connect the gadgets for all $F_{i,j}$'s in sequence, that is, a $\Delta$-restless \nonstrpath{s,z} passes through the gadgets one after another, selecting one edge of each part $F_{i,j}$ of the edge set $F$. Formally, for each $i,j\in[k]$ with $i<j$, if $i< j-1$, we connect vertices~$v_{m,4}^{(i,j)}$ and~$v_{1,1}^{(i+1,j)}$ with an edge at time $y_{i+1,j}$, and if $i= j-1<k-1$, we connect vertices~$v_{m,4}^{(i,j)}$ and~$v_{1,1}^{(i,j+1)}$ with an edge at time $y_{i,j+1}$. 
It is easy to check that after the removal of $3\cdot \binom{k}{2}$ many vertices $\{s^{(1,2)}_1, s^{(1,2)}_2, s^{(1,2)}_3, s^{(1,3)}_1, \ldots, s^{(1,k)}_3, \ldots s^{(k-1,k)}_3\}$, the validation gadgets are a set of disjoint paths, see \cref{fig:pathred3}.

Finally, we create two new vertices $s$ and $z$, 
we connect vertices $s$ and $w_1^{(1)}$ (the ``first'' vertex of the vertex selection gadgets) with an edge at time one, 
we connect vertices $s$ and $s^{(1)}$ (the ``skip vertex'' of the first vertex selection gadget) with an edge at time one, 
and we connect $z$ and $v_{m,4}^{(k-1,k)}$ (the ``last'' vertex of the validation gadgets) with an edge at time $k\cdot n+m\cdot (3k^2+5k+3)$. 
We further connect vertices~$w_{n+1}^{(k)}$ and~$v_{1,1}^{(1,2)}$ 
(connecting the vertex selection gadgets and the validation gadgets) with an edge at time $k\cdot n$. 
Finally, we set~$\Delta=1$. This completes the construction. 
It is easy to check that $\TG$ can be constructed in polynomial time 
and that the distance to disjoint paths of $\UG$ is at most $k+3\cdot\binom{k}{2}$ and that every edge has only one time stamp. 

\begin{figure}[t]
\begin{center}
\begin{tikzpicture}[line width=1pt, scale=0.35]
    \node[vert2] (A11) at (0,0) {}; 
    \node[vert2] (A12) at (0,-2) {}; 
    \node[vert2] (A21) at (0,-4) {}; 
    \node[vert2] (A22) at (0,-6) {}; 
    \node[vert2] (A31) at (0,-8) {}; 
    \node[vert2] (A32) at (0,-10) {}; 
    \node[vert2] (A41) at (0,-12) {}; 
    \node (A41l) at (-2,-12) {$v_{h,1}^{(i,j)}$}; 
    \node[vert2] (A42) at (0,-14) {}; 
    \node[vert2] (Am1) at (0,-20) {}; 
    \node[vert2] (Am2) at (0,-22) {}; 
    
    \node[vert3,fill=lipicsyellow!80] (S1) at (8,-11) {}; 
    \node[vert3,fill=lipicsyellow!80] (S2) at (14,-11) {}; 
    \node[vert3,fill=lipicsyellow!80] (S3) at (20,-11) {}; 
    \node (S1l) at (9,-12.6) {$s_1^{(i,j)}$}; 
    \node (S2l) at (14,-12.6) {$s_2^{(i,j)}$}; 
    \node (S3l) at (19,-12.6) {$s_3^{(i,j)}$}; 
    
    \node[vert2] (B11) at (28,0) {}; 
    \node[vert2] (B12) at (28,-2) {}; 
    \node[vert2] (B21) at (28,-4) {}; 
    \node[vert2] (B22) at (28,-6) {}; 
    \node[vert2] (B31) at (28,-8) {}; 
    \node[vert2] (B32) at (28,-10) {}; 
    \node[vert2] (B41) at (28,-12) {}; 
    \node (B41l) at (30,-12) {$v_{h,3}^{(i,j)}$}; 
    \node[vert2] (B42) at (28,-14) {}; 
    \node[vert2] (Bm1) at (28,-20) {}; 
    \node[vert2] (Bm2) at (28,-22) {}; 
    
    \draw (A11) -- (A12);
    \draw (A12) -- (A21);
    \draw (A21) -- (A22);
    \draw (A22) -- (A31);
    \draw (A31) -- (A32);
    \draw (A32) -- (A41);
    \draw (A41) -- (A42);
    \draw[loosely dotted] (A42) -- (0,-15.6);
    \draw[loosely dotted] (Am1) -- (0,-18.4);
    \draw (Am1) -- (Am2);
    
    \draw (B11) -- (B12);
    \draw (B12) -- (B21);
    \draw (B21) -- (B22);
    \draw (B22) -- (B31);
    \draw (B31) -- (B32);
    \draw (B32) -- (B41);
    \draw (B41) -- (B42);
    \draw[loosely dotted] (B42) -- (28,-15.6);
    \draw[loosely dotted] (Bm1) -- (28,-18.4);
    \draw (Bm1) -- (Bm2);

    \draw (A11) -- (S1);
    \draw (A21) -- (S1);
    \draw (A31) -- (S1);
    \draw[color=ourred,dashed] (A41) -- (S1);
    \draw (Am1) -- (S1);
    
    \draw (B11) -- (S3);
    \draw (B21) -- (S3);
    \draw (B31) -- (S3);
    \draw[color=ourred,dashed] (B41) -- (S3);
    \draw (Bm1) -- (S3);

   \node[vert2] (Vij) at (9,-2) {};
   \node[vert2] (Vji) at (19,-2) {}; 
   \node (Vijl) at (7.4,-3.6) {$v_{a,j}^{(i)}$};
   \node (Vjil) at (20.6,-3.6) {$v_{b,i}^{(j)}$}; 
   
    \draw[color=ourred, dashed] (Vij) -- (S1);
    \draw[color=ourred, dashed] (Vij) -- (S2);
    \draw[color=ourred, dashed] (Vji) -- (S2);
    \draw[color=ourred, dashed] (Vji) -- (S3);
    
	\node (Uil) at (10,1) {$U_i$};
	\node (Ujl) at (18,1) {$U_j$};    
    
    \draw (Vij) -- (13,-2) -- (10,0) -- (7,-2) -- (Vij);
    \draw (Vji) -- (21,-2) -- (18,0) -- (15,-2) -- (Vji);
    \draw[loosely dotted] (13,-2) -- (15,-2);
    \draw[loosely dotted] (7,-2) -- (6,-2);
    \draw[loosely dotted] (21,-2) -- (22,-2);
    
    \draw[loosely dotted] (A11) -- (-2,0);
    \draw[loosely dotted] (Bm2) -- (30,-22);
\end{tikzpicture}
    \end{center}
	\caption{Visualization of the validation gadget for $F_{i,j}$ from the reduction of \cref{thm:probRestlessPath:W1hFVS}. The ``first path'' of the gadget is depicted vertically on the left, the ``second path'' on the right. The connections to the vertex selection gadgets for the edge $e_h^{(i,j)}=\{u_a^{(i)},u_b^{(j)}\}\in F_{i,j}$ are depicted. The edges in red (dashed) correspond to the path through the gadget if edge $e_h^{(i,j)}$ is ``selected'' and all these edges have the same time stamp. The vertex selection gadgets corresponding to~$U_i$ and $U_j$ are depicted as triangles in the upper center part. 
			The three vertices $s_1^{(i,j)}$, $s_2^{(i,j)}$, and $s_3^{(i,j)}$ are colored yellow (squared). 
	Note that after the removal of $s_1^{(i,j)}$, $s_2^{(i,j)}$, and $s_3^{(i,j)}$, 
	the validation gadget for $F_{i,j}$ is a set of disjoint paths.}\label{fig:pathred3}
\end{figure}
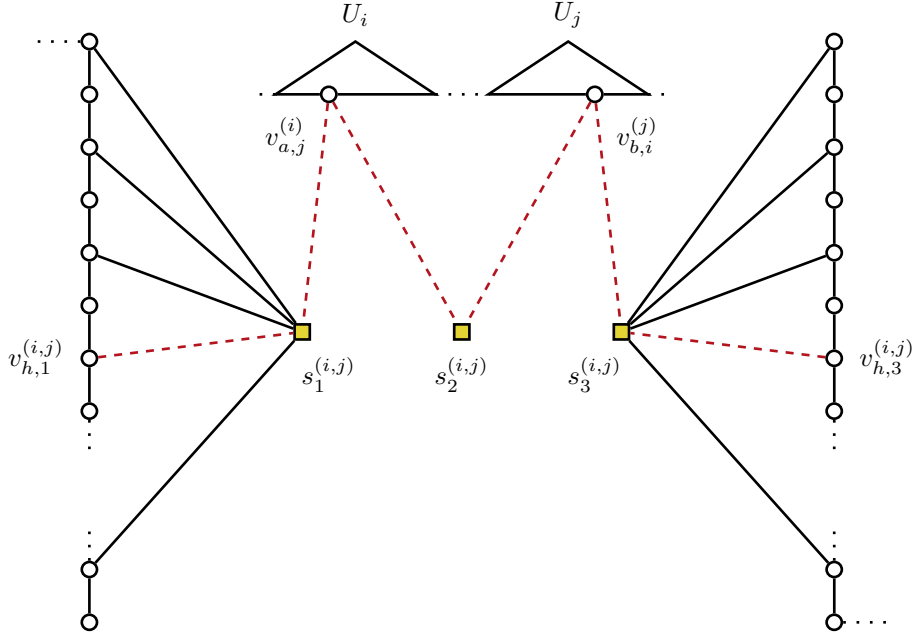

\proofpar{Correctness.} 
Now we show that $H$ contains a clique of size $k$ if and only if there is a $\Delta$-restless temporal path from $s$ to $z$ in~$\TG$.

\RArrow Assume that $H$ contains a clique of size $k$ and let $X\subseteq V(H)$ with $|X| = k$ be the set of vertices that form the clique in $H$. Now we show how to construct a $\Delta$-restless \nonstrpath{s,z} in $\TG$. Note that since $H$ is $k$-partite, we have that~$|U_i\cap X|=1$ for all $i\in [k]$. The temporal path starts at vertex $s$ in $\TG$ and then first passes through the vertex selection gadgets. If $u_j^{(i)}\in X$ for some~$i\in [k]$ and $j\in [n]$, then the temporal path skips the segment corresponding to $u_j^{(i)}$ in the vertex selection gadget for $U_i$. More formally, the temporal path follows the vertices $w_1^{(i)},v_{1,1}^{(i)}, v_{1,2}^{(i)}, \ldots,v_{1,k}^{(i)}, w_2^{(i)},\ldots, v_{j-1,k}^{(i)}, w_{j}^{(i)}, $ $ s^{(i)}, w_{j+1}^{(i)}, v_{j+1,1}^{(i)}, \ldots, v_{n,k}^{(i)}, w_{n+1}^{(i)}$ in that order, 
that is, the path skips vertices $v_{j,1}^{(i)}, v_{j,2}^{(i)}, \ldots, v_{j,k}^{(i)}$. It is easy to check that the time labels of the edges in the vertex selection gadget allow for a restless temporal path as described that respects the waiting time $\wait$.

In the validation gadget for $F_{i,j}$ with $i<j$, the path ``selects'' the edge $(U_i\cap X)\cup(U_j\cap X)\in F_{i,j}$ that connects the vertices from the parts $U_i$ and $U_j$ that are contained in the clique~$X$. Let $(U_i\cap X)\cup(U_j\cap X)=\{u_a^{(i)},u_b^{(j)}\}=e_h^{(i,j)}\in F_{i,j}$. Formally, the path follows vertices $v_{1,1}^{(i,j)}, v_{1,2}^{(i,j)}, v_{2,1}^{(i,j)}, v_{2,2}^{(i,j)}, \ldots, v_{h,1}^{(i,j)}, s_1^{(i,j)},v_{a,j}^{(i)},$ $s_2^{(i,j)},v_{b,i}^{(j)}, s_3^{(i,j)},v_{h,4}^{(i,j)}, v_{h+1,3}^{(i,j)}, v_{h+1,4}^{(i,j)}, \ldots, v_{m,4}^{(i,j)}$ in that order. Note that vertices $v_{a,j}^{(i)}$ and $v_{b,i}^{(j)}$ have not been used by the path in the vertex selection gadgets, since they appear in the segments that were skipped by the temporal path in the corresponding vertex selection gadgets. Furthermore, since the clique in~$H$ only contains one edge that connects vertices from $U_i$ and $U_j$, the vertices $v_{a,j}^{(i)}$ and $v_{b,i}^{(j)}$ have not been used by the temporal path in an earlier validation gadget. It is easy to check that the time labels of the edges in the validation gadget allow for a $\Delta$-restless temporal path as described. After the last validation gadget the path arrives at vertex~$z$. Hence, we have found a $\Delta$-restless \nonstrpath{s,z} in~$\TG$.

\LArrow Assume that we are given a $\Delta$-restless \nonstrpath{s,z} in $\TG$. We now show that $H$ contains a clique of size $k$.

After starting at $s$, the $\Delta$-restless temporal path first passes the vertex selection gadgets. 
Here, we need to make the important observation, that for each $i\in[k]$, any $\Delta$-restless \nonstrpath{s,z} has to ``skip'' at least one segment corresponding to one vertex $u^{(i)}_j\in U_i$ in the vertex selection gadget corresponding to $U_i$, otherwise the temporal path cannot traverse the validation gadgets.
More formally, assume for contradiction that there is a $\Delta$-restless \nonstrpath{s,z} and an $i\in[k]$ such that the temporal path visits all vertices in the vertex selection gadget corresponding to $U_i$. Let $j\in[k]$ with $j\neq i$. Assume that $i<j$ (the other case works analogously). We claim that the temporal path cannot traverse the validation gadget for $F_{i,j}$. 
For the temporal path to go from $s^{(i,j)}_1$ to $s^{(i,j)}_2$ by construction it has to visit at least one vertex from the vertex selection gadget for $U_i$. If all vertices have already been visited, that would mean the $\Delta$-restless \nonstrpath{s,z} visits one vertex twice---a contradiction. 

The waiting time $\wait$ prevents the temporal path from ``skipping'' more than one segment. More formally, any $\Delta$-restless \nonstrpath{s,z} arrives at the ``skip vertex''~$s^{(i)}$ of the vertex selection gadget for $U_i$ at time $(i-1)\cdot n + j$, for some $j\in[k-1]$. By construction this means the path visits $w_j^{(i)}$, then $s^{(i)}$, and then has to continue with $w_{j+1}^{(i)}$ since there is only one time edge the path can use without violating the waiting time $\Delta$. It follows that the temporal path skips exactly the segment corresponding to $u_j^{(i)}\in U_i$.

This implies that any $\Delta$-restless \nonstrpath{s,z} that traverses the vertex selection gadgets leaves exactly one segment of every vertex selection gadget unvisited. 
Let the set $X = \{u^{(i)}_j\in U_i \mid i\in[k] \wedge j\in[n] \wedge v_{j,1} \text{ is an unvisited vertex.}\}$ be the set of vertices corresponding to the segments that are ``skipped'' by the given $\Delta$-restless \nonstrpath{s,z}. It is easy to check that $|X|=k$. We claim that $X$ is a clique in $H$.

Assume for contradiction that it is not. Then there are two vertices $u_{i'}^{(i)}, u_{j'}^{(j)}\in X$ such that the edge $\{u_{i'}^{(i)}, u_{j'}^{(j)}\}$ is not in $F$. Assume that $i<j$. We show that then the $\Delta$-restless \nonstrpath{s,z} is not able to pass through the validation gadget for~$F_{i,j}$. By assumption we have that $\{u_{i'}^{(i)}, u_{j'}^{(j)}\}\notin F_{i,j}$. Note that the validation gadget is designed in a way that the first path ``selects'' an edge from $F_{i,j}$ and then the waiting time of one enforces that a $\Delta$-restless \nonstrpath{s,z} can only move from the first path to the second path of a validation gadget if the two endpoints of the selected edge are vertices whose corresponding segments in the vertex selection gadget were skipped. We have seen that for every $U_i$ with $i\in[k]$, the path segment corresponding to exactly one vertex of that set was skipped. Since $\{u_{i'}^{(i)}, u_{j'}^{(j)}\}\notin F_{i,j}$, we have that for every edge in $F_{i,j}$ that the segment corresponding to at least one of the two endpoints of the edge was \emph{not} skipped. Hence, we have that the $\Delta$-restless temporal path cannot pass through the validation gadget of $F_{i,j}$ and cannot reach~$z$---a contradiction. 
\end{proof}

\section{An \FPT-algorithm for short restless temporal path}\label{sec:fptalg}

In this section, we discuss how to find \emph{short} restless temporal paths.
Recall that in \probSRestlessPath, we are given an additional integer $k$ as
input and are asked whether there exists a $\Delta$-restless \nonstrpath{s,z}
that uses at most $k$ time edges. 
By \cref{thm:probRestlessPath:NPh} this problem is \NP-hard. 
Note that in the contact tracing scenario from the beginning,
we can expect to have a small $k$ and a large temporal graph.
\begin{theorem}
	\label{thm:fpt-length}
	\probSRestlessPath is 
	\begin{enumerate}[\normalfont\bfseries\color{darkgray}(i),leftmargin=0.9cm] %
			\item solvable in $2^k\cdot |\TG|^{O(1)}$~time with a constant one-side error\footnote{The algorithm always outputs \emph{no} if there is no $\Delta$-restless temporal $(s,z)$-path and outputs otherwise \emph{yes} with constant probability.},
			\item deterministically solvable in ${2}^{O(k)}\cdot|\TG|\wait$~time,
	\end{enumerate}
\end{theorem}
Note that we can solve \probSRestlessPath{} such that the running time is independent 
from the lifetime $\lifetime$ of the temporal graph.
To show \cref{thm:fpt-length}, we first reduce the problem to a specific path problem in directed graphs.
Then, we apply known algebraic tools for multilinear monomials detection.
Here, \cref{thm:fpt-length} (i) %
is based on Williams \cite{W09}. %
To get a deterministic algorithm with a running time almost linear in $|\TG|$, %
we show a different approach based on representative sets \cite{FLPS16} which results in \cref{thm:fpt-length} (ii).

\ownparagraph{Reduction to directed graphs.}
We introduce a so-called \emph{$\wait$-$(s,z)$-expansion} for two vertices $s$ and $z$ of a temporal graph with waiting times.
That is, a time-expanded version of the temporal graph 
which reduces reachability questions to directed graphs.
While similar approaches have been applied several times
\cite{Akr+19a,Ber96,MOS19,Wu+16,Zsc+19}, to the best of our knowledge, 
this is the first time that waiting-times are considered.
In a nutshell, the $\wait$-$(s,z)$-expansion has  
for each vertex $v$ at most $\lifetime$ 
many copies $v^1,\dots, v^{\lifetime}$ %
and if an \dipath{} visits $v^i$, it means that the corresponding \twalk{} visits $v$ at time $i$. 
\begin{definition}[$\wait$-$(s,z)$-Expansion]
	\label{def:expansion}
	Let $\TGcompact$ be a temporal graph with two distinct vertices $s,z \in V$ such that $\{s,z\} \not \in E_t$, for all $t \in [\lifetime]$. Let $\wait\le\lifetime$.
The \emph{$\wait$-$(s,z)$-expansion} of $\TG$ is the directed graph $D=(V',E')$ with
\begin{enumerate}[(i)]%
		\item $V' := \{s,z\} \cup \left\{ v^{t}\ \middle\vert\ v \in e, e \in E_t, v \not \in \{ s, z \} \right\}$,
		\item $E_s := \left\{ (s,v^t) \ \middle\vert\   \{s,v\} \in E_t \right\}$,
		\item $E_z := \left\{(v^i,z) \ \middle\vert\ v^i \in V',  \{v,z\} \in E_t, 0 \leq t-i \leq \wait  \right\}$, and
		\item $E' := E_s \cup E_z \cup \left\{(v^i,w^t) \ \middle\vert\ v^i \in V' \setminus \{ s, z \}, \{v,w\} \in E_t, 0 \leq t-i \leq \wait  \right\}$.
\end{enumerate}
Furthermore, we define $V'(s) := \{ s \}$, $V'(z) := \{ z \}$, and $V'(v) := \{ v^{t} \in V' \mid t \in [\lifetime] \}$, for all~$v \in V \setminus \{s,z\}$.
\end{definition}
Next, we show that a $\wait$-$(s,z)$-expansion of a temporal graph can be computed efficiently.
\begin{lemma}
	\label{lem:exp-runtime}
	Given a temporal graph $\TGcompact$, two distinct vertices $s,z\in V$, and $\wait\le\lifetime$, 
	we can compute its~$\wait$-$(s,z)$-expansion $D$ with $|V(D)| \in O(|\TG|)$ in $O(|\TG|\cdot \wait)$ time.
\end{lemma}
\begin{proof}
	Let $V' := \{s,z\}$ and $E'$ be empty in the beginning.
	We will fill up $V'$ and $E'$ simultaneously.
	In order to do that efficiently, we will maintain for each vertex $v\in V$ a ordered list $L_v$ 
	such that $t \in L_v$ if and only if $v^t \in V'$.
	We assume that $|V| \leq \sum_{i=1}^\lifetime|E_i|$, because vertices which are isolated in every layer are irrelevant 
	for the $\wait$-$(s,z)$-expansion and can be erased in linear time.
	
	We proceed as follows.
	For each~$t \in \{1,\dots,\lifetime\}$ (in ascending order), we iterate over~$E_t$.
	For each~$\{v,w\} \in E_t$, we distinguish three cases.
	\begin{description}
	\item[($w=s$):] We add $v^t$ to $V'$, $(s,v^t)$ to $E'$, and add $t$ to $L_v$.
		This can be done in constant time.
	\item[($w=z$):] We add $v^t$ to $V'$, and add $t$ to $L_v$.	
		Now we iterate over all~$i \in L_v$ (in descending order)
		and add $(v^i,z)$ to $E'$
		until $ t - i > \wait$.
		This can be done in $O(\wait)$ time.
	\item[($\{s,z\} \cap \{v,w\} = \emptyset$):] We add $v^t,w^t$ to $V'$, and add $t$ to $L_v$ and $L_w$.
		Now we iterate over~$i\in L_v$ (in descending order) 
		and add $(v^i,w^t)$ to $E'$
		until $ t - i > \wait$.
		Afterwards, we iterate over~$i\in L_w$ (in descending order) 
		and add $(w^i,v^t)$ to $E'$
		until $ t - i > \wait$.
		This can be done in $O(\wait)$ time.
	\end{description}
	
	Observe that after this procedure the digraph $D=(V',E')$ is the $\wait$-$(s,z)$-expansion of $\TG$ and
	that we added at most $2$ vertices for each time-edge in $\TG$.
	Hence, $V' \leq |\TG|$.
	This gives a overall running time of $O(|\TG|\cdot\wait)$.
\end{proof}
It is easy to see that there is a \twalk{} in the temporal graph 
if and only if there is an \dipath{} in the $\wait$-$(s,z)$-expansion.
Next, we identify the necessary side constraint to identify \tpath{}s in the $\wait$-$(s,z)$-expansion.
\begin{lemma}\label{lem:exp-correct}
	Let $\TGcompact$ be a temporal graph, $s,z \in V$ two distinct vertices, $\wait\le\lifetime$, and $D=(V',E')$ the $\wait$-$(s,z)$-expansion of $\TG$.
	There is a \tpath{} in $\TG$ of length~$k$ 
	if and only if there is an \dipath{} $P'$ in $D$ of length $k$ 
	such that for all $v \in V$ it holds that~$|V'(v) \cap V(P')| \leq 1$.
\end{lemma}
\begin{proof}
	\RArrow
	Let $P = \big( ((s,v_1, t_1), (v_1, v_2, t_2), \dots, (v_{k'-1},z, t_{k'} ) \big)$ 
	be a \tpath{} in~$\TG$ of length $k$.
	We can inductively construct an \dipath{} $P'$ in $D$.
	Observe that $P_1' := ((s,v_1^{t_1}))$ is an \dipath[s,v_1^{t_1}]{} of length $1$ in $D$, 
	because the arc $(s,v_1^{t_1})$ is in $E_s$ of $D$.
	Now let $i \in [k'-2]$ and $P'_i$ be an \dipath[s,v_i^{t_i}]{} of length $i$ such that 
	\begin{enumerate}[(i)]
		\item for all $j \in [i]$, we have that $|V'(v_j) \cap V(P'_i)| = 1$, and
		\item for all $v \in V \setminus \{s,v_1,\dots,v_i\}$, we have that $|V'(v) \cap V(P'_i)| = 0$.
	\end{enumerate}
	In order to get an \dipath[s,v_{i+1}^{t_{i+1}}]{} $P'_{i+1}$ of length $i+1$, we extend $P'_{i}$ by the arc $(v_i^{t_i},v_{i+1}^{t_{i+1}})$.
	Observe, that $v_{i+1}^{t_{i+1}} \in V'$ because of the time-edge $(\{v_i,v_{i+1}\},t_{i+1})$ in $\TG$
	and that the arc $(v_i^{t_i},v_{i+1}^{t_{i+1}}) \in E'$, because we have $0 \leq t_{i+1} - t_i \leq \Delta$.
	Observe that 
	\begin{enumerate}[(i)]
		\item for all $j \in [i+1]$, we have that $|V'(v_j) \cap V(P'_{i+1})| = 1$, and 
		\item for all $v \in V \setminus \{s,v_1,\dots,v_{i+1}\}$, we have that $|V'(v) \cap V(P'_{i+1})| = 0$.
	\end{enumerate}
	
	Hence, we have an \dipath[s,v_{k'-1}^{t_{k'-1}}]{} $P'_{k-1}$ of length $k-1$ satisfying (i) and (ii) 
	which can be extended (in a similar way) to an \dipath{} of length $k$
	such that for all $v \in V$ it holds that $|V'(v) \cap V(P')| \leq 1$.

	\LArrow
	Let $P'$ be a \dipath{} in $D$ of length $k$ 
	such that for all $v \in V$ it holds that $|V'(v) \cap V(P')| \leq 1$.
	Let $V(P') = \{ s,v_1^{t_1},\dots,v_{k-1}^{t_{k-1}},z\}$.
	Observe that an arc from $s$ to $v_1^{t_1}$ in $D$ implies that there is a time-edge $(\{s,v_1\},t_1)$ in $\TG$.
	Similarly, an arc from $v_i^{t_i}$ to $v_{i+1}^{t_{i+1}}$ implies that there is a time-edge $(\{v_i,v_{i+1}\},t_{i+1})$ in $\TG$
	and that $0 \leq t_{i+1} - t_i \leq \Delta$, for all $i \in [k-2]$.
	Moreover, an arc from $v_{k-1}^{t_{k-1}}$ to $z$ implies that there is some $t_{k}$ such that 
	there is a time-edge $(\{v_k,z\},t_{k})$ in~$\TG$ with $0 \leq t_{k} - t_{k-1} \leq \Delta$.
	Hence, $P = \big( (s,v_1, t_1), (v_1, v_2, t_2), \dots, (v_{k'-1},z, t_{k'} ) \big)$ is a \twalk{} of length $k$ in $\TG$.
	Finally, $|V'(v) \cap V(P')| \leq 1$, for all $v \in V$, implies that $v_i \not = v_j$ for all $i,j \in \{0,\dots,k\}$ with $i \not = j$.
	Thus, $P$ is a \tpath{} of length $k$.
\end{proof}

\ownparagraph{Obtaining \cref{thm:fpt-length} (i).}
We now adapt the algorithm of Williams \cite{W09} to our
specific needs.
To this end, we introduce some standard notation from algebraic theory.

An \emph{arithmetic circuit} $C$ over a commutative ring $R$ is a simple labelled directed acyclic graph 
with its internal nodes labeled by $+$ (sum gates) or $\times$ (product gates) and 
its nodes of in-degree zero (input gates) labeled with elements from $R\cup X$, 
where $X$ is a set of variables. 
There is one node of out-degree zero, called the output gate.
The size of $C$ is the number of vertices in the graph.
An arithmetic circuit $C$ over $R$ computes a polynomial $P(X)$ over $R$ in the natural way:
an input gate represents the polynomial it is labeled by. 
A sum (product) gate represents the sum (product) of the polynomials represented by its in-degree neighbors.
We say $C$ represents $P(X)$ if the polynomial of the output gate of $C$ is equivalent to $P(X)$.

\begin{lemma} \label{lem:the-polynomial}
		Let $k \in \N$ and $D=(V,A)$ be a directed graph with partition $V=\biguplus_{i=0}^n V_i$, where $V_0 = \{s\}$ and $V_n=\{z\}$.
		Then,
		there is an arithmetic circuit $C$ representing a polynomial $Q(X)$ of degree
		at most $k+1$ such that $Q(X)$ has a multilinear\footnote{No variable occurs to a power of two or higher.} monomial of degree at most $k+1$ if and only if there is an \npath{s,z} $P$ of length at most $k$ in $D$ where $|V(P) \cap V_i| \leq 1$ for all $i \in [n]$.
		Moreover, $|X|=n+1$, $C$ is of size $O(k(n+|A|))$, has no scalar multiplication, and all product gates in have in-degree two.
\end{lemma}
The idea of the polynomial is similar to the one of Williams \cite{W09}, 
but here instead of having one variable for each vertex we just have one variable for all vertices in one part of the partition of $V$.
\begin{proof}
We define the polynomial recursively
as $Q(X) = x_\bot Q_z^k$ over variables $X= \{ x_\bot, x_0 \} \cup \{x_i \mid i \in [n] \}$, where
\begin{align}
		\label{eq:polynomial}
		& Q_s^0 := x_0,\nonumber\\ 
		\forall v \in V \setminus \{s\} &\colon Q_v^0 := x_\bot, &\text{and}\\
  \forall v \in V,\forall j \in [k] &\colon 
  Q_v^j := \sum_{(u,v)\in A} Q_u^{j-1} x_i, &\text{ where } v \in V_i. \nonumber
\end{align}
Note that we can, by simply following  \eqref{eq:polynomial}, construct an
arithmetic circuit $C$ which represents $Q(X)$ in $O(k(n+|A|))$ time such that
each product gate has an in-degree of two.
Furthermore, observe that $Q(X)$ has no scalar multiplication and is of degree
at most $k+1$.

The following induction completes the proof:
We claim that for all $v\in V$ and $j \in [k] \cup \{0\}$,
$Q(X)$ has a multilinear monomial $M$ of degree at most $j+1$
if and only if
there is an \npath{s,v} $P$ of length at most $j$ in $D$ where $|V(P) \cap V_i| \leq 1$ for all $i \in [n]$.
Moreover, $M$ contains the variable $x_i$ if and only if $|V(P) \cap V_i| = 1$, for all $i \in [n]$.
Is is easy to verify that the claim is true for $j=0$.

Now assume as induction hypothesis that for all $u \in V$ and all $j'< j \in [k]$, 
the polynomial $x_\bot Q_u^{j'}$ has a multilinear monomial $M$ of degree at most $j'+1$ 
if and only if
there is an \npath{s,u} $P$ of length at most $j'$ in $D$ where $|V(P) \cap V_i|
\leq 1$, for all $i \in [n]$.
Moreover, $M$ contains the variable $x_i$ if and only if $|V(P) \cap V_i| = 1$, for all $i \in [n]$.
Let $v \in V_p$.

\RArrow
Assume there is a multilinear monomial $M$ of degree at most $j+1$ in $x_\bot Q_v^j$.
Since $x_\bot Q_v^j = \sum_{(u,v)\in A} x_p (x_\bot Q_u^{j-1})$, 
we know that $M$ contains $x_p$ and there is a $(u,v) \in A$ such that $x_\bot Q_u^{j-1}$ contains a multilinear monomial $M'$ which does not contain $x_p$.
By induction hypothesis, there is an \npath{s,u} $P'$ of length at most $j-1$
such that $|V(P') \cap V_i| = 1$ if and only if $M'$ contains $x_i$ for all $i \in [n]$.
Hence, there is an \npath{s,v} $P$ (obtained by extending $P'$ with $v$) such
that $|V(P) \cap V_i| \leq 1$ for all $i \in [n]$.
Furthermore, we have that $|V(P) \cap V_i|=1$ if and only if $M$ contains $x_i$ for all $i \in [n]$.

\LArrow
Assume there is an \npath{s,v} $P$ of length at most $j$ in $D$ where $|V(P) \cap V_i| \leq 1$ for all $i \in [n]$.
Let $P'$ be the \npath{s,u} obtained by removing $v$ from $P$.
Hence, $P'$ is of length at most $j-1$, and $|V(P) \cap V_i| \leq 1$ for all $i \in [n]$.
By induction hypothesis $x_\bot Q_u^{j-1}$ contains a multilinear monomial $M$ of degree at most $j$
which does not contain $x_p$.
Since $x_\bot Q_v^j = \sum_{(u,v)\in A} x_p (x_\bot Q_u^{j-1})$, we know that $x_\bot Q_v^j$ contains a 
$M$ multiplied by $x_p$ as monomial.
Thus, $x_\bot Q_v^j$ has a multilinear monomial of degree at most $j+1$ which contains variable $x_i$ if and only if 
$|V(P') \cap V_i| = 1$, for all $i \in [n]$.
\end{proof}
Now we can apply the following result of Williams \cite{W09}.
\begin{theorem}[\cite{W09}]
		\label{thm:rand-blackbox}
		Let $Q(X)$ be a polynomial of degree at most $k$, 
		represented by an arithmetic circuit of size $n$ with no scalar multiplications and
		where all product gates have in-degree two.
		There is a randomized algorithm that  
		runs in $2^k n^{O(1)}$ time,
		outputs \yes{} with high probability ($\geq \nicefrac{1}{5}$) if there is a multilinear term
		in the sum-product expansion of $Q$, 
		and always outputs \no{} if there is no multilinear term.
\end{theorem}
\cref{thm:fpt-length} (i) follows from
\cref{lem:exp-runtime,lem:exp-correct,lem:the-polynomial,thm:rand-blackbox}.
This can be derandomized by Theorem 5.2 of Fomin et al.~\cite{FLPS17} resulting in $O({3.841}^{k}\cdot(|\TG|\wait)^2|V|\log|V|)$~time algorithm.
We now show how to improve the polynomial part of a deterministic algorithm.

\ownparagraph{Obtaining \cref{thm:fpt-length} (ii).}
To show \cref{thm:fpt-length} (ii), we first note that in the $(s,z)$-expansion of an $(s,z)$-path $P$ in the directed graph describes a \tpath{}
exactly when $V(P)$ is an \emph{independent set of some specific matroid}.
We then show an algorithm to find such a path $P$ (if there is one).
To this end, we introduce a problem, \IPP, and some standard terminology from matroid theory~\cite{Oxl92}.
A pair $(U,\I)$, where $U$~is the \emph{ground set}
and $\I\subseteq 2^U$ is a family of \emph{independent sets},
is a \emph{matroid} if the following holds:
$\emptyset \in \I$;
if $A' \subseteq A$ and $A \in \I$, then $A' \in \I$; and
if $A,B \in \I$ and $|A| < |B|$, then there is an $x \in B \setminus A$ such that $A \cup \{x\} \in \I$.
An inclusion-wise maximal independent set~$A\in \I$
of a matroid~$M=(U,I)$ is a \emph{basis}.
The cardinality of the bases of~$M$
is called the \emph{rank} of~$M$.
The \emph{uniform matroid of rank~$r$} on $U$
is the matroid~$(U,\I)$
with $\I=\{S\subseteq U\mid |S|\leq r\}$.
A matroid~$(U,\I)$ is \emph{linear} or \emph{representable over a field $\mathbb F$}
if there is a matrix~$A$ with entries in $\mathbb F$ and the columns labeled by the elements of~$U$
such that $S \in \I$ if and only if the columns of~$A$ 
with labels in~$S$ are linearly independent over~$\mathbb F$. 
Such a matrix $A$ is called a \textit{representation} of $(U,\I)$.
Now we are ready to state the \IPP\ problem. 

 \probDef{\IPP}
 { A digraph $D=(V,E)$, two distinct vertices~$s,z\in V$, a representation $A_M$ of a matroid $M = (V,\I)$ of rank $r$ over a finite field $\mathbb F$.} %
 {Is there an $(s,z)$-dipath $P$ of length at most $k$ in $D$ such that $V(P) \in \I$?}

For the remainder of this section,
whenever we speak about \emph{independent sets}, these are independent sets of a 
matroid and not a set of vertices which induce an edgeless graph.

Agrawal~et~al.~\cite{AJKS19} studied, independently from us, a similar problem 
where the edges of the path shall be an independent set of a matroid.
To show \cref{thm:fpt-length} (ii), 
we need a single-exponential algorithm which has only a linear dependency on the input size.
To this end, we show the following, based on representative families.
\begin{theorem}
	\label{thm:fpt-ipp}
	An instance $(D,s,z,A_M)$ of \IPP{} 
	can be solved in time of $O(2^{\omega r}m)$ 
	operations over the field $\mathbb F$, 
	where   
	$\mathbb F$ is the field of $A_M$,
	$r$ is rank of $M$, 
	$m$ is the number of edges in $D$, and
	$2<\omega<2.373$ is an upper-bound for the matrix multiplication exponent\footnote{Note that we require $2 < \omega$ even though this might be not true. We do this to upper-bound the polynomial part in $r$.
			The bound $\omega < 2.373$ is known \cite{alman2021refined}.
	}.
\end{theorem}
\label{sec:fptforipp}

In this section, we provide a fixed-parameter algorithm for \IPP{} parameterized by rank $r$ of the matroid.
Since the rank $r$ is at most $|V(D)|$, this algorithm is asymptotically optimal, see \cref{cor:ipp-hard}.
To show \cref{thm:fpt-ipp}, we provide an algorithm (\cref{alg:fpt-ipp}),
show its correctness (\cref{lem:fpt-ipp-correct}), 
and prove the running time upper-bound (\cref{lem:fpt-ipp-time}).
The idea of our algorithm is based on the algorithm 
of Fomin~et~al.~\cite{FLPS16} for \textsc{$k$-Path}
and independently from us Agrawal~et~al.~\cite{AJKS19} 
showed an algorithm which runs in $2^{O(r)} n^{O(1)}$ time for \IPP{}
and Lokshtanov~et~al.~\cite{LMPSZ18} provided a dynamic program, 
running in $5.18^{r} n^{O(1)}$ time,
for the special case of \IPP{} when the matroid given in the input is 
a transversal matroid.
However, in contrast to Agrawal~et~al.~\cite{AJKS19} and Lokshtanov~et~al.~\cite{LMPSZ18},
we pay attention to the detail that the algorithm behind \cref{thm:fpt-ipp} 
runs in linear time, if we can perform one field operation in constant time.

The main tool of our algorithm are representative families of independent sets.
\begin{definition}[Representative family]
	\label{def:qrep}
  Given a matroid $(U,\I)$,
  and a~family~$\mathcal S \subseteq 2^U$,
  we say that a subfamily~$\widehat{\mathcal S} \subseteq \mathcal S$
  is a \emph{$q$-representative
    for $\mathcal S$}
  if,
  for each set~$Y \subseteq U$ of size at most~$q$,
  it holds that:
  \begin{itemize}
  \item if there is a set~$X \in \mathcal S$  with
    $X \uplus Y \in \I$, 
  \item then there is a
	  set~$\widehat X \in \widehat{\mathcal S}$ with
  $\widehat X \uplus Y \in \I$.
  \end{itemize}
\end{definition}

A $p$-family is a family $\mathcal F$ such that each set $S \in \mathcal F$ is of size exactly $p$.
For linear matroids, we can compute small representative families efficiently. Formally, the following is known.
\begin{theorem}[{Fomin et al.~\cite[Theorem~1.1]{FLPS16}}]
	\label{thm:eff-rep-fam}
  Let $M = (U,\I)$~be a linear matroid
  of rank~$r=p + q$ given together
  with its representation $A_M$ over field $\mathbb F$.
  Let $\mathcal S$ be a $p$-family 
  of independents of $M$.
  Then a $q$-representative family $\widehat{\mathcal S} \subseteq \mathcal S$
  of size at most $r \choose p$ can be found in
  $O\left({r \choose p}tp^\omega + t{r \choose q}^{\omega-1}\right)$ 
  operations over $\mathbb F$, where $\omega < 2.373$ is the matrix multiplication exponent. 
\end{theorem}

We are now ready to give the pseudo-code of the algorithm behind \cref{thm:fpt-ipp} (see \cref{alg:fpt-ipp}).

In \cref{alg:fpt-ipp}, $\mathcal A \bullet_{M} \mathcal B$ is defined 
as $\{ A \cup B \mid A\in\mathcal A, B \in\mathcal B, A \cap B = \emptyset, A \cup B \in \mathcal I \}$ 
for families $\mathcal A,\mathcal B \subseteq \mathcal I$ and matroid $M = (U,\mathcal I)$.

\begin{algorithm2e}[t]
	\KwIn{An instance $(D=(V,E),s,z,A_M)$ of \IPP{}, where $A_M$ is a representation of matroid $M=(V(D),\I)$ over field $\mathbb F$ and of rank $r$.}
	\KwOut{Determines whether $(D=(V,E),s,z,A_M)$ is a \yes- or \no-instance.}

	\medskip
	$T[v,i] \gets \emptyset$, for all $v \in V$ and $i \in [r-1]$.\label{line:init-start}\;
	$T[s,0] \gets \{s\}$.\label{line:init-end}\;
	\For(\label{line:for-r}){$i\gets1$ \KwTo $r-1$}{ 
		\ForEach{$w \in V$}{ 
			$\mathcal N_{w,i} \gets \emptyset$.\label{line:N-start}\;
			\ForEach{$(v,w) \in E$ with $T[v,i-1] \not= \emptyset$ and $\{w\} \in \mathcal I$}{ 
				$\mathcal N_{w,i} \gets \mathcal N_{w,i} \cup \left(T[v,i-1] \bullet_M \left\{\left\{w\right\}\right\}\right)$.\label{line:N-end}\;
			}
			$T[w,i]\gets (r-i-1)$-representative of $\mathcal N_{w,i}$.\label{line:compress} \hfill(using \cref{thm:eff-rep-fam})\;
		}
		\lIf(\label{line:yes}){$T[z,i] \not=\emptyset$}{
			\KwRet{$(D=(V,E),s,z,A_M)$ is a \yes-instance.}
		}
	}
	\KwRet{$(D=(V,E),s,z,A_M)$ is a \no-instance.}
	\caption{\IPP{} parameterized by the rank $r$.}
\label[algorithm]{alg:fpt-ipp}  
\end{algorithm2e}
\begin{lemma}\label{lem:fpt-ipp-correct}
	\cref{alg:fpt-ipp} is correct.
\end{lemma}
\begin{proof}
	Let $\mathcal P_{w,i} := \{ X \in \I \mid $ there is an \dipath[s,w] $P$ of length $i$ such that $V(P)=X \}$, for all $w \in V$ and $i \in [r-1]$.
	Observe that $\mathcal P_{w,i}$ is an $(i+1)$-family of independent sets.
	We show by induction that after iteration $i$ of the for-loop in Line \eqref{line:for-r} 
	the entry $T[w,i]$ is an $(r-i)$-representative of $\mathcal P_{w,i}$, for all $w \in V$ and $i \in [r-1]$.
	Then the correctness follows, since we check after each of these iterations whether $T[w,i]$ is non-empty (Line \eqref{line:yes}).
	Observe that $\mathcal P_{s,0} = \{ s \}$ and $\mathcal P_{v,0} = \emptyset$ for all $v \in V \setminus \{s\}$.
	Hence, the entries of $T$ computed in Lines \eqref{line:init-start} and \eqref{line:init-end} fulfill our induction hypothesis.

	Now let $i \in  [r-1]$ be the current iteration of the for-loop in Line \eqref{line:for-r} and assume that for all $j < i$ 
	we have that $T[w,j]$ is an $(r-j)$-representative of $\mathcal P_{w,j}$, for all $w \in V$.
	Fix a vertex $w \in V$.
	We first show that if there is an $X \in T[i,w]$, 
	then there is an \dipath[s,w] $P_w$ of length $i$ such that $X=V(P_w) \in \I$.
	Observe that in Lines \eqref{line:N-start}--\eqref{line:N-end} we look at each possible predecessor $v\in V$ of $w$ 
	in an \dipath[s,w]{} of length $i$, take each set $X' \in T[v,i-1]$ and 
	check whether $X' \cup \{w\}$ is an independent set of size $i+1$.
	If this is the case, we add it to $\mathcal N_{w,i}$.
	After Line \eqref{line:compress}, we have that $T[w,i] \subseteq \mathcal N_{w,i}$.
	Since $X' \in T[v,i-1]$, we know that there is an \dipath[s,v]{} $P_v$ of length $i-1$ with $X' = V(P)$.
	Thus, if there is an $X \in T[i,w]$, then there is an \dipath[s,w] $P_w$ of length $i$ such that $X=V(P_w) \in \I$

	Now let $X \in \mathcal P_{w,i}$
	and $Y \subseteq V(D)$ be a set of vertices of size at most $r - i -1$ such that $X \cap Y = \emptyset$ and $X \uplus Y \in \I$.
	Hence, there is an \dipath[s,w]{} $P$ of length $i$ such that $V(P) = X$.
	Let $v$ be the predecessor of $w$ in $P$.
	Let $P_v$ be the \dipath[s,v]{} of length $i-1$ induced by $P$ without $w$.
	Hence, $V(P_v) \in \mathcal P_{v,i-1}$.
	Moreover, $V(P_v) \cap (Y \cup \{ v \}) = \emptyset$ and $V(P_v) \uplus (Y \cup \{ w \}) \in \I$.
	Since $T[v,i-1]$ is an $(r-i+1)$-representative family of $\mathcal P_{v,i-1}$,
	we know that there is an $\widehat{X} \in T[v,i-1]$ 
	such that $\widehat{X} \cap (Y \cup \{ w \}) = \emptyset$
	and $\widehat{X} \cup (Y \cup \{ w \}) \in \I$.
	In Lines \eqref{line:N-start}--\eqref{line:N-end} we add $\widehat{X} \cup \{w \}$ to $\mathcal N_{w,i}$.
	Let $X' := \widehat{X} \cup \{w \}$ and 
	note that $X' \cap Y = \emptyset$ 
	and  $X' \uplus Y \in \I$.
	Since $T[w,i]$ is an $(r-i)$-representative family of $\mathcal N_{w,i}$, we know that
	there is an $\widehat{X'} \in T[w,i]$
	such that $\widehat{X'} \cap Y = \emptyset$ 
	and $X' \uplus Y \in \I$.
	Thus, $T[w,i]$ is an $(r-i)$-representative of $\mathcal P_{w,i}$.
\end{proof}
Next, we show that \cref{alg:fpt-ipp} is actually a fixed-parameter algorithm parameterized by the length of a shortest $\Delta$-restless temporal $(s,z)$-path.
\begin{lemma} \label{lem:fpt-ipp-time}
	\cref{alg:fpt-ipp} runs in time of $O(2^{\omega r}\cdot m)$ operations over $\mathbb F$, where
	$\mathbb F$ is the field of $A_M$,
	$r$ is the rank of the matroid, 
	$m$ is the number of edges, and
	$2<\omega<2.373$ is an upper-bound for the matrix multiplication exponent.
\end{lemma}
\begin{proof}
	Without loss of generality we assume to have a total ordering on $V$. 
	We represent a subset of $V$ as a sorted string.
	Hence, union and intersection of two sets of size at most $r$ takes $O(r)$ time.
	We can thus look up and store sets of size at most $r$ 
	in a trie (or radix tree) 
	in $O(r)$ time \cite{cormen09}. %
	Note that we do not have the time to completely initialize the arrays of size $|V|$ in each trie node.
	Instead, we will initialize each array cell of a trie node at its first access.
	To keep track of the already initialized cells, 
	we use \emph{sparse sets} over $V$ which allows membership test, 
	insertion, and deletion of elements in constant time \cite{BL93}.

	We denote the \emph{in-neighborhood} of a vertex $w$ by $N^-(w) := \{ v \in V \mid (v,w) \in E \}$.
	Furthermore, let $H_{i,w}$ be the running time of Lines \eqref{line:N-start}--\eqref{line:N-end} 
	in iteration $i$ of the for-loop in Line \eqref{line:for-r}, 
	and~$R_{i,w}$ be the number of operations over $\mathbb F$ of Line \eqref{line:compress}
	in iteration $i$ of the for-loop in Line \eqref{line:for-r}.
	Then we can run \cref{alg:fpt-ipp} in time of $O\left( \sum_{i=1}^{r-1} \sum_{w \in V} H_{i,w} + \sum_{i=1}^{r-1} \sum_{w \in V} R_{i,w} \right)$
	operations over $\mathbb F$---that is, the running time respecting the time needed for operations over $\mathbb F$.
	Let $i \in [r-1]$ and $w \in V$.
	In the $i$-th iteration of the for-loop in Line \eqref{line:for-r},
	$|T[v,j]| \leq {r \choose j+1}$ for all $j < i$ and $v \in V$, 
	since we used \cref{thm:eff-rep-fam}
	in prior iterations.
	Let $2 < \omega < 2.373$ be an upper-bound for the matrix multiplication exponent.
	Hence, $|\mathcal N_{w,i}| \leq {r \choose i+1} |N^-(w)|$ and
	$H_{i,w} \in O({r \choose i+1} |N^-(w)| \cdot r^\omega)$, 
	because the independence test can be done via matrix multiplication.
	Thus,
	\begin{align*}
		&O\left( 
			\sum_{i=1}^{r-1} \sum_{w \in V} H_{i,w}
		\right)
		\subseteq O\left( 
			\sum_{i=1}^{r-1} \sum_{w \in V} 
			|N^-(w)|{r \choose i+1} \cdot r^\omega
		\right)
		\subseteq O\left( 
		 	2^{r+o(r)}m 
		\right).
	\end{align*}
	Moreover, by \cref{thm:eff-rep-fam}, 
	we have 
	\begin{align*}
		 O\left(	\sum_{i=1}^{r-1} \sum_{w \in V} R_{i,w} \right)
			 &\subseteq O\left(	\sum_{i=1}^{r-1}  
		m {r \choose i}{r \choose i+1}(i+1)^\omega +
		\sum_{i=1}^{r-1} 
		m {r \choose i}{r \choose r-i-1}^{\omega-1}
		\right)\\
				&\subseteq	O\left( 2^r m( 2^{r}r^{\omega} + 2^{r(\omega-1)}) \right)
				\subseteq	O\left( m( 2^{2r+\log_2(r)\omega} + 2^{r\omega}) \right)
				\subseteq	O\left( 2^{\omega r} m \right),
		\end{align*}
		where the last inclusion is true because we assume $2<\omega$.

Thus, we can run \cref{alg:fpt-ipp} in time 
of $O\left(2^{\omega r}m\right)$ operations 
over $\mathbb F$.
\end{proof}

Now \cref{thm:fpt-ipp} follows from \cref{lem:fpt-ipp-time,lem:fpt-ipp-correct}.

Observe that by \cref{lem:exp-correct}, there is a \tpath{} in the temporal graph $\TG$ 
if and only if 
there is an $(s,z)$-path $P$ in the $\Delta$-$(s,z)$-expansion $D = (V',E')$ of $\TG$
such that $V(P)$ is an independent set in the 
\emph{partition matroid}\footnote{Partition matroids are linear \cite{Mar09}.} $M = (V',\{ X \subseteq V' \mid \forall v\in V \colon |X \cap V'(v)| \leq 1 \})$.
Note that $M$ is of rank $|V|$ and hence too large to show \cref{thm:fpt-length} with \cref{thm:fpt-ipp}.

A \emph{$k$-truncation} of a matroid $(U,\I)$ 
is a matroid $(U,\{ X \in \I \mid |X| \leq k \})$ 
such that all independent sets are of size at most $k$.
The $k$-truncation of a linear matroid is also a linear matroid \cite{Mar09}.
In our reduction from \TPP{} to \IPP{} we 
use a $(k+1)$-truncation of matroid $M$.
Two general approaches are known to compute a representation 
for a $k$-truncation of a linear matroid---one 
is randomized \cite{Mar09} 
and one is deterministic \cite{LMPS18}.\footnote{For both algorithms, a representation of the original matroid must be given.}
Both approaches require a large field 
implying that one operation over that field is rather slow.
However, for our specific matroid we 
employ the Vandermonde matrix 
to compute 
a representation over a 
small finite field.
Note that we would not get a running time
linear in the input size by applying the algorithm of Lokshtanov~et~al.~\cite{LMPS18} or Marx~\cite{Mar09} on $M$.

\begin{lemma}\label{lem:rep}
	Given a universe $U$ of size $n$, 
	a partition $P_1 \uplus \dots \uplus P_q = U$,
	and an integer $k \in \N$,
	we can compute in $O(kn)$ time a representation $A_M$ for 
	the matroid 
	$M = \Big(U, \Big\{ X \subseteq U \ \Big\vert\ |X| \leq k 
		\text{ and } \forall i \in [q] \colon |X \cap P_i| \leq 1  \Big\}\Big)$,
		where $A_M$ is defined over a finite field $\mathbb F$ 
		and one operation over $\mathbb F$ takes constant time.
\end{lemma}
\begin{proof}
For this running time analysis
we assume the \emph{Word RAM model of computation}, introduced by \cite{FW90}, which is similar 
to the \emph{RAM model of computation} 
but one memory cell can store only $O(\log n)$ many bits, where $n$ is the input size.
This avoids abuse of the unit cost random access machine 
by for example multiplying very large numbers in constant time. 

	Without loss of generality we assume that $q \leq n$.
	Let $p$ be a prime number with $q \leq p \leq 2q$.
	Such a prime exists by the folklore Betrand-Chebyshev Theorem~\cite{proofsfromthebook} 
	and can be computed in $O(n)$ time using Lagarias-Odlyzko Method~\cite{TCH12}.
	To perform one operation on the prime field $\mathbb F_p$, one can first perform the primitive operation in $\mathbb Z$
	and them take the result modulo~$p$.
	Since $p \leq 2q \in O(n)$,
	each element of $\mathbb F_p$ fits into one cell of the
	\emph{Word RAM model of computation}.
	Thus, we can perform one operation over $\mathbb F_p$ in constant time.

	Let $x_1,\dots,x_q$ be pair-wise distinct elements from $\mathbb F_p$.
	To compute an $(k \times n)$-matrix~$A_M$ as representation for $M$ over $\mathbb F_p$, we
	compose (column-wise) for each element $u \in P_i$ the 
	vector $\mathbf v_i := \begin{pmatrix}x_i^0 & x_i^1 & \dots & x_i^{k-1} \end{pmatrix}^T$, where $i \in [q]$.
	That gives a running time of $O(k \cdot n)$ operations over $\mathbb F_p$, 
	since we can compute $\mathbf v_i$ in $O(k)$ operations over $\mathbb F_p$.

	It remains to show that $A_M$ is a representation of $M$.
	Let $X \subseteq U$. 
	If there is an $i \in [p]$ such that~$|X \cap P_i| > 1$, 
	then the corresponding columns of $A_M$ are linearly dependent, because we have the vector~$\mathbf v_i$ twice.
	Now we assume that for all $i \in [q]$ we have $|X \cap P_i| \leq 1$.
	Furthermore, if $|X| > k$, 
	then we know that the corresponding columns of $A_M$ are linearly dependent, 
	because $A_M$ is an $(k \times n)$-matrix.
	We can observe that if $|X| = k$, then the corresponding columns in $A_M$ 
	form a Vandermonde matrix, whose determinate is known to be non-zero.
	Hence, if $|X| \leq k$, then the corresponding columns in $A_M$ are linearly independent.
	Thus, $A_M$ is a representation of~$M$.
 \end{proof}
We now show a reduction from \TPP{} to \IPP{} using \cref{lem:exp-runtime,lem:exp-correct,lem:rep}

 \begin{lemma}%
		\label{lem:tpp-to-ipp}
	Given an instance 
	$(\TG,s,z,k,\wait)$ of \TPP{},
	we can compute in $O(\max\{k,\wait\}\cdot |\TG|)$ time
	an instance $(D,s,z,A_M)$ of \IPP{} such that 
				$M$ has rank $k+1$, and
			$(\TG,s,z,k,\wait)$ is a \yes-instance  
			if and only if
			$(D,s,z,A_M)$ is a \yes-instance, %
	where one operation over the finite field of $A_M$ takes constant time.
\end{lemma}
\begin{proof}
	Let $(\TGcompact,s,z,k,\wait)$ be an instance of \TPP{}.
	We construct an instance $(D,s,z,A_M,k)$ of \IPP{} in the following way.
	Let digraph $D = (V',E')$ be the $\wait$-$(s,z)$-expansion of $\TG$ which can be computed, by \cref{lem:exp-runtime}, in $O(|\TG|\cdot\wait)$ time
	such that $V' \in O(|\TG|)$.
	Observe that $\bigcup_{v \in V} V'(v)$ is a partition of $V'$.
	Now, we construct a representation $A_M$ (over a finite field where we can perform one operation in constant time) of the matroid 
	\begin{align*}
		M =  \Big(V', \big\{ X \subseteq V' \ \big\vert\ |X| \leq k+1 \text{ and } \forall v \in V \colon |X \cap V'(v)| \leq 1 \big\}   \Big)
	\end{align*}
	in $O(k\cdot|\TG|)$ time by \cref{lem:rep}.
	Note that $M$ is an $(k+1)$-truncated partition matroid and hence has rank $k+1$.
	This completes the construction and gives us an overall running time of $O(\max\{k,\wait\}\cdot |\TG|)$.

	We now claim $(\TG,s,z,k,\wait)$ is a \yes-instance  
			if and only if
			$(D,s,z,A_M)$ is a \yes-instance and contains an independent \dipath\ of length at most $k$.
	
	\LArrow Let $P$ be a \tpath{} of length $k' \leq k$ in $\TG$.
	Then, by \cref{lem:exp-correct} there is an \dipath{} $P'$ of length $k'$ such that 
	for all $v \in V$ it holds that~$|V'(v) \cap V(P')| \leq 1$.
	Since $|V(P')| = k'+1 \leq k+1$, we know that $V(P')$ is an independent set of $M$.	
	Thus, $P'$ is a witness of length at most $k$ for $(D,s,z,A_M)$ being a \yes-instance.

	\RArrow Let $P'$ be an \dipath{} of length $k' \leq k$ in $D$ 
	such that $V(P')$ is an independent set of $M$.
	Clearly, for $v \in V$ it holds that~$|V'(v) \cap V(P')| \leq 1$.
	Then, by \cref{lem:exp-correct}, there is a \tpath{} of length $k'$ in $\TG$.
\end{proof}
		\begin{proof}[Proof of \cref{thm:fpt-length} (ii)]
	Let $I = (\TGcompact,s,z,k,\wait)$ be an instance of \TPP{}.
	To decide whether there is a  witness of length $k$ of $I$ being a \yes-instance,
	we first use \cref{lem:tpp-to-ipp} to compute 
	an instance $I' = (D,s,z,A_M)$ of \IPP{} in $O(|\TG|\cdot\max\{\wait,k\})$ time, 
	where we can compute one operation over the field $\mathbb F$ of $A_M$ in constant time 
	and the matroid $M$ which is represented by $A_M$ is of rank $k+1$.
	Note that $I'$ is a \yes-instance if and only if there is witness of length $k$ for $I$ being a \yes-instance.
	Second, we solve~$I'$ by \cref{thm:fpt-ipp} in $O(2^{\omega (k+1)}\cdot |\TG|\cdot\wait)$ time.

	Thus, we have an overall running time of~$2^{O(k)} \cdot |\TG|\cdot\wait$.
\end{proof}

Moreover, from \cref{lem:tpp-to-ipp} it is intermediately clear 
that the lower-bounds of \cref{cor:path:eth,thm:probRestlessPath:NPh} 
translate to \IPP{}.
\begin{corollary}
	\label{cor:ipp-hard}
	\IPP{} is NP-hard and unless the ETH fails there is
	no $2^{o(n)}$-time algorithm for it,
	where $n$ is the number of vertices.
\end{corollary}
Note that from \cref{thm:probRestlessPath:NPh} we can further deduce that
there is not much hope for \emph{fast} or \emph{early} restless temporal
paths, that is, restless temporal path that have a small duration or an early
arrival time. The instance constructed in the
reduction has lifetime~$\lifetime = 3$ and hence the duration as well as the arrival time of any restless
temporal path in this instance is at most three. This implies that we presumably
cannot find fast or early restless temporal paths efficiently.

\section{Computational complexity landscape for the underlying graph }
\label{sec:structparam}

In this section we investigate the parameterized computational complexity of \probRestlessPath 
when parameterized by structural parameters of the underlying graph. 
We start by observing that whenever a parametrization forbids path of unbounded length, 
then we can use \cref{thm:fpt-length} to show fixed-parameter tractability.
For example, if we consider the vertex cover number~$\text{vc}_\downarrow$ of the underlying graph, 
then we can deduce that any path in the underlying graph and 
hence any restless temporal path can have length at most~$2\text{vc}_\downarrow+1$. 
Thus, by \cref{thm:fpt-length}, we get fixed-parameter tractability of \probRestlessPath{} 
when parameterized by the vertex cover number of the underlying graph.
\begin{observation}
		\label{obs:vertexcover}
\probRestlessPath parameterized by the vertex cover number $\text{vc}_\downarrow$ of the underlying graph
is fixed-parameter tractable.
\end{observation}
From a classification standpoint, we can improve this a little further by observing that 
the length of a path in the underlying graph can be bounded by~$2^{O(\text{td}_\downarrow)}$~\cite{ND12}, 
where $\text{td}_\downarrow$ is the treedepth of the underlying graph.
\begin{observation}
	\label{obs:treedepth}
\probRestlessPath parameterized by the treedepth $\text{td}_\downarrow$ of the underlying graph
is fixed-parameter tractable.
\end{observation}

One of the few dark spots of the landscape is the 
\emph{feedback edge number}\footnote{For a given graph $G=(V,E)$ a set $F \subseteq E$ is a \emph{feedback edge set}
if $G - F$ does not contain a cycle.
The \emph{feedback edge number} of a graph $G$
is the size of a minimum feedback edge set for $G$.} 
of the underlying graph which is resolved in the following way.
\begin{theorem}
		\label{thm:fpt-fes}
\probRestlessPath{} can be solved in $2^{O(\fes)}\cdot |\TG|$ time, 
	where~$\fes$ is the feedback edge number of the underlying graph.
\end{theorem}
By \cref{cor:path:eth} we know that \cref{thm:fpt-fes} is asymptotically optimal, unless ETH fails.
In a nutshell, our algorithm to prove \cref{thm:fpt-fes} has the following five steps:
 \begin{enumerate}
 	\item Exhaustively remove all degree-$1$ vertices from $\UG$ (except for $s$ and $z$). 
 	\item Compute a minimum-cardinality feedback edge set~$F$ of the graph $\UG$. %
	\item Compute a set $\mathcal P$ of $O(\fes)$ many paths in~$\UG - F$ 
			such that every path in~$\UG-F$ is a concatenation of some paths in $\mathcal P$.
	\item ``Guess'' the feedback edges in $F$ and paths in~$\mathcal P$ of an $(s,z)$-path in $\UG$.
	\item Verify whether the ``guessed'' $(s,z)$-path is a $\Delta$-restless \nonstrpath{s,z} in~$\TG$.  			
 \end{enumerate}
 First, we show that we can safely remove all (except $s$ and $z$) degree-one vertices from the underlying graphs~$\UG$.  
\begin{rrule}[Low Degree Rule]
	\label{rr:lowdeg}
	Let $I = (\TGcompact,s,z,\Delta)$ be an instance of \probRestlessPath{}, 
	$\UG$ be the underlying graph of $\TG$, $v \in V \setminus \{s,z\}$, and $\deg_{\UG}(v)\leq 1$.
	Then, output $(\TG - \{v\},s,z,\Delta)$.
\end{rrule}
\begin{lemma}%
		\label{lem:rr:lowdeg}
	\cref{rr:lowdeg} is safe and can be applied exhaustively in $O(|\TG|)$ time.
\end{lemma}
\begin{proof}
	Let $I = (\TGcompact,s,z,\Delta)$ be an instance of \probRestlessPath{}, 
	For the safeness we can observe that a vertex $v \in V \setminus \{s,z\}$ with $\deg_{\UG}(v) \leq 1$ 
	cannot be visited by any $\Delta$-restless \nonstrpath{s,z}. 
	To apply \cref{rr:lowdeg} exhaustively, we iterate once over 
	the set of time edges to store for each vertex $v \in V \setminus \{s,z\}$ 
	its degree in a counter $c_v$.
	Afterwards, we collect all vertices of degree $0$ in $X$ and 
	all vertices of degree $1$ in $V_1$.
	Now we iterate over each vertex  $v \in V_1$, remove $v$ from $V_1$, add $v$ to $X$, 
	decrement the counter $c_u$ of its neighbor $u$.
	If $c_u$ becomes~$1$ we add $u$ to $V_1$.
	Note that this procedure ends after $O(|V|)$ time.

	Finally, we iterate one last time over the temporal graph $\TG$ 
	to construct the temporal graph $\TG' := \TG - X$.
	The instance $(\TG',s,z,\Delta)$ of \probRestlessPath 
	is the resulting instance when we apply \cref{rr:lowdeg} exhaustively on $I$.
\end{proof}

Next, we consider a static graph~$G$ with no degree-one or degree-zero vertices.
Let $F$ be an minimum feedback edge set of $G$ and let $V_F$ be the endpoints of the edges in~$F$, 
that is $V_F=\{v \in e \mid e \in F\}$. 
Let $V^{\geq 3}$ be the set of all vertices with a degree greater than two in~$G-F$.  
We can partition the graph~$G-F$ into a set~$\mathcal P$ of \emph{$V_F \cup  V^{\geq 3}$-connecting paths}, 
that are, all paths in $G-F$ who start and end in $V_F \cup  V^{\geq 3}$ and have no internal vertices in that set of vertices. 
Note that all degree-one vertices of $G-F$ are in $V_F$. 
Hence, the graph~$G-F$ can be partitioned into $V_F \cup  V^{\geq 3}$-connecting paths. 
We can show that $|\mathcal P| \in O(\fes)$.
\begin{lemma}%
		\label{lem:fes}
	Let $G$ be a graph with no degree-one vertices and $F$ be an minimum feedback edge set of $G$. 
	The set~$\mathcal P$ of $V_F \cup  V^{\geq 3}$-connecting paths of $G-F$ has size~$O(|F|)$ and can be computed in~$O(|G|)$ time. 
\end{lemma} 
\begin{proof}
We can compute the set~$\mathcal P$ in $O(|G|)$ time as follows. 
We start with $\mathcal P = \emptyset$ and pick any leaf~$v \in V(G-F)$ of degree one.
Recall that $v \in V_F$ and that $G-F$ is cycle-free. 
There is at most one vertex $w \in (V^{\geq 3} \cup V_F) \setminus \{ v\}$ such that there is a path $P$ between $v$ and $w$
which does not contain internal vertices from $V^{\geq 3} \cup V_F$.
Note that also $P$ is unique.
We add $P$ to $\mathcal P$ and remove $V(P) \setminus \{w\}$ from the graph.
Now we repeat this procedure with the next leaf of degree one until the graph has no edges. 

It is easy to verify that the number of paths is bounded by the number of vertices in~$V^{\geq 3} \cup V_F$.
We know that $|V_F|$ is upper-bounded by~$2 |F|$.
It remains to show that~$|V^{\geq 3}|$ is in $O(|F|)$.
 
As shown by Bentert et al.~\cite[Lemma 2]{BDKNN18}, the number of vertices with degree greater or equal to three is bounded by $3 |F|$ in a graph with no degree-one vertices. 
Hence, the number of $V_F \cup  V^{\geq 3}$-connecting paths is bounded by~$5 |F|$. 
\end{proof}
With \cref{lem:rr:lowdeg,lem:fes} we can prove \cref{thm:fpt-fes}.
\begin{proof}[Proof of \cref{thm:fpt-fes}]
Let $I = (\TGcompact,s,z,\Delta)$ be an instance of \probRestlessPath{} and 
$\UG$ be the underlying graph of $\TG$. 
Without loss of generality, we can assume that all vertices in $V(\UG) \setminus \{s,z\}$ have a degree greater than one. 
If this is initially not the case, then we safely remove all degree-one vertices of the underlying graph exhaustively in $O(|\TG|)$ time by \cref{lem:rr:lowdeg}. 
 
First we compute an minimum feedback edge set $F$ of $\UG$ in $O(|\UG|)$ time. %
Then, we compute the set~$\mathcal P$ of $V_F \cup V^{\geq 3} \cup \{s,z\}$-connecting paths of $\UG -F$ in~$O(|\UG|)$ time by \cref{lem:fes}. 
Note that the additional vertices $s$ and $z$ can increase the size of $\mathcal P$ by at most four. 
Now, for any subset of feedback edges $F' \subseteq F$ and $\mathcal P' \subseteq \mathcal P$, 
	we check whether  $F' \cup E(\mathcal P')$ form an $(s,z)$-path~$P$ in $\UG$ where $E(\mathcal P') := \bigcup_{P' \in \mathcal P'} E(P')$.   
This can be decided in $O(|\UG|)$ time by a simple breadth-first search on~$\UG$ starting at the vertex~$s$ and using only edges in $F' \cup E(\mathcal P')$. 
Last, we verify whether $P$ forms a $\Delta$-restless \nonstrpath{s,z} in~$\TG$. 
Therefore, we consider the temporal graph $\TG_P=(V,(E_i \cap E(P))_{i\in [\ell]})$ which has $P$ as underlying graph. 
Note that we can construct $\TG_P$ in $O(|\TG|)$ time, by iterating once over the set of time edges of $\TG$.
By \cref{lem:restless-path-on-a-path} we can decide in $O(|\TG_P|)$ time whether $\TG_P$ has a $\Delta$-restless \nonstrpath{s,z}. 

It is easy to check that the algorithm described above runs in~$2^{O(|F|)} |\TG|$ time. 

\proofpar{Correctness.} 
It remains to show the correctness of the algorithm.

\RArrow If our algorithm outputs \yes, then there is a $\Delta$-restless \nonstrpath{s,z} in~$\TG_P$. 
The temporal graph $\TG_P$ contains a subset of the time edges of $\TG$, hence the $\Delta$-restless \nonstrpath{s,z}  in $\TG_P$ is also present in $\TG$. It follows that $I$ is a yes-instance.

\LArrow Assume $I$ is a \yes-instance. 
Then there exists a $\Delta$-restless \nonstrpath{s,z} in the temporal graph $\TG$. 
Let $P =\big( (v_0,v_1, t_1), \dots, (v_{k-1}, v_k, t_{k} ) \big)$ be such a path. 
Hence, $P' =\big( \{v_0,v_1\}, \dots, \{v_{n-1}, v_n\} \big)$ 
is an $(s,z)$-path in the underlying graph~$\UG$. 
Let $F'=F\cap E(P')$. 
If we remove the edges in $F'$ from $P'$ then what remains is a collection of paths where each path is a concatenation of paths in $\mathcal P$.
Hence, there exists a subset~$\mathcal P'\subseteq \mathcal P$   
such that $ F' \cup E ( \mathcal P' ) = E(P)$. 
Thus, we will find $P'$ in $\UG$ and, by \cref{lem:restless-path-on-a-path}, 
we will correctly verify that this $P'$ forms a $\Delta$-restless \nonstrpath{s,z} in~$\mathcal G$.
\end{proof}

The results from \cref{sec:path:hardness,sec:fptalg,sec:structparam} provide a good picture
of the parameterized complexity landscape for \probRestlessPath, 
meaning that for most of the widely known (static) graph parameters we know 
whether the problem is in \FPT or \wone-hard or para-\NP-hard, see \cref{fig:hierarchy}.

Our understanding of the class of temporal graphs 
where we can solve \probRestlessPath{} efficiently narrows down to the following points.
We can check efficiently whether there is a \tpath{} $P$ in temporal graph $\TG$ if 
\begin{enumerate}
		\item there is a bounded number of $(s,z)$-path in $\UG$ (cf.~\cref{thm:fpt-fes} and \cref{lem:restless-path-on-a-path}),
		\item there is a bound on the length of $P$ (cf.~\cref{thm:fpt-length,obs:treedepth,obs:vertexcover}).
\end{enumerate}
Apart from that we established with \cref{thm:probRestlessPath:NPh,cor:underlying-nph,thm:probRestlessPath:W1hFVS} hardness results for temporal graphs having restricted underlying graphs, see \cref{fig:hierarchy}.

Finally, we show that we presumably cannot expect to obtain polynomial kernels for all parameters considered so far and most structural parameters of the underlying graph.
\begin{proposition}\label{prop:path:nopk}
\probRestlessPath parameterized by the number $n$ of vertices does not admit a polynomial kernel for all $\Delta\ge 1$ unless \NoKernelAssume.
\end{proposition}
We employ the OR-cross-composition framework by~Bodlaender, Jansen, and Kratsch~\cite{BJK14} to refute the existence of a polynomial kernel for a parameterized problem under the assumption that \NNoKernelAssume, the negation of which would cause a collapse of the polynomial-time hierarchy to the third
level. In order to formally introduce the framework, we need some definitions.

An equivalence
relation~$R$ on the instances of some problem~$L$ is a
\emph{polynomial equivalence relation} if
\begin{enumerate}
 \item one can decide for each two instances in time polynomial in their sizes whether they belong to the same equivalence class, and
 \item for each finite set~$S$ of instances, $R$ partitions the set into at most~$(\max_{x \in S} |x|)^{O(1)}$ equivalence classes.  
\end{enumerate}

Using this, we can now define OR-cross-compositions.

\begin{definition}
An \emph{OR-cross-composition} of a problem~$L\subseteq \Sigma^*$ into a
parameterized problem~$P$ (with respect to a polynomial equivalence
relation~$R$ on the instances of~\(L\)) is an algorithm that takes
$n$ $R$-equivalent instances~$x_1,\ldots,x_n$ of~$L$ and
constructs in time polynomial in $\sum_{i=1}^n |x_i|$ an instance
$(x,k)$ of~\(P\) such that
\begin{enumerate}
\item $k$ is polynomially upper-bounded in $\max_{1\leq i\leq n}|x_i|+\log(n)$ and 
\item $(x,k)$ is a \yes-instance of $P$ if and only if there is an $i\in [n]$ such that $x_{i}$ is a \yes-instance of $L$. 
\end{enumerate}
\end{definition}

If an \NP-hard problem~\(L\) OR-cross-composes into a parameterized
problem~$P$, then~$P$ does not admit a polynomial kernel, unless \NoKernelAssume~\cite{BJK14}.

\begin{proof}[Proof of \cref{prop:path:nopk}]
We provide an OR-cross-composition from \probRestlessPath onto itself. 
We define an equivalence relation~$R$ as follows: Two instances
$(\TGcompact,s,z,\Delta)$ and $(\TG'=(V',
(E'_i)_{i\in[\lifetime']}),s',z',\Delta')$ are equivalent under~$R$ if and only
if $|V|=|V'|$ and $\Delta=\Delta'$. Clearly, $R$ is a polynomial equivalence relation.

Now let $(\TG_1=(V_1,
(E_{1,i})_{i\in[\lifetime_1]}),s_1,z_1,\Delta),\ldots,(\TG_n=(V_n,
(E_{n,i})_{i\in[\lifetime_n]}),s_n,z_n,\Delta)$ be $R$-equi\-valent instances
of \probRestlessPath. We construct a temporal graph $\TG^\star=(V^\star,
(E^\star_{i})_{i\in[\lifetime^\star]})$ as follows. Let $|V^\star|=|V_1|$ and
$s^\star, z^\star\in V^\star$. We identify all vertices~$s_i$ with $i\in[n]$
with each other and with $s^\star$, that is, $s^\star=s_1=\ldots=s_n$.
Analogously, we identify all vertices $z_i$ with $i\in[n]$ with each other and
with $z^\star$, that is, $z^\star=z_1=\ldots=z_n$. We arbitrarily identify the
remaining vertices of the instances with the remaining vertices from $V^\star$,
that is, let
$V^\star\setminus\{s^\star,z^\star\}=V_1\setminus\{s_1,z_1\}=\ldots=V_n\setminus\{s_n,z_n\}$.
Now let $E^\star_1=E_{1,1}, E^\star_2=E_{1,2}, \ldots,
E^\star_{\lifetime_1}=E_{1,\lifetime_1}$. Intuitively, the first instance
$(\TG_1=(V_1, (E_{1,i})_{i\in[\lifetime_1]})$ essentially forms the first
$\lifetime_1$ layers of $\TG^\star$. Then we introduce $\Delta+1$ trivial
layers, that is,
$E^\star_{\lifetime_1+1}=E^\star_{\lifetime_1+2}=\ldots=E^\star_{\lifetime_1+{\Delta}+1}=\emptyset$.
Then we continue in the same fashion with the second instance and so on. We
have that $\lifetime^\star=\sum_{i\in[n]} \lifetime_i +
(n-1)\cdot(\Delta+1)$. 

This instance can be constructed in polynomial time and the number of vertices is the same as the vertices in the input instances, 
hence~$|V^\star|$ is polynomially upper-bounded by the maximum size of an input instance. 
Furthermore, it is easy to check that $\TG^\star$ contains a $\Delta$-restless \nonstrpath{s^\star,z^\star} 
if and only if there is an $i\in[n]$ 
such that $\TG_i$ contains a $\Delta$-restless \nonstrpath{s_i,z_i}. 
This follows from the fact that all instances are separated in time by $\Delta+1$ trivial layers, 
hence no $\Delta$-restless \nonstrpath{s^\star,z^\star} can use time edges from different original instances. 
Since \probRestlessPath is \NP-hard (\cref{thm:probRestlessPath:NPh}) the result follows.
\end{proof}

\newcommand{\TFVN}{Timed Feedback Vertex Number\xspace}
\newcommand{\tfvn}{timed feedback vertex number\xspace}
\newcommand{\tfvs}{timed feedback vertex set\xspace}
\newcommand{\tfvnvar}{\ensuremath{x}\xspace}
\section{Timed feedback vertex number}
\label{sec:tfvn}

In this section we introduce a new temporal version of the well-studied ``feedback
vertex number''-parameter. Recall that by \cref{thm:probRestlessPath:W1hFVS} we
know that \probRestlessPath{} is \wone{}-hard when parameterized by the
feedback vertex number of the underlying graph. %
This motivates studying larger parameters with the goal to obtain tractability
results. We propose a new parameter called \emph{\tfvn} which, intuitively,
quantifies the number of vertex appearances that need to be removed from a
temporal graph such that its underlying graph becomes cycle-free. Note that
having vertex appearances in the deletion set allows us to ``guess'' \emph{when}
we want to enter and leave the deletion set with a $\Delta$-restless
\nonstrpath{s,z} in addition to guessing in which order the vertex appearances
are visited. 

We remark that there also have been studies of removing (time) edges from
temporal graph to destroy \emph{temporal} cycles~\cite{Haa+20}, that is, temporal paths
from a vertex back to itself. Similarly, one also could remove vertex
appearances to destroy temporal cycles, resulting in a parameter that is smaller
than the timed feedback vertex number and incomparable to the feedback vertex
number of the underlying graph. Note that the mentioned parameters aiming at
destroying temporal cycles are unbounded in our reductions. We leave the
parameterized complexity of \probRestlessPath{} with respect to those parameters
open for future research.

Before defining \tfvn formally, we introduce notation for removing vertex
appearances from a temporal graph. Intuitively, when we remove a vertex
appearance from a temporal graph, we do not change its vertex set, but remove
all time edges that have the removed vertex appearance as an endpoint. Let
$\TGcompact$ be a temporal graph and $X\subseteq V\times[\lifetime]$ a set of
vertex appearances. Then we write $\TG-X:= (V, (E'_i)_{i\in[\lifetime]})$,
where $E_i'=E_i\setminus \{e\in E_i\mid \exists (v,i)\in X \text{ with } v\in
e\}$.
Formally, the \tfvn is defined as follows.

\begin{definition}[\TFVN]
Let \TGcompact be a temporal graph. A \emph{\tfvs} of $\TG$ is a set $X\subseteq V\times[\lifetime]$ of vertex appearances such that $\ug{\TG-X}$ is cycle-free.
The \emph{\tfvn} of a temporal graph~$\TG$ is the minimum cardinality of a \tfvs of $\TG$.
\end{definition}
We can observe that for any temporal graph the \tfvn is as least as large as the feedback vertex number of the underlying graph 
and upper-bounded by the product of the feedback vertex number of the underlying graph and the lifetime. 
We further remark that the \tfvn is invariant under reordering the layers.
At the end of this section we show how a \tfvs can be computed efficiently.

The main result of this section is that \probRestlessPath{} is fixed-parameter tractable when parameterized by the \tfvn of the input temporal graph.
To this end, we show the following.
\begin{theorem}\label{thm:tfvn}
	Given a \tfvs $X$ of size $\tfvnvar$ for a temporal graph~$\TGcompact$,
	we can decide in $O(6^x x!\cdot \max\{|\TG|^3,|V|^4 x^2\})$ time,
	whether there is a~$\Delta$-restless temporal $(s,z)$-path in $\TG$, where $s,z \in V$, $\Delta \in \N$.
\end{theorem}
The algorithm we present to show \cref{thm:tfvn} solves \textsc{Chordal Multicolored Independent Set}, where 
given a chordal graph\footnote{A graph is \emph{chordal} if it does not contain induced cycles of length four or larger.} $G=(V,E)$ 
and a vertex coloring $c \colon V\rightarrow [k]$, we are asked to decide whether $G$ contains an independent set of size $k$ that contains exactly one vertex of each color.
This problem is known to be \NP-complete~\cite[Lemma~2]{van2015interval} and solvable in $O(3^{k}\cdot |V|^2)$ time~\cite[Proposition~5.6]{bentert2019indu}.
Our algorithm for \probRestlessPath{} roughly follows these computation steps:
\begin{enumerate}
\item ``Guess'' which of and in which order the vertex appearances from the \tfvs appear in the $\Delta$-restless \nonstrpath{s,z}.
\item Compute the path segments between two \tfvs vertices by solving a \textsc{Chordal Multicolored Independent Set} instance.
\end{enumerate}

\begin{algorithm2e}[t]
		\KwIn{Temporal graph \TGcompact with $s,z \in V$, 
		timed feedback vertex set $X$ with $s,z \not \in \{v \mid (v,t) \in X\}$,
and $\Delta \in \N$.}
		\KwOut{\yes, if there is a $\Delta$-restless temporal $(s,z)$-path, otherwise \no.}

	\medskip
	\For(\label{line:for:partition}){\textbf{each} valid partition $O\uplus I \uplus U = X$}{ 
			$\TG' \gets \TG - U$ and $x \gets |I \cup O|$.\;
			$\mathcal T \gets \TG' - (\{ v \in V \mid (v,t) \in O \cup I\} \cup \{s,z\})$.\; 
			\For(\label{line:for:order}){\textbf{each} $\Delta$-ordering $(v_0,t_0)\leq \dots \leq (v_{x+1},t_{x+1}) \text{ of } I\cup O \cup (\{s,z\} \times \{\bot\})$}{

					$\mathcal P_i \gets \emptyset$, for all $i \in [0,x+1]$\label{line:initP}.\;

					\For(\label{line:for:i}){$i\gets 1$ \KwTo $x+1$}{ 
							\lIf(\label{line:samevertex}){$v_{i-1} = v_{i}$}{ $\mathcal P_i = \{ \emptyset \}$. }
							\For(\label{line:for:ee}){\textbf{each} $e_1 = (\{v_{i-1},w\},t)$, $ e_2 = (\{u,v_{i}\},t')$ of $\TG'$ where $v_{i-1} \not = v_{i}$}{ 
								$\mathcal T' \gets \mathcal T  + \{e_1,e_2\}$\label{line:T3}.\;
								\lIf(\label{line:valid3}){$\exists$ $(t_{i-1},t_{i})$-valid \tpath[v_{i-1},v_{i}] $P$ in $\mathcal T'$}{
										$\mathcal P_i \gets \mathcal P_i \cup \{ V(P) \setminus \{v_{i-1},v_{i}\} \}$.\;
								}	
								\label{line:for:ee:end}
							}
					}
					$G \gets $ intersection graph of the multiset $\{ P^{(i)} \in \mathcal P_i \mid i \in [x+1] \}$ \label{line:intersection-graph}.\;

					Define $c \colon V(G) \rightarrow [x+1], P^{(i)} \mapsto i$.\;
					\lIf(\label{line:mc-indset}){$(G,c)$ has a multicolored independent set of size $x+1$}{
						\KwRet{\yes.}
					}
			}
	}
	\KwRet{\no.}
	\caption{FPT algorithm for \probRestlessPath parameterized by timed feedback vertex set.}
\label[algorithm]{alg:fpt-tfvs}  
\end{algorithm2e}
We give a precise description of our algorithm in \cref{alg:fpt-tfvs}.
Here, a partition $O\uplus I \uplus U$ of a set of vertex appearances $X$ is \emph{valid} if
we have $v \not=v'$, for all distinct $(v,t),(v',t') \in I$ 
and for all distinct $(v,t),(v',t') \in O$. 
A vertex appearance~$(v,t) \in I$ signals that a \tpath[s,z]{} arrives in $v$ at time~$t$
and $(v,t)\in O$ signals that it departs from~$v$ at time~$t$.  
Let $M := O\cup I \cup (\{s,z\} \times \{\bot\})$.
We call a linear ordering~$(v_0,t_0)\leq_M \cdots \leq_M (v_{x+1},t_{x+1})$ of $M$ a \emph{$\Delta$-ordering} if 
$(v_0,t_0)=(s,\bot)$, $(v_{x+1},t_{x+1})=(t,\bot)$, 
$t_i \leq t_j$ if and only if $i<j \in [x]$, and
for all $v\in V$ with $(v,t_i) \in I$ and $(v,t_j) \in O$ it holds that
$i+1=j$ and $t_i \leq t_j \leq t_i+\Delta$. 
Moreover, observe that for a vertex appearance $(v,t)\in I$, the \tpath[s,z]{} has to depart from $v$ not later than $t+\Delta$ and for vertex appearance $(v,t)\in O$, it has to arrive in $v$ not earlier than $t - \Delta$.  
To this end, we define the notion of a valid path between two consecutive vertex appearances:
\begin{definition}
		\label{def:valid}
		Let $O\uplus I \uplus U$ be a valid partition of $X$,
		and let $(v_i,t_i),(v_{i+1},t_{i+1}) \in I \cup O \cup (\{s,z\} \times \{\bot\})$ with $v_i \not = v_{i+1}$,
		and
		 $P$ a \tpath[v_{i},v_{i+1}] with 
		departure time $t_d$ and
		arrival time $t_a$.
		Then $P$ is \emph{$(t_{i},t_{i+1},I,O)$-valid} if the following holds true
		\begin{enumerate}[(i)] %
				\item $(v_{i},t_{i}) \in I \implies t_{i} \leq t_d \leq t_{i}+\Delta$,
				\item $(v_{i},t_{i}) \in O \implies t_d = t_{i}$,
				\item $(v_{i+1},t_{i+1}) \in I \implies t_a = t_{i+1}$, and
				\item $(v_{i+1},t_{i+1}) \in O \implies t_a \leq t_{i+1} \leq t_a+\Delta$.
		\end{enumerate}
		If it is clear from context, then we write $(t_{i},t_{i+1})$-valid.		
\end{definition}
Note that if there
exists a~$(t_i,t_{i+1})$-valid \tpath[v_i,v_{i+1}]{} $P_{i+1}$ and
$(t_{i+1},t_{i+2})$-valid \tpath[v_{i+1},v_{i+2}]{} $P_{i+2}$, then we can
``glue'' them together and get a  $(t_{i},t_{i+2})$-valid $\Delta$-restless
$(v_{i},v_{i+2})$-walk (not necessarily a path).
Thus if there exist a valid $\Delta$-restless temporal path between all
consecutive pairs in a $\Delta$-ordering which are pairwise vertex disjoint
(except for the endpoints), then there exist a \tpath[s,z].

The idea of \cref{alg:fpt-tfvs} is that a \tpath{} $P$ induces a valid partition of the timed feedback vertex set $X$
such that $(v,t) \in I$ if $P$ arrives $v$ at time $t$,
$(v,t) \in O$ if $P$ leaves $v$ at time $t$, or otherwise $(v,t) \in U$.
Furthermore, if we order $M:=I \cup O \cup (\{s,z\} \times \{\bot\})$ according to the traversal of $P$ (from $s$ to $z$), then this is a $\Delta$-ordering such that
a subpath $P'$ of $P$ corresponding to 
consecutive $(v,t),(v',t') \in M$ with $v\not=v'$ is $(t,t',I,O)$-valid in some temporal graph $\mathcal T'$ 
of Line \eqref{line:T3},
see \cref{fig:tfvs}.
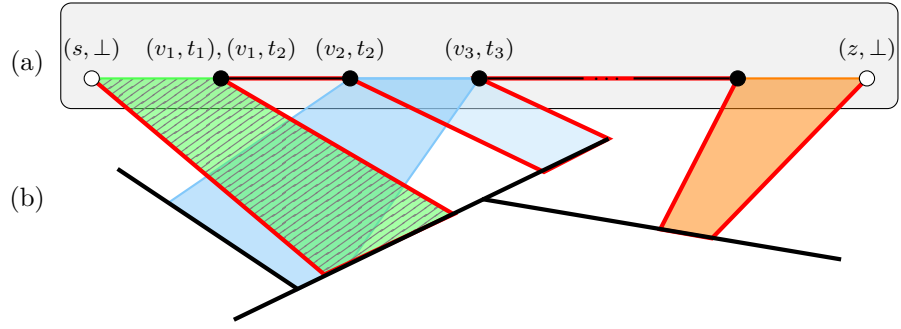
\begin{figure}[t]
\centering
	\begin{tikzpicture}[xscale=0.4,yscale=0.4,xscale=0.85]

			\draw [rounded corners,fill=gray!10] (0.8,-1) -- (0.8,2.5) -- (33.2,2.5) -- (33.2,-1) -- cycle;

			\fill [opacity=0.5,fill=ourblue!50] (12,0) -- (19.5,-3.1) -- (22,-2) --
			(17,0) -- cycle; \draw [ourblue,thick] (12,0) -- (19.5,-3.1) -- (22,-2) -- (17,0) -- cycle;

			\fill [opacity=0.5,fill=ourblue] (12,0) -- (5,-4.1) -- (10,-7) -- (12,-6.1) -- (17,0) -- cycle;
			\draw [ourblue,thick] (12,0) -- (5,-4.1) -- (10,-7) -- (12,-6.1) -- (17,0) -- cycle;

			\fill [ opacity=0.5,fill=green!70] (2,0) -- (11,-6.5) -- (16,-4.5) -- (7,0) -- cycle;
			\fill [pattern=lines,pattern color=gray] (2,0) -- (11,-6.5) -- (16,-4.5) -- (7,0) -- cycle;
			\draw [green!70,thick] (2,0) -- (11,-6.5) -- (16,-4.5) -- (7,0) -- cycle;

			\fill [opacity=0.5,fill=orange] (27,0) -- (24,-5) -- (26,-5.3) -- (32,0)   -- cycle;
			\draw [draw=orange,thick] (27,0) -- (24,-5) -- (26,-5.3) -- (32,0)   -- cycle;

			\draw [ultra thick,red] (2,0) -- (11,-6.5) -- (16,-4.5) -- (7,0) -- (12,0) -- (19.5,-3.1)
					-- (22,-2) -- (17,0) -- (27,0)
			 -- (24,-5) -- (26,-5.3) -- (32,0) 
					;
			\node[vert2,label={\footnotesize $(s,\bot)$}] (s) at (2,0) {};
			\node[vert2,black,label={\footnotesize $(v_1,t_1),(v_1,t_2)$}] (x) at (7,0) {};
			\node[vert2,black,label={\footnotesize $(v_2,t_2)$}] (v) at (12,0) {};
			\node[vert2,black,label={\footnotesize $(v_3,t_3)$}] (w) at (17,0) {};
			\node (dots) at (22,0) {\footnotesize $\dots$};
			\node[vert2,black] (y) at (27,0) {};
			\node[vert2,label={\footnotesize $(z,\bot)$}] (z) at (32,0) {};
			\draw[thick] (x) to (v);
			\draw[thick] (w) to (dots);
			\draw[thick] (y) to (dots);

			\draw[opacity=0,thick] (22,-2) to (7.5,-8);
			\draw[ultra thick] (22,-2) to (7.5,-8);
			\draw[ultra thick] (3,-3) to (10,-7);
			\draw[ultra thick] (17.1,-4) to (31,-6);

			\node at (-0.5,0.5) {(a) };%
			\node at (-0.5,-4) {(b)};%

	\end{tikzpicture}
	\caption{Illustration of \cref{alg:fpt-tfvs}, where (a) depicts the set $(\{s,z\} \times \{\bot\}) \cup I \cup O$
			and (b) sketches the underlying graph of the temporal graph $\mathcal T$
			which is a forest.
	The back solid dots correspond to one or two 
	vertex appearances. The \tpath{} is the red thick path which
	uses valid (\cref{def:valid}) $\Delta$-restless temporal  $(s,v_1)$- and $(v_1,v_2)$-paths over $\mathcal T$.}
	\label{fig:tfvs}
	\vspace{-2ex}
\end{figure}

The algorithm tries all possible partitions of $X$ and all corresponding $\Delta$-orderings.
For each of these, we store the vertices $V(P_i) \cap V(\mathcal T)$
in the family~$\mathcal P_i$, for all valid \tpath[v_{i-1},v_i]{}, where $(v_{i-1},t),(v_i,t')$ are two consecutive vertex appearances in the $\Delta$-ordering.
Here, we assume without loss of generality that no vertex
appearance of $s,z$ is in $X$.
More specifically, for each two consecutive vertex appearances
$(v_{i-1},t),(v_i,t')$ in the $\Delta$-ordering our algorithm iterates over all
pairs of time edges leading from $(v_{i-1},t)$ into the ``forest'' and
from the ``forest'' back to $(v_i,t')$. Since this fixes the entry points into
the forest in each iteration, any two $(t_i,t_{i+1})$-valid \tpath[v_i,v_{i+1}]s
present in the iteration use the same vertices of the underlying graph. Hence it
suffices to check whether one exists.
Note that, if we have $|\mathcal P_i| \geq 0$ for all $i\in \{1,\dots,x+1\}$, 
then there is a $\Delta$-restless $(s,z)$-walk in $\TG$.
Hence, to find a \tpath{}, we have to find $x+1$ pair-wise disjoint sets $P^{(1)}_1,\dots,P^{(x+1)}_{x+1}$ such that $P_i \in \mathcal P_i$.
Here, we observe that the intersection graph in Line \eqref{line:intersection-graph} is \emph{chordal} \cite{gavril1974intersection}
and use an algorithm of Bentert et al.~\cite{bentert2019indu} for \textsc{Chordal Multicolored Independent Set} as a subroutine to find such pairwise-disjoint $P^{(1)}_1,\dots,P^{(x+1)}_{x+1}$.

\begin{lemma}
		\cref{alg:fpt-tfvs} runs in $O(6^x x!x^2\cdot \max\{|\TG|^3,|V|^4\})$ time, if $x=|X| \leq |V|$.
\end{lemma}
\begin{proof}
		Let $(\TG,s,z,X,\Delta)$ be the input of \cref{alg:fpt-tfvs} and
		$x := |X|$.
		There are at most~$3^x$ many iterations of the loop in Line \eqref{line:for:partition} 
		and we can check in $O(|\TG|)$ time whether a given partition $O \uplus I \uplus U = X$ is valid.
		Since there are $O(x!)$ are many $\Delta$-orderings of $I \uplus O \uplus (\{s,z\} \times \{\bot\})$,
		the number of iterations of the loop in Line \eqref{line:for:order} is also bounded by $O(x!)$.
		Furthermore, we can check in $O(|\TG|)$ time whether a 
		given permutation $((v_i,t_i))_{i=1}^{|U \uplus I|}$ of $U \uplus I$ is a $\Delta$-ordering where $s=v_0, z=v_{|U \uplus I|+1}, t_0=t_{|U \uplus I|+1}=\bot$.
		Note that during one iteration of the loop in Line \eqref{line:for:order}
		we consider an time edge of $\TG$ at most two times as $e_1$ and two times as $e_2$
		in Line \eqref{line:for:ee}.
		Hence, we have $O(|\TG|^2)$ many iteration of the loop in Line \eqref{line:for:ee}, 
		during one iteration of the loop in Line \eqref{line:for:order}.
		Observe that \cref{lem:restless-path-on-a-path} implies that we can compute a $\Delta$-restless 
		\nonstrpaths{s,z} in linear time if the underlying graph is a forest.
Moreover, each $\Delta$-restless temporal $(v_{i-1},v_i)$-path in $\mathcal T'$ departs at time $t$ and arrives at $t'$ as $e_1$ and $e_2$ are in any temporal path from $v_{i-1}$ and $v_i$.
		Hence, Line \eqref{line:valid3} can be computed in $O(|\TG|)$ time.
		Thus, we can compute Lines \eqref{line:initP}--\eqref{line:for:ee:end} in $O(|\TG|^3)$ time.
		Observe that each set in $\mathcal P_i$ is either an empty set or contains the vertices of a path in the forest $\ug{\mathcal T}$, for all $i \in [x]$.
		Hence, the intersection graph $G$ has at most $|V|^2\cdot x$ vertices and is chordal.
Thus, Line \eqref{line:mc-indset} can be computed in $O(3^x|V|^4\cdot x^2)$ time with an algorithm of Bentert et at.~\cite[Proposition~5.6]{bentert2019indu}.
		This gives an overall running time of $O(6^x x!\cdot \max\{|\TG|^3,|V|^4x^2\})$.
\end{proof}
\begin{lemma}
		\cref{alg:fpt-tfvs} is correct.
\end{lemma}

\begin{proof}
Let \TGcompact be a temporal graph with $s,z \in V$ and let~$X$ be a timed feedback vertex set of \TG.
We assume without loss of generality~that $s$ and $z$ have no vertex appearance in $X$, that is, $s,z \not \in \{v \mid (v,t) \in X\}$. 
If this is not the case, then we can add a new vertex $\hat s$ to \TG and for each edge $\{s,v\} \in E_i$, we add $\{\hat s, s\}$ to $E_i$.   
It is clear that there exists a $\Delta$-restless \nonstrpath{s,z}~$P$ %
if and only if  there exists a $\Delta$-restless \nonstrpath{\hat s,z}~$\hat P$. %
The set $X$ remains a time feedback vertex set because $\hat s$ has degree one in the underlying graph \UG. Hence, we can now ask for a $\Delta$-restless \nonstrpath{\hat s,z} in \TG.  
The same holds for the vertex $z$ by a symmetric argument.

We show now that \cref{alg:fpt-tfvs} outputs \emph{yes} if and only
if there is a $\Delta$-restless
\nonstrpath{s,z} in $\TG$. 

\RArrow We claim that if we find a multicolored independent set in $(G,c)$, then there is a $\Delta$-restless
\nonstrpath{s,z} in $\TGcompact$. 
Let $D=\{P_1^{(1)}, \ldots, P_{x+1}^{(x+1)}\}$ be such an multicolored independent set, let $(v_0,t_0) \leq \cdots \leq (v_{x+1},t_{x+1})$ be the respective $\Delta$-ordering when
the set~$D$ was found, and let $I \uplus O\uplus X$ be the valid partition of~$X$. Hence, $P_i$ represents a $(t_{i-1}, t_i,I,O)$-valid $\Delta$-restless \nonstrpath{v_{i-1},v_i}.   
Due to $D$ being an independent set, it holds that $P^{(i)}_i \cap P^{(j)}_j = \emptyset$ for all $i \not = j \in [x+1]$. 
For all $i \in [x+1]$ it further holds that if $(v_i,t_i) \in I$, then $P_{i-1}$ arrives in $v_i$ at time $t_i$ and $P_i$ departs from $v_i$ not later than~$t$ with $t_i \leq t \leq t_i +\Delta$. If $(v_i, t_i) \in O$, then $P_{i-1}$ arrives in $v_i$ at time~$t$ with $ t \leq t_i \leq t +\Delta$ and $P_i$ departs from $v_i$ at time $t_i$. 
Hence, $P_{i-1} \cdot P_i$ is a $\Delta$-restless \nonstrpath{v_{i-1},v_{i+1}} in~\TG. 
Consequently, $P_1\cdots P_{x+1}$ is a $\Delta$-restless \nonstrpath{s,z} in \TG.

\LArrow Assume $\TG$ contains a $\Delta$-restless \nonstrpath{s,z} $P$%
, then let $I \uplus O \uplus U = X$ be the partition of $X$ that is induced by $P$. That is, for all $(v,t) \in I$ there exists a time edge $(w,v,t)$ in $P$,  
for all $(v,t)\in O$ there exists a time edge $(v,w,t)$ in $P$, and for all $(v,t) \in U$ there exist no time edge $(v,w,t)$ or $(w,v,t)$ in $P$.
The partition $I \uplus O \uplus U$ is a valid partition. 
Otherwise there exist two distinct vertex appearances $(v,t),(v,t') \in O$ such that there exist two time edges $(w,v,t), (u,v,t')$ in $P$ indicating that $P$ visits the vertex~$v$ twice. 
The same argument works for two vertex appearances of the same vertex in $I$.
Let $(v_1,t_1) \leq \cdots \leq (v_{x},t_{x})$ be the vertex appearances in the order in which they are visited by $P$. 
It holds that $t_1 \leq \ldots \leq t_x$ and for $i<j \in [x]$ if $v_i = v_j$, then there cannot exist a vertex appearance between $v_i$ and $v_j$ 
(otherwise~$P$ would visit $v_i$ twice). 
Thus $j = i+1$, $(v_i, t_i) \in I$, $(v_j,t_j) \in O$, and $t_{i} \leq t_j \leq t_{i} + \Delta$. 
It follows that $(s,\bot)=(v_0,t_0)\leq (v_1,t_1) \leq \cdots \leq (v_{x},t_{x}) \leq (v_{x+1},t_{x+1})=(z,\bot)$ is a $\Delta$-ordering of~$I \uplus O \uplus (\{s,z\} \times\{\bot\})$. 
Let $P_i$ be the subpath of $P$ starting in vertex $v_{i-1}$ and ending in $v_i$ for $i \in [x+1]$. 
If $v_{i-1}= v_i$, then it holds that $P_i$ is empty and $\mathcal P_i = \{\emptyset\}$~(Line~\eqref{line:samevertex}).
Otherwise, let $P_i = (e_i^{(1)}=(v_{i-1},v_i^{(1)},t_i^{(1)}),\ldots,e_i^{(p_i)}=(v^{(p_i)}_i, v_i, t_i^{(p_i)}))$. 
Note that if $(v_{i-1},t_{i-1}) \in O$, then $t_i^{(1)} = t_{i-1}$; 
if $(v_{i-1},t_{i-1}) \in I$, then $t_{i-1} \leq t_i^{(1)} \leq t_{i-1} + \Delta$;  
if $(v_{i},t_{i}) \in I$, then $t_i^{(p_i)} = t_{i}$; and  
if $(v_{i},t_{i}) \in O$, then $t_i^{(p_i)} \leq t_i \leq t_i^{(p_i)} + \Delta$.
Thus path $P_i$ is a $(t_{i-1},t_i,I,O)$-valid path in $\mathcal T + \{e_i^{(1)},e_i^{(p_i)}\}$ , 
and hence $V(P_i)\setminus \{v_{i-1},v_{i}\} \in \mathcal P_i$~(Line~\eqref{line:for:ee:end}). 
  Let $Q_i = V(P_i) \setminus \{v_{i-1},v_i\}$. 
  It holds that for $i\not = j \in [x+1]$ the paths $P_i$ and $P_j$ can intersect only in their endpoints because $P$ does not visit a vertex twice and thus~$Q_i \cap Q_j = \emptyset$. 
  For each $P_i$ there exists a vertex~$P^{(i)}_i$ in the intersection graph~$G$ representing with $c(P^{(i)}_i)=i$. 
  For $i,j \in [x+1]$, there exist no edge $\{P^{(i)}_i,P^{(j)}_j\}$ in $G$ because $Q_i \cap Q_j = \emptyset$. Hence, $G$ has a multicolored independent set~$D = \{P^{(1)}_1,\ldots,P^{(x+1)}_{x+1}\}$ of size $x+1$ and \cref{alg:fpt-tfvs} outputs~\emph{yes}.
\end{proof}

To conclude from \cref{thm:tfvn} the fixed-parameter tractability of \probRestlessPath parameterized the \tfvn, 
we need to compute a \tfvs efficiently.
This is clearly \NP-hard, since it generalizes the \NP-complete \textsc{Feedback Vertex Set} problem~\cite{Kar72}. 
However, we establish the following possibilities to compute a \textsc{Feedback Vertex Set}.
\begin{theorem}\label{prop:tfvs-all}
		A minimum \tfvs of temporal graph $\TG$ can be computed in $4^\tfvnvar \cdot |\TG|^{O(1)}$ time, where $\tfvnvar$ is the \tfvn of $\TG$.
Furthermore, there is a polynomial-time $8$-approximation for \tfvs.
\end{theorem}
To prove \cref{prop:tfvs-all}, we first show that a \tfvs{} of a temporal graph
can be computed via the following problem.
\probDef{Subset Feedback Vertex Set with Undeletable Vertices (SFVS-UV)}
{ Graph $G=(V,E)$, subsets $V^\infty,T \subseteq V$ and integer $k \in \N$.} 
{ Is there a $X \subseteq (V \setminus V^\infty)$ of size at most $k$ such that no simple cycle in $G-X$ 
	contains a vertex in $T$?}
Then the $8$-approximation for \tfvs follows from the $8$-approximation  of Even et al. \cite{EJZ00} 
for \textsc{Weighted Subset Feedback Set} \footnote{There is a straightforward reduction 
		from \textsc{SFVS-UV} to \textsc{Weighted Subset Feedback Set}
using infinite weights.}.
Now we see two ways in the literature to deduce a FPT-algorithm for \tfvs{}.
One is via a reduction from \textsc{SFVS-UV} to the more general problem \textsc{Group Subset Feedback Vertex Set}.
The other is through Cygan et al.~\cite{CPPWO13} who claim that
\textsc{SFVS-UV} is equivalent (under parameterized reductions for $k$) to \textsc{Subset Feedback Vertex Set}. 
The latter is \textsc{SFVS-UV} where $V^\infty=\emptyset$.
While the arguments of Cygan et al.~\cite{CPPWO13} only work 
if $V^\infty \cap T = \emptyset$, 
we here provide the missing arguments to show that the statement itself is true and hence fill a gap in the literature. 

We start with the reduction to \textsc{SFVS-UV}.
\begin{lemma}
		\label{lem:tfvs-to-sfvs}
	Given a temporal graph $\TG$ and an integer $x \in \mathbb{N}$, we can construct in $O(|\TG|+|V|\ell^2)$ time
	an instance $I = (G,V^\infty,T,x)$ of \textsc{SFVS-UV} 
	such that $I$ is a \yes-instance 
	if and only if
	$\TG$ has a \tfvs{} of size at most $x$.
\end{lemma}

\begin{constr}
		\label{constr:tfvs-to-sfvs}
Given a temporal graph $\TGcompact$ with underlying graph $\UG=(V,E)$, we construct the instance 
$I=(G = (V',E'),V^\infty,T,x)$ of \textsc{Subset Feedback Vertex Set with
Undeletable Vertices}, where $V' := \bigcup_{v \in V} V_v \cup \bigcup_{e \in E}
\cup V_e $ and $E' := \bigcup_{v \in V} E_v \cup \bigcup_{e \in E} E_e \cup
\bigcup_{t \in [\lifetime]}\bigcup_{e \in E_t} E_{(e,t)}$.
Here, 
\begin{align*}
	\forall v \in V \colon	V_v &:= \{ v_i \mid i\in[\lifetime], v \in e, e \in E_i
	\},\\
	\forall e=\{u,w\} \in E \colon	V_e &:= \{ e^{(u)},e^{(T)},e^{(w)} \}, \\
		T  	  &:= \{ e^{(T)}\mid e=\{u,w\} \in E  \}, \\
	\forall v \in V \colon	E_v &:= \{ \{v_i,v_j\} \mid v_i,v_j\in V_v, v_i \neq v_j \}, \\
	\forall e=\{u,w\} \in E \colon	E_e &:= \{ \{e^{(u)},e^{(T)}\},\{e^{(T)},e^{(w)}\} \}, \text{and}\\
	\forall t \in [\lifetime],\forall e=\{u,w\} \in E_t \colon E_{(e,t)} &:= \{
	\{e^{(u)},u_t\},\{w_t,e^{(w)}\} \mid u_t \in V_u, w_t \in V_w\}.
\end{align*}
\begin{figure}
\centering
	\begin{tikzpicture}[scale=0.3]

					\node  at (-8,-4) {$G$:};

					\node  at (-11,-10) {$\leadsto$};

			\node[fill,gray,inner sep=1pt] (dl) at (-5.5,-9) {};
			\node[fill,gray,inner sep=1pt] (dr) at (1.5,-9) {};
			\node[fill,gray,inner sep=1pt] (tr) at (1.5,-5.5) {};
			\node[fill,gray,inner sep=1pt] (tl) at (-5.5,-5.5) {};

			\filldraw[white!90!gray] (tl) rectangle (dr);
			\draw[draw=gray] (dl) -- (dr) -- (tr) -- (tl) -- (dl);

			\node[vertexset] (v1) at (-4,-4) {};
			\node[evertex] (e1a) at (-5,-5) {};
			\node[evertex] (e1b) at (-6,-6) {};
			\node[evertex] (e1c) at (-7,-7) {};
			\node[vertexset] (v2) at (-8,-8) {};
			\draw[edge] (v1) to (e1a) to (e1b) to (e1c) to (v2);

			\def\a{2}
			\def\b{-4}
			\node[vertexset,label=right:$V_u$] (v3) at (0+\a,0+\b) {};
			\node[evertex] (e2a) at (-1+\a,-1+\b) {};
			\node[evertex] (e2b) at (-2+\a,-2+\b) {};
			\node[evertex] (e2c) at (-3+\a,-3+\b) {};
			\node[vertexset] (v4) at (-4+\a,-4+\b) {};
			\draw[edge] (v3) to (e2a) to (e2b) to (e2c) to (v4);

			\node[vertexset] (v5) at (4,-8) {};

			\node[evertex] (e3a) at (-6.5,-8) {};
			\node[evertex] (e3b) at (-5,-8) {};
			\node[evertex] (e3c) at (-3.5,-8) {};
			\draw[edge] (v2) to (e3a) to (e3b) to (e3c) to (v4);

			\node[evertex] (e4a) at (-0.5,-8) {};
			\node[evertex] (e4b) at (1,-8) {};
			\node[evertex] (e4c) at (2.5,-8) {};
			\draw[edge] (v4) to (e4a) to (e4b) to (e4c) to (v5);

			\node[evertex] (e5a) at (-2.5,-4) {};
			\node[evertex] (e5b) at (-1,-4) {};
			\node[evertex] (e5c) at (.5,-4) {};
			\draw[edge] (v1) to (e5a) to (e5b) to (e5c) to (v3);

			\node[evertex2,fill=red] (e1b) at (-6,-6) {};
			\node[evertex2,fill=red] (e2b) at (-2+\a,-2+\b) {};
			\node[evertex2,fill=red] (e3b) at (-5,-8) {};
			\node[evertex2,fill=red] (e4b) at (1,-8) {};
			\node[evertex2,fill=red] (e5b) at (-1,-4) {};

			\node[evertex2,fill=red] (e1b) at (-6,-6) {};
			\node[evertex2,fill=red] (e2b) at (-2+\a,-2+\b) {};
			\node[evertex2,fill=red] (e3b) at (-5,-8) {};
			\node[evertex2,fill=red] (e4b) at (1,-8) {};
			\node[evertex2,fill=red] (e5b) at (-1,-4) {};

			\draw[opacity=0,-,dashed] (-9,-10) to (4,-10);

			\begin{scope}[xshift=8cm,yshift=0.5cm]
					\node[fill,gray,inner sep=1pt] (zdl) at (-8.5,-16.5) {};
					\node[fill,gray,inner sep=1pt] (zdr) at (4.5,-16.5) {};
					\node[fill,gray,inner sep=1pt] (ztr) at (4.5,-10.5) {};
					\node[fill,gray,inner sep=1pt] (ztl) at (-8.5,-10.5) {};

			\draw[thick,draw=gray] (zdl) -- (zdr) -- (ztr) -- (ztl) -- (zdl);
			\draw[draw=gray,dashed] (dl) -- (zdl);
			\draw[draw=gray,dashed] (tl) -- (ztl);
			\draw[draw=gray,dashed] (tr) -- (ztr);
			\draw[draw=gray,dashed] (dr) -- (zdr);
			\filldraw[white!90!gray] (ztl) rectangle (zdr);

			\node[evertex] (ez) at (-1,-12) {};
			\node[evertex2,fill=red] (ezr) at (0,-11) {};
			\node[evertex] (ex) at (-6,-14) {};
			\node[evertex2,fill=red] (exr) at (-8,-14) {};
			\node[evertex] (ey) at (2,-14) {};
			\node[evertex2,fill=red] (eyr) at (4,-14) {};
			\node[vert2,label={[label distance=-3pt]160:$v_1$}] (vt1) at (-2,-13) {};
			\node[vert2,label=below:$v_3$] (vt2) at (-2,-15) {};
			\node[vert2,label={[label distance=-3pt]230:$v_2$}] (vt3) at (-3,-14) {};
			\node[vert2,label={[label distance=-3pt]275:$v_4$}] (vt4) at (-1,-14) {};

			\draw[edge] (vt1) to (vt2);
			\draw[edge] (vt1) to (vt3);
			\draw[edge] (vt1) to (vt4);
			\draw[edge] (vt2) to (vt3);
			\draw[edge] (vt4) to (vt2);
			\draw[edge] (vt4) to (vt3);
			
			\draw[edge] (ex) to (vt1);
			\draw[edge] (ex) to (vt3);
			\draw[edge] (ey) to (vt1);
			\draw[edge] (ey) to[bend left=30] (vt2);
			\draw[edge] (ey) to (vt4);
			
			\draw[edge] (ez) to (vt1);
			\draw[edge] (ez) to (vt1);
			\draw[edge] (ez) to (vt1);

			\draw[edge] (ez) to (ezr);
			\draw[edge] (ex) to (exr);
			\draw[edge] (ey) to (eyr);
			\end{scope}

			\begin{scope}[scale=0.8,xshift=-23cm,yshift=-6cm]
					\node  at (-8,-4) {$\TG$:};
					\node[vert2] (w1) at (-4,-4) {};
				\node[vert2] (w2) at (-8,-8) {};
				\node[vert2,label=right:$u$] (w3) at (0+\a,0+\b) {};
				\node[vert2,label=below:$v$] (w4) at (-4+\a,-4+\b) {};
				\node[vert2] (w5) at (4,-8) {};
				\draw[edge] (w1) to node[fill=white] {$t$} (w2) to node[fill=white] {$1,2$} (w4) to node[fill=white] {$1$} (w3) to node[fill=white] {$t'$} (w1);
				\draw[edge] (w4) to node[fill=white] {$1,3,4$} (w5);
			\end{scope}
	\end{tikzpicture}
	\caption{
		An illustration of \cref{constr:tfvs-to-sfvs} for a temporal graph~$\TG$ (left) to graph $G$ (right).
		The set $V_u$ in $G$ of a vertex~$u$ in~$\TG$ is depicted by a large circle.
		The vertices in $V_e$ of an edge~$e$ in the underlying graph of~$\TG$ are filled.
		The vertices in~$T$ are squared (red).
		}
\label{fig:constr-to-wsfvs}

\end{figure}
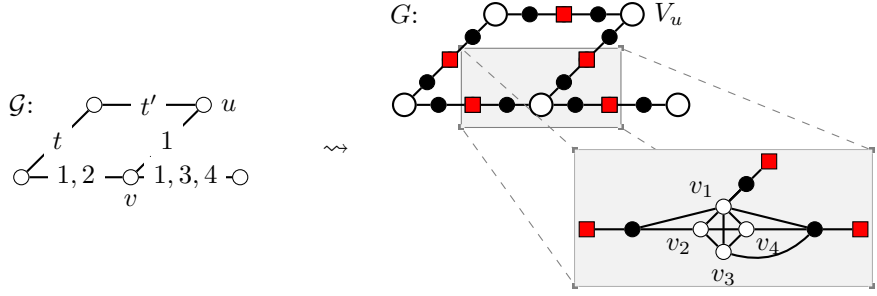
Finally we set $V^\infty := \bigcup_{e \in E} \cup V_e$.
Consider \cref{fig:constr-to-wsfvs} for an example.
\hfill$\blacklozenge$
\end{constr}

\begin{proof}[Proof of \cref{lem:tfvs-to-sfvs}]
Let $\TGfull$ be a temporal graph, $x \in \N$, and $I = (G,V^\infty,T,x)$ be
the resulting instance from \cref{constr:tfvs-to-sfvs}.
It is easy to check that \cref{constr:tfvs-to-sfvs} can be computed in $O(|\TG|+|V|\ell^2)$ time.

We now claim that there is a timed feedback vertex set
$X$ of size at most $x$ in $\TG$ if and only if 
there is a subset feedback vertex $X$ of size at most $x$ in $G$ such that $X \cap V^\infty = \emptyset$.

\RArrow
Let $X$ be a timed feed back vertex for $\TG$.
Then, set $Y := \{ v_t \in V_v \mid (v,t) \in X \}$.
Hence, $|Y| \leq x$ and $Y \cap V^\infty = \emptyset$.
We claim that $Y$ is a subset feedback vertex set for $I$.
Assume towards a contradiction that there is a simple cycle $C$ in $G - Y$ which contains a vertex of $T$.
Furthermore, we assume without loss of generality that there is no shorter cycle in $G-Y$
which contains a vertex of $T$.
Observe that this implies that $C$ does not visit three distinct vertices $ v_a,v_b,v_c \in V_v$, for any $v \in V$,
because otherwise there is a shorter cycle using one of the edges $\{v_a,v_b\}$, $\{v_a,v_c\}$ or $\{v_c,v_b\}$ in $E_v$.
Moreover if $C$ visits two distinct vertices $v_a,v_b \in V_v$, then $\{v_a,v_b\} \in E_v$ is part of $C$, for all $v\in V$,
because otherwise there is a shorter cycle using the edge $\{v_a,v_b\}$.
Furthermore, for all edge $e \in E$ we have that $ V_e \subseteq V(C)$ or $V_e \cap V(C) = \emptyset$, 
because $G[V_e]$ induces a $P_3$ and hence using only an endpoint of that $P_3$ would imply that $C$ visits two vertices $v_a,v_b \in V_v$ without the edge $\{v_a,v_b\}$, for some $v \in V$.
Since $T$ only contains the middle vertex $e^{(T)} \in E_e$ of the $P_3$ induced by $G[V_e]$,
we can observe that $C$ is a subdivision of a cycle in $\UG(\TG-X)$ which contradicts that $X$ is a timed feedback vertex set for $\TG$.

\LArrow
Let $Y \subseteq (V'\setminus V^\infty)$ be of size at most $x$ 
such that no simple cycle in $G'-X$ which contains a vertex in $T$.
We set $X := \{ (v,t) \mid v_t \in \bigcup_{u\in V} V_u \cap Y \}$.
Hence, $X$ is of size at most $x$.
We claim that $X$ is a timed feedback vertex set for $\TG$.
Assume towards a contradiction that this is not the case and there is a cycle $C$ in $\UG(\TG-X)$.
We now build a cycle in $G-Y$ containing a vertex from $T$.
Note that for each edge $e$ used in $C$ none of the vertices in $V_e$ are in $Y$, 
otherwise $V^\infty \cap Y \not= \emptyset$.
Hence, set $V_C := \bigcup_{e \in E(C)} E_e$, where $E(C)$ is the edge set of $C$.
Since any two incident edges $e_1,e_2 \in E(C)$ are in the underlying graph of
$\TG-X$, we know that there are $t_1,t_2 \in [\lifetime]$ such that $(e_1,t_1)$
and $(e_2,t_2)$ are time edges of $\TG-X$.
Hence, for two incident edges $e_1,e_2 \in E(C)$ with $\{v \}= e_1 \cap e_2$ we
pick $t_1,t_2 \in [\lifetime]$ such that $(e_1,t_1)$ and $(e_2,t_2)$ are time
edges of $\TG - X$ add $v_{t_1},v_{t_2} \in V_v$ to $V_C$.
Observe that $G[V_C]$ contains a cycle and that $V_C \cap T \not = \emptyset$.
Since we constructed $V_C$ from a cycle in the underlying graph of $\TG - X$ we have $V_C \cap Y = \emptyset$.
This is a contradiction.
\end{proof}

Now we can use the polynomial-time $8$-approximation of Even et al. \cite{EJZ00} for \textsc{Weighted Subset Feedback Vertex Set} 
and \cref{lem:tfvs-to-sfvs} to conclude the following.\footnote{Here, vertices get weight $\infty$ if there are undeletable, and one otherwise.}

\begin{corollary}\label{prop:tfvsapprox}
		There is a polynomial-time $8$-approximation for \tfvs.
\end{corollary}

In the remainder of this section we prove the following. 
\begin{lemma}
		\label{lem:sfvs-uv}
		Given an instance $I=(G,V^\infty,T,k)$ of \textsc{Subset Feedback Vertex Set with Undeletable Vertices} 
		we can construct in $O(k^2(|V|+|E|))$ time an instance $I'=(G,T',k')$ of
		\textsc{Subset Feedback Vertex Set} with $k'\le k$ such that $I$ is a
		\yes-instance if and only if $I'$ is a \yes-instance.
\end{lemma}
Note that the running time of the algorithm behind \cref{lem:sfvs-uv} depends only linearly on the size of the graph.
The proof of \cref{lem:sfvs-uv} is deferred to the end of this section.
First, we introduce two data reductions rules and then perform the reduction behind \cref{lem:sfvs-uv} in two steps.
We use these data reduction rules to get an equivalent instance where $G[T \cap V^\infty]$ is an independent set.
We start by detecting some \no-instances.
\begin{rrule}
		\label{rr:undeletable-cycle}
		Let $I=(G,V^\infty,T,k)$ be an instance of \textsc{Subset Feedback Vertex Set with Undeletable Vertices}
		such that there is a vertex $u$ and a simple cycle $C$ intersecting $T$ where $V(C) \setminus (T \cap V^\infty) \subseteq \{u\} \subseteq V^\infty$.
		Then output a trivial \no-instance.
\end{rrule}
We now show that we can detect in linear-time whether \cref{rr:undeletable-cycle} is applicable and that it is \emph{safe}.
The latter means that the application of \cref{rr:undeletable-cycle} does not turn a \yes-instance into a \no-instance or vice versa.
\begin{lemma}
		\label{lem:undeletable-cycle}
		\cref{rr:undeletable-cycle} is safe and can be applied in linear time.
\end{lemma}
\begin{proof}
Since $C$ is a witness that $I$ is a \no-instance, \cref{rr:undeletable-cycle} is safe.
We check whether there is cycle $C$ such that  $ V(C) \setminus (T \cap V^\infty) = \emptyset$
by simply checking whether $G[T \cap V^\infty]$ is a forest.
Assume that $G[T \cap V^\infty]$ is a forest, otherwise we are done and output a trivial \no-instance.
First, we partition $V(G[T \cap V^\infty]) = Q_1 \uplus \dots \uplus Q_c$ such that $Q_i$ is a maximal connected component of $G[T \cap V^\infty]$. Clearly, this can be done in linear time.
For each connected component $Q_i$ of $G[T \cap V^\infty]$
we first unmark all vertices in $V\setminus (T \cap V^\infty)$. 
Then we iterate over all vertices $v \in Q_i$ and mark all vertices in $N_G(v) \cap (V^\infty \setminus T)$.
If we find a vertex $w \in Q_i$ such that there is a vertex $u \in N_G(w) \cap (V^\infty \setminus T)$ which is already marked,
then the path from $v$ to $w$ in $Q_i$ together with $u$ is a cycle $C$ where $V(C) \setminus (T \cap V^\infty) \subseteq \{u\} \subseteq V^\infty$. Hence, we output a trivial \no-instance in this case.
Moreover, if there is some simple cycle $C$ where $V(C) \setminus (T \cap V^\infty) \subseteq \{u\} \subseteq V^\infty$,
then all vertices $V(C) \setminus \{u\}$ belong to the same connected component of $G[T \cap V^\infty]$.
Thus, the above described procedure will find $C$.
A simple application of the Handshaking Lemma shows that this procedure ends after linear time.
\end{proof}

The purpose of the next data reduction rule is to merge undeletable terminal vertices which do share an edge.
\begin{rrule}
		\label{rr:merge}
		Let $I=(G,V^\infty,T,k)$ be an instance of \textsc{Subset Feedback Vertex Set with Undeletable Vertices}
		such that there is $\{v,w\} \in E$ with $\{v,w\} \subseteq T \cap V^\infty$ and $N_G(v) \cap N_G(w) \cap V^\infty = \emptyset$.
		Then set $\widehat G = (V(G - w),E(G-w) \cup \{ \{v,u\} \mid \{w,u\} \in E \})$
		and $G' :=\widehat G - (N_G(v) \cap N_G(w))$.
		Output $I'=(G',V^\infty \cap V(G'), T\cap V(G'),k - |N_G(v)\cap N_G(w)|)$.
\end{rrule}
Note that for each edge in $G[T \cap V^\infty]$, either \cref{rr:undeletable-cycle} or \cref{rr:merge} is applicable.
While it is easy to apply \cref{rr:merge} exhaustively in polynomial time, 
more effort is required to do the same in linear time.
To this end, we will construct in linear time an equivalent instance such that \cref{rr:undeletable-cycle,rr:merge} are not applicable.

\begin{lemma}
		\label{lem:rr:merge}
Given an instance $I$ of \textsc{Subset Feedback Vertex Set with Undeletable Vertices}, we can compute an equivalent instance $I'$ of \textsc{Subset Feedback Vertex Set with Undeletable Vertices} in linear time such that \cref{rr:undeletable-cycle,rr:merge} are not applicable to $I'$. 
\end{lemma}
\begin{proof}
		We first check in linear time by \cref{lem:undeletable-cycle} whether \cref{rr:undeletable-cycle} is applicable. 
		Assume that \cref{rr:undeletable-cycle} is not applicable, otherwise we are done.
		Hence, $G[T \cap V^\infty]$ is a forest.
		We now aim to apply \cref{rr:merge} for all edges in $G[T \cap V^\infty]$ at once.
		First, we partition $V(G[T \cap V^\infty]) = Q_1 \uplus \dots \uplus Q_c$ 
		such that $Q_i$ is a maximal connected component of $G[T \cap V^\infty]$. 
		Then we replace $Q_i$ with a fresh vertex $q_i$.
		To this end, we construct the graph~$G' := (V', E')$ where
		\begin{align*}
				V' := &(V \setminus (T \cap V^\infty)) \cup \{ q_i \mid i \in [c] \} \text{ and}\\
				E' := &\{ \{a,b\} \in E \mid \{a,b\} \subseteq V' \} \cup \{ \{q_i,w\} \mid w \in V' \text{ and } \exists v \in Q_i \colon \{v,w\} \in E \}.
		\end{align*}
		Note that $G'$ and hence $V'^\infty = (V^\infty \cap V') \cup \{ q_i \mid i \in [c]\}$ 
		and $T' = (T \cap V') \cup \{ q_i \mid i \in [c] \}$ can be computed in linear time.
		To compute the remaining budget~$k'$, we set $K = \emptyset$.
		Then, for each connected component $Q_i$ of $G[T \cap V^\infty]$, 
		we first unmark all vertices. 
		Second, we iterate over all vertices $v \in Q_i$ and mark all vertices in $N_G(v) \setminus (T \cap V^\infty)$.
		If we find a vertex $w \in Q_i$ such that a vertex $u \in N_G(w) \setminus (T \cap V^\infty)$ which is already marked,
		then we add $u$ to~$K$, because $G[Q_i \cup \{ u\}]$ contains a cycle intersecting $T$ where $u$ is the only deletable vertex.
		Recall that $u \not\in V^\infty$, since \cref{rr:undeletable-cycle} was not applicable.
		A simple application of the Handshaking Lemma show that this procedure ends after linear time.
		If $k < |K|$, then we return a trivial \no-instance, 
		because for each vertex $v \in K$ there is a simple cycle 
		intersecting $T$ where $v$ is the unique vertex not in $V^\infty \cap T$.
		Otherwise, we output $I' = (G' - K,V'^\infty,T' \setminus K,k' = k - |K|)$.
		It is easy to verify that \cref{rr:merge} is not applicable in $I'$.
		We now claim that $I$ is a \yes-instance
		if and only if
		$I'$ is a \yes-instance.

		\RArrow
		Assume that $X$ is a solution for $I$.
		Note that for each vertex $v \in K$, there is a simple cycle 
		intersecting $T$ where $v$ is the unique vertex not in $V^\infty \cap T$.
		Hence, $K \subseteq X$.
		Set $X' := X \setminus K$ and observe that $X' \subseteq V(G')$. %
		The set $X'$ is a solution for $I'$, because it is of size at most $k'$ 
		and for each cycle in $G'-K$ which intersects to $T' \setminus K$,
		we can construct a cycle in $G$ which intersects to $T$ 
		by replacing a vertex $q_j \in \{ q_i \mid i \in [c] \}$ with a path in $G[Q_j]$.

		\LArrow
		Assume that $X'$ is a solution for $I'$.
		We set $X := X' \cup K$ and note that $X$ is of size at most $k$.
		We may assume towards a contradiction that $X$ is not a solution for $I$.
		Hence there is a simple cycle $C$ in $G - X$ which contains a vertex of $T$.
		We now construct a closed walk of $C'$ in $G'- T$ from $C$ by replacing a 
		subpath on vertices in $Q_i$ with $q_i$, for all $i \in [c]$.
		Since we only replaced vertices from $T \cap V^\infty$ of $I$ with vertices in $T' \cap V'^\infty$,
		the closed walk $C'$ contains a simple cycle in $G'-X'$ containing a vertex from $T'$---a contradiction.
\end{proof}

We now show an algorithm to dispose all vertices undeletable vertices in $T$
such that the running time dependence only linearly on the size of the graph.
\begin{lemma}
		\label{lem:remove-non-terminals}
		Let $I=(G,V^\infty,T,k)$ be an instance of \textsc{Subset Feedback Vertex Set with Undeletable Vertices}.
		We can construct in $O(k(|V|+|E|))$ time 
		an instance $I'=(G',V'^\infty,T',k')$ of \textsc{Subset Feedback Vertex Set with Undeletable Vertices}
		such that 
		\begin{enumerate}
			\item $k'\leq k$, 
			\item $V'^\infty \setminus T' = \emptyset$,
			\item $I$ is a \yes-instance if and only if $I'$ is a \yes-instance.
		\end{enumerate}
\end{lemma}
\begin{proof}
		Let $I=(G,V^\infty,T,k)$ be an instance of \textsc{Subset Feedback Vertex Set with Undeletable Vertices}.
		Since we aim for a running time of $O(k(|V|+|E|))$, 
		we can assume that \cref{rr:undeletable-cycle,rr:merge} are not applicable on $I$, see \cref{lem:rr:merge}.

		The goal now is to duplicate each vertex in $V^\infty \setminus T$ $k+1$ times, 
		such that we cannot delete all of them even if there are not in $V^\infty$.
		Note that we might create many copies of an edge doing this naïvely.
		To avoid this, we observe that it suffices to replace 
		each maximal connected component in $G[V^\infty \setminus T]$ with $k+1$ vertices.
		We partition $V(G[V^\infty \setminus T]) = Q_1 \uplus \dots \uplus Q_c$ 
		such that $Q_i$ is a maximal connected component of $G[V^\infty \setminus T]$
		and we construct~$G' := (V',E')$, where
		\begin{align*}
				V' &:= (V \setminus (V^\infty\setminus T)) \cup \{ q_i^j \mid i \in [c],j \in [k+1] \} \text{ and }\\
				E' &:= E(G[V \setminus (V^\infty\setminus T)]) \cup \{ \{q_i^j,w\} \mid \exists v \in Q_i \colon \{v,w\}\in E\}.
		\end{align*}
		Since $|\{ q_i^j \mid i \in [c],j \in [k+1] \}| \leq |V|(k+1)$ and 
		we copy each edge at most $k+1$ times, $G'$ is constructed after $O(k(|V|+|E|))$ time.

		We now claim that $I$ is a \yes-instance 
		if and only if $I' := (G',T' := V^\infty \cap V(G'),T,k)$ is a \yes-instance 
		of \textsc{Subset Feedback Vertex Set with Undeletable Vertices}.
		
		\RArrow
		Let $X$ be a solution for $I$.
		Then $X$ is also a solution for $I'$, 
		because $X \cap V^\infty = \emptyset$ and for each cycle $C'$ in $G'$ containing a vertex from $T$, 
		we can construct a closed walk in $G$ by replacing a vertex $q_i^j$ by a path in $G[Q_i]$.
		This closed walk induces a simple cycle $C$ in $G$ containing a vertex in $T$. 
		Hence, $C'$ also contains a vertex from $X$.

		\LArrow
		Let $X$ be a solution for $I'$.
		We assume without loss of generality that $X \cap \{ q_i^j \mid i \in [c],j \in [k+1] \} = \emptyset$.
		Suppose towards a contradiction that $X$ is not a solution for $I$.
		Let $C$ be an cycle in $G$ satisfying $V(C) \cap X = \emptyset \not = V(C) \cap T$.
		We construct a closed walk $C'$ in $G'$ from $C$ by replacing each maximal consecutive subpath in $C$
		containing only vertices from $Q_i$ with $q_i^1$, for all $i \in [c]$.
		Note that $C$ contains a simple cycle which intersects $T'$ but not $X$---a contradiction.
\end{proof}

Now we are finally ready to prove \cref{lem:sfvs-uv}.
\begin{proof}[Proof of \cref{lem:sfvs-uv}]
		First, we apply \cref{lem:remove-non-terminals} on $I$ and 
		hence assume that $V^\infty \setminus T = \emptyset$.
		Furthermore, by \cref{lem:rr:merge} we assume that \cref{rr:merge,rr:undeletable-cycle} are not applicable.
		Thus, $G[V^\infty \cap T]$ is an independent set.
		We now create $k+1$ copies of each vertex in $V^\infty \cap T$ such that 
		we cannot remove all of them, even if there are deletable.
		However, we have to be careful what kind of new cycles this creates.
		We do the following.
		Let $V^\infty \cap T = \{v_1,v_2,\dots,v_p\}$ and take a fresh set of 
		vertices $Q_i :=\{ q_i^j \mid j \in [k+1]\}$ for each $i \in [p]$.
		We construct $G' := (V',E')$, where
		\begin{align*}
				V' := &(V \setminus (V^\infty \cap T)) \cup \bigcup_{i=1}^p Q_i  \cup \{ w_i',w_i \mid \{v_i,w\} \in E \}, \text{ and}\\
		E' := &\{ \{v,w\} \in E \mid v,w \in V'\} \cup \\
				&\{ \{w_i,w\},\{w_i,w_i'\}, \{w_i',q_i^j\}  \mid \{v_i,w\} \in E, i \in [p]\} , j\in[k+1] \}.
		\end{align*}
		Output $I'=(G',T',k)$, 
		where $T' := (T \setminus V^\infty) \cup \{ w_i \mid \{v_i,w\} \in E \}$.
		Since $|\{ q_i^j \mid i \in [c],j \in [k+1] \}| \leq |V|(k+1)$ and
		we create for each edge at most $k+2$ new edges, $G'$ is constructed after $O(k(|V|+|E|))$ time.
		Together with the preprocessing of \cref{lem:remove-non-terminals} this gives an overall running time of $O(k^2(|V|+|E|))$.

		We now claim that $I$ is a \yes-instance 
		if and only if $I'$ is a \yes-instance of \textsc{Subset Feedback Vertex Set}.

		\RArrow
		Let $X$ be a solution for $I$.
		Observe that $X \subseteq V'$, because $V \setminus V' \subseteq V^\infty$.
		Assume towards a contradiction that there is a cycle $C'$ in $G'$ such that $V(C') \cap X = \emptyset$.
		Assume without loss of generality that $C'$ is a shortest of the set of cycles satisfying $V(C') \cap X = \emptyset$.
		Hence, for each $i \in [p]$ we have that  $|V(C') \cap Q_i| \leq 1$,
		and since  $G[V^\infty \cap T]$ is an independent set, 
		a vertex $q_i^j \in  V(C') \cap Q_i$ has two neighbors $w'_i,u'_i$, where $w,u \in V$.
		Thus, we can construct a cycle $C$ in $G$ by replacing each subpath $w'_i,q_i^j,u'_i$ 
		with the vertex $v_i \in T \cap V^\infty$, where $q_i^j \in  V(C') \cap Q_i$ and $i\in[p],j\in[k+1]$.
		This contradicts $X$ being a solution for~$I$.

		\LArrow
		Let $X'$ be a solution for $I'$.
		For all $i \in [p]$, we assume without loss of generality that
		\begin{description}
				\item[$X' \cap Q_i = \emptyset$:] 
						This can be assumed, because $|Q_i| > k \geq |X'|$ and all vertices in $Q_i$ have the same neighborhood.
				\item[$X' \cap \{ w_i \mid \{v_i,w\} \in E \} = \emptyset$:]
						This can be assumed, because these vertices are of degree two and thus a vertex $w_i$ can be replaced by its origin $w$.
				\item[$X' \cap \{ w_i' \mid \{v_i,w\} \in E \} = \emptyset$:]
						Such a vertex $w_i'$ can be replaced by its origin $w$ as well,
						because for each cycle $C'$ in $G'$ passing through $T'$ which does not include $w$,
						we know that $v_i'q_i^jw_i'q_i^{j'}u_i'$ is a subpath of $C'$ for some $v,u\in V$, $q_i^j,q_i^{j'} \in Q_i$, and $j,j' \in [k+1]$.
						Hence $C'-\{q_i^j,w_i\}$ is also a cycle in $G'$ that contains a vertex from $T'$.
						Thus, there is $V(C')\cap (X\setminus \{w_i'\}) \not= \emptyset$.
		\end{description}
		Hence $X' \subseteq V \setminus V^\infty$.
		Now assume towards a contradiction that there is a cycle $C$ in $G$ which does not contain a vertex in $X'$.
		Hence there $v_i \in V(C) \cap (T \cap V^\infty)$, otherwise $C$ is also a cycle in $G'$.
		We construct a cycle $C'$ in $G'$ from $C$ by replacing each subpath $u,v_i,w$ in $C$ with
		$u,u_i,u'_i,v_i,w'_i,w_i,w$, for all $v_i \in V(C)\cap (T \cap V^\infty)$.
		This contradicts $X'$ being a solution because $u_i \in T'$.
\end{proof}
By using \cref{lem:tfvs-to-sfvs,lem:sfvs-uv} we now can compute a minimum \tfvs by any known \textsc{Subset Feedback Vertex Set} algorithms.
Thus, \cref{lem:tfvs-to-sfvs,lem:sfvs-uv,prop:tfvsapprox} together with the
algorithm of Iwata et al.~\cite{IWY16} imply \cref{prop:tfvs-all}.

\section{Conclusion}\label{sec:path:conclusion}
We have analyzed the (parameterized) computational complexity of
\probRestlessPath, a canonical variant of the problem of finding temporal
paths, where the waiting time at every vertex is restricted. 
Unlike its non-restless counterpart or the ``walk-version'', this problem turns out to be computationally hard, 
even in quite restricted cases. 
On the positive side, we give an efficient algorithm to find short restless temporal paths 
and we could identify structural parameters of the underlying graph and of the temporal graph itself 
that allow for fixed-parameter algorithms.

\subsection*{Acknowledgements.}
We thank the referees for their careful reading and constructive comments which
significantly improved the presentation of these results.

\bibliographystyle{plainurl}%

\bibliography{literature}

\end{document}